\documentclass[draft]{article}
\usepackage{amsmath,amsfonts,latexsym, amssymb,   mathrsfs, eucal}

\usepackage{color}

\newcommand{\fa}{{\frak a}} 

\newcommand{\ttau}{\vartheta}

\newcommand{\fn}{{\mathfrak n}}

\newcommand{\wt}{\widetilde}

\newcommand{\beqa}{\begin{eqnarray}}
\newcommand{\eeqa}{\end{eqnarray}}
\newcommand{\e}{\varepsilon}
\newcommand{\eps}{\varepsilon}

\newcommand{\rd}{{\rm d}}
\newcommand{\bR}{{\mathbb R}}
\newcommand{\bC}{{\mathbb C}}

\newcommand{\bZ}{{\mathbb Z}}
\newcommand{\non}{\nonumber}
\newcommand{\wH}{{K}}

\newcommand{\abs}[1]{\left| #1 \right|}

\newcommand{\tr}{\mbox{Tr\,}}

\newcommand{\unu}{{\underline {\nu}}}
\newcommand{\umu}{{\underline {\mu}}}

\renewcommand{\Im}{\mathfrak Im}

\newcommand{\ba}{{\bf{a}}}

\newcommand{\bx}{{\bf{x}}}

\newcommand{\bu}{{\bf{u}}}
\newcommand{\bv}{{\bf{v}}}
\newcommand{\bw}{{\bf{w}}}

\newcommand{\bq}{{\bf {q}}}

\newcommand{\bh}{{\bf{h}}}
\newcommand{\bn}{{\bf{n}}}

\newcommand{\fq}{{\frak q}}

\newcommand{\sa}{{[\alpha ]}}

\newcommand{\bT}{{\T}}

\newcommand{\bnu}{\mbox{\boldmath $\nu$}}

\newcommand{\mg}{{m_N}}

\newcommand{\al}{\alpha}

\newcommand{\be}{\begin{equation}}
\newcommand{\ee}{\end{equation}}

\newcommand{\barv}{ {[ v ]}}
\newcommand{\barZ}{ {[ Z ]}}

\newcommand{\la}{\lambda}

\newcommand{\cL}{{\mathscr L}}

\newcommand{\cG}{{\mathcal G}}

\newcommand{\cM}{{\mathcal M}}
\newcommand{\cN}{{\mathcal N}}
\newcommand{\cH}{{\mathcal H}}

\newcommand{\ov}{\overline}

\newcommand{\re}{{\mathfrak{Re} \, }}
\newcommand{\im}{{\mathfrak{Im} \, }}
\newcommand{\E}{{\mathbb E }}
\newcommand{\R}{{\mathbb R }}
\newcommand{\N}{{\mathbb N}}

\newcommand{\Ci}{{ C_{inf}}}
\newcommand{\Cs}{{ C_{sup}}}

\renewcommand{\P}{{\mathbb P}}

\newcommand{\C}{{\mathbb C}}

\newcommand{\T}{\mathbb T}

\newcommand{\bU}{ {\bf  U}}
\newcommand{\bS}{ {\bf  S}}

\newtheorem{theorem}{Theorem}
\newtheorem{corollary}[theorem]{Corollary}
\newtheorem{lemma}[theorem]{Lemma}
\newtheorem{proposition}[theorem]{Proposition}

\newtheorem{definition}{Definition}
\newcommand{\qed}{\hfill\fbox{}\par\vspace{0.3mm}}
\newenvironment{proof}{{\bf Proof.}} {\hfill\qed}

\numberwithin{equation}{section}
\numberwithin{theorem}{section}
\numberwithin{definition}{section}
\numberwithin{remark}{section}

{\bf \title{Rigidity of Eigenvalues  of Generalized  Wigner Matrices}}

\author{
L\'aszl\'o Erd\H os${}^1$\thanks{Partially supported
by SFB-TR 12 Grant of the German Research Council}, 
Horng-Tzer Yau${}^2$\thanks{Partially supported
by NSF grants DMS-0757425, 0804279}  \; and Jun Yin${}^2$\thanks{Partially supported by NSF grants DMS-1001655} 
 \\\\
Institute of Mathematics, University of Munich, \\
Theresienstr. 39, D-80333 Munich, Germany \\ lerdos@math.lmu.de ${}^1$ \\ \\
Department of Mathematics, Harvard University\\
Cambridge MA 02138, USA \\  htyau@math.harvard.edu,  jyin@math.harvard.edu ${}^2$ \\ 
\\}

\begin{document}
\date{Oct 25, 2011}

\maketitle

\begin{abstract}

Consider  $N\times N$ hermitian or symmetric random matrices $H$
with independent entries, 
where the  distribution  of the $(i,j)$ matrix element is given by
the probability  measure $\nu_{ij}$  with zero expectation
and with variance $\sigma_{ij}^2$.  
 We assume that the variances satisfy the normalization
condition  $\sum_{i} \sigma^2_{ij} = 1$ for all $j$ and
that  there is a  positive constant $c$ such that 
$c\le N \sigma_{ij}^2 \le c^{-1}$. We further assume that
the probability distributions $\nu_{ij}$ have a uniform subexponential decay.
We prove that the Stieltjes transform of the empirical 
eigenvalue distribution of $H$ 
is given by the Wigner semicircle law uniformly up to the edges of the spectrum
 with  an error of order $ (N \eta)^{-1}$ 
where $\eta$ is  the imaginary part of the spectral parameter 
in the Stieltjes transform. There are three  corollaries to this strong local semicircle law: 
(1) Rigidity of eigenvalues:  If  $\gamma_j =\gamma_{j,N}$ denotes 
 the {\it classical location} of the $j$-th eigenvalue 
under the semicircle law ordered in increasing order, then 
 the $j$-th eigenvalue $\lambda_j$ is 
close to  $\gamma_j$ in 
the sense that for some positive 
constants $C, c$
\[
\P \Big ( \exists \,  j : \; |\lambda_j-\gamma_j| 
\ge (\log N)^{C\log\log N}
 \Big [ \min \big ( \, j ,  N-j+1 \,  \big) \Big  ]^{-1/3}
    N^{-2/3} \Big ) \le C\exp{\big[-(\log N)^{c\log\log N} \big]}
\]
for $N$ large enough.  (2)  The proof of 
{\it Dyson's conjecture}  \cite{Dy} which states that 
the time scale of the Dyson Brownian motion to  reach local 
equilibrium is of order $N^{-1}$  up to logarithmic corrections.
(3)  The edge universality holds in the sense that the probability
 distributions of the largest (and the smallest) 
eigenvalues of two generalized Wigner  ensembles are the same in the
 large $N$ limit  provided that the second moments of the two ensembles are identical.

\end{abstract}


\medskip

\medskip


\medskip

\newpage 

\section{Introduction}

Random matrices were introduced by E.  Wigner to model the excitation spectrum of  large nuclei. 
The central idea is based on the observation that
the eigenvalue gap distribution   for a large complicated system is universal
 in the sense that it depends only on the symmetry class
 of the physical system but not  on other  detailed structures.  As a  special 
 case of this general belief, 
 the eigenvalue gap distribution of random matrices should be independent of the
 probability distributions of the ensembles and thus is given by the classical Gaussian ensembles.  
Besides the eigenvalue gap distribution, similar predictions 
 hold also for short distance correlation functions of the eigenvalues. Since the gap distribution 
can be expressed in terms 
 of correlation functions, mathematical analysis is usually performed on correlation functions.
 {F}rom now on, we refer to {\it  universality} for the fact  that the  short distance behavior of
the eigenvalue correlation functions of a random matrix ensemble are the same as those of 
 the Gaussian ensemble of the same symmetry class
(Gaussian unitary, orthogonal or symplectic ensemble, i.e., GUE, GOE,
GSE).

 The universality question can be roughly divided into the bulk universality 
in the interior of the spectrum and 
 the edge universality near the spectral edges.   
 Over the past two decades,  spectacular  progress on  bulk and  edge universality   
 was made for invariant ensembles, see, e.g., 
\cite{BI, DKMVZ1, DKMVZ2, PS} and \cite{AGZ, De1, De2} for a review.   For non-invariant ensembles
with i.i.d. matrix elements ({\it Standard Wigner ensembles})  edge universality can be  proved 
via the moment method and its various  generalizations, see, e.g.,  \cite{SS, Sosh, So1}. 
In a striking contrast, the only rigorous  results  for the bulk universality of 
non-invariant Wigner ensembles
were the work by Johansson \cite{J} and subsequent improvements \cite{BP, J1} on
  {\it  Gaussian divisible Hermitian ensembles}, i.e., Hermitian ensembles   of the form
\be
H_s=  H_0+ s V,
\label{HaV}
\ee
where $H_0$  is a Wigner matrix, 
$V$ is an independent standard GUE matrix
and $s$ is a fixed positive constant independent of $N$. The Hermitian assumption is essential 
since  the key formula used  in \cite {J} 
and the earlier work \cite{BH} is valid only for {\it Hermitian ensembles}.

The bulk universality, however, was expected to hold  for general classes of  Wigner matrices, see  
Mehta's book \cite{M}, {Conjectures 1.2.1 and 1.2.2 on page 7.
 We will refer to these two conjectures as the Wigner-Dyson-Gaudin-Mehta conjecture due to their pioneering work.} 
Until a few years ago this conjecture
remained unsolved,  mainly due to the fact that  all { existing methods on local eigenvalue statistics 
depended} on explicit formulas which were not available for Wigner matrices.  
In a series of papers \cite{ESY1, ESY2, ESY3,  EPRSY, ESY4,  ESYY, EYY, EYY2}, we  developed a new approach 
to understand local eigenvalue statistics.  
This approach,  in particular,  led to the first proof  \cite{ EPRSY} of the  Wigner-Dyson-Gaudin-Mehta conjecture 
for Hermitian Wigner matrices with { smooth distributions for the matrix elements.} 
We now give a brief summary of this approach which motivates the current paper. 

The first step  was to derive a 
{\it local semicircle law}, a precise estimate of the local eigenvalue 
density,   down to energy scales containing   around $ N^\e$ eigenvalues. 
In fact, we also obtain precise bounds on the matrix elements of
the Green function.
The second step is a general approach for the universality of Gaussian divisible ensembles
by embedding the matrix \eqref{HaV} into a stochastic flow of matrices
and use that the eigenvalues evolve according to a distinguished coupled
system of stochastic differential equations, called the Dyson Brownian motion \cite{Dy}.
 The central idea is to estimate the time to local equilibrium  for 
the Dyson Brownian motion 
with the introduction of a new stochastic flow, 
the {\it local relaxation flow}, which locally behaves like a Dyson Brownian motion
but has a faster decay to global equilibrium.  
This approach \cite{ESY4, ESYY}
entirely eliminates  the usage of explicit formulas and  it  provides 
 a  unified proof for the  universality 
of Gaussian divisible ensembles for all symmetry classes.
Furthermore,  it also gives  a conceptual interpretation that the origin of the  universality 
is due to the local ergodicity of  Dyson Brownian motion.

More precisely, we will use a slightly different version  of \eqref{HaV}, namely
\be
  H_t = e^{-t/2}H_0+(1-e^{-t})^{1/2} V,
\label{HaVOU}
\ee
to ensure that the variance of $H_t$ remains independent of $t$.
Denote by $\lambda_j$ the $j$-th eigenvalue of the random matrix $H_t$, labelled in
increasing order, $\la_1\le \la_2 \le \ldots \le \la_N$,
and $\gamma_j$ 
the {\it classical location} of the $j$-th eigenvalue, i.e., $\gamma_j$ is defined by 
\be
 N \int_{-\infty}^{\gamma_j} \varrho_{sc}(x) \rd x = j, \qquad 1\leq j\le N,  \quad
\label{def:gamma}
\ee
where $\varrho_{sc}(x)=\frac{1}{2\pi}\sqrt{(4-x^2)_+}$
 is the semicircle law. Our main result on the universality for the Dyson Brownian motion 
states that, roughly speaking,  the short distance correlation functions for $H_t$ 
at the  time $t \sim N^{-2\fa}$ and $H_{t = \infty}$ are identical in weak sense provided
 that  the following main condition holds: 

\medskip

{\bf Assumption III.} There exists an $\fa>0$ such that
\be
 \sup_{ t \ge N^{-2\fa}}  \frac{1}{N} \E_t \sum_{j=1}^N(\la_j-\gamma_j)^2
  \le CN^{-1-2\fa}
\label{assum3}
\ee
with a constant $C$ uniformly in $N$. Here $\E_t$ is the expectation w.r.t. Dyson Brownian motion at the time $t$. 
The condition \eqref{assum3} has been derived from a sufficiently strong
version of the local semicircle law.

Once the universality for the Gaussian divisible ensemble is established, the last
 step is to approximate 
all matrix ensembles by  Gaussian divisible ones.  This step can be done via a 
{\it reverse heat flow} 
argument \cite{EPRSY, ESYY}  for ensembles with  smooth probability distributions 
or more generally via  the {\it Green function comparison theorem} \cite{EYY} 
which compares the distributions of eigenvalues of two ensembles around a fixed energy.
 The key input
for the latter approach was to prove a-priori  estimates on the 
 matrix elements of the Green function. These estimates  have been
obtained together with the local semicircle law.

To summarize, 
our approach to universality consists of  the following three  main steps: 
{\it Step 1.  Local semicircle law.  Step 2.   Universality for Gaussian divisible ensembles. 
Step 3.  Approximation by Gaussian divisible ensembles. }
Both Step 2 and 3 rely on a strong  local semicircle law from the Step 1.

Shortly after the preprint \cite{EPRSY}  appeared, 
another  method for the universality  was posted  by Tao and Vu \cite{TV}. 
 This method contains similar three ingredients as in  \cite{EPRSY}; 
their key result,  prior to the 
Green function comparison theorem appeared in \cite{EYY},
states that the
probability distributions of the $j$-th eigenvalue of two ensembles  for a fixed label 
 $j$ in the bulk 
are identical as $N\to\infty$ 
 provided that the first four moments of the matrix elements
 of the  two ensembles  are identical.  
This result also implies the universality of the  correlation functions for  Hermitian
Wigner ensembles \cite{TV} 
by combining it with 
the Gaussian divisible results of \cite{J, BP} for the Step 2.  For symmetric ensembles, 
 it requires the first four moments 
matching those of GOE. 
As in our approach,  a  key analytic input for \cite{TV} is the local semicircle law established in \cite{ESY2}.
The bulk universality  in the case of symmetric matrices in the  generality  as stated in Mehta's book \cite{M}  (in particular, 
without the assumption to match four moments),   was proved in  \cite{ESY4, EYY2}. The key input is to link 
universality to local ergodicity of Dyson Brownian motion,   reviewed in the previous paragraphs.

Due to the fundamental role
of the local semicircle law,   its error estimates  were improved many times since its
 first proof in \cite{ESY2}. Furthermore, it was extended to sample 
covariance ensembles \cite{ESYY} 
and  {\it generalized Wigner ensembles} \cite{EYY}  whose matrix elements are allowed to
 have different  but comparable variances.  
The best existing error estimates for local semicircle law of generalized Wigner ensembles, given 
in \cite{EYY2},  are already  almost optimal   in the bulk of the spectrum, but not near 
the edges.   
In this paper, we will prove a {\it strong local semicircle law}, Theorem \ref{45-1},  which, 
up to  $\log N$ factors, gives   optimal error estimates 
 everywhere in the spectrum.   There are four important consequences of this result:

\begin{enumerate}
\item  It implies that  
Assumption III  holds with the right hand side of \eqref{assum3} given by
 $N^{-2} (\log N)^{C\log\log N}$ for some constant $C$,
i.e., $\fa$ can be chosen arbitrary close to 1/2. Thus  the Dyson 
Brownian motion reaches local equilibrium at 
 $t \sim N^{-1 + \delta}$ for arbitrary 
small $\delta$. Up to the factor $N^{\delta}$,  this  is optimal.  Since  the time to the global 
equilibrium for the Dyson Brownian motion is order one, we have thus
 established {\it Dyson's conjecture} 
 \cite{Dy} that 
the Dyson Brownian motion reaches equilibrium in two well-separated 
stages with time scales of order one and $N^{-1}$.
As a historical  note, we mention that  
Dyson had obtained the two time scales via heuristic  physical argument and commented that 
a rigorous proof of his prediction is lacking.  Furthermore, the notion of local equilibrium was used by Dyson in a very vague 
sense, see \cite{ESY4} for a more detailed discussion.

\item  It  implies certain explicit error estimates 
for the universality of correlation functions in short scales. 

\item  It implies the {\it rigidity of eigenvalues}  in the sense that
\be\label{71.1}
\P \Bigg\{ \exists j\; : \; |\lambda_j-\gamma_j| 
\ge (\log N)^{C\log\log N}
 \Big [ \min \big ( \, j ,  N-j+1 \,  \big) \Big  ]^{-1/3}   N^{-2/3} \Bigg\}
 \le  C\exp{\big[-(\log N)^{c\log \log N} \big]}
\ee
for some positive constants $C$ and  $c$.
In other words, the eigenvalue is near its classical location 
with an error of at most $N^{-1} (\log N)^{C\log\log N}$
for generalized Wigner matrices in the bulk and the estimate deteriorates
 by a factor $ \big (\frac N {j} \big )^{1/3} $ near the edge  $j \ll N$.

 \item It implies the {\it edge universality} in the sense that the probability 
distributions of the largest (and the smallest)  eigenvalues of two generalized 
Wigner  ensembles are equal  in the large $N$ limit  provided that the second
 moments  of the two ensembles are identical. We recall the standard assumption
 that the first moments of the matrix elements are always  zero 
for all generalized Wigner ensembles. 
The comparison between our  edge universality  theorem and the previous results will be given 
at the end of Section 2 after the statement of Theorem \ref{twthm}.

\end{enumerate}

It is well-known that  the gaps between extremal eigenvalues  and their fluctuations 
 are of order $N^{-2/3}$. 
Thus the edge deterioration factor in \eqref{71.1} is the natural interpolation between $N^{-1}$ 
in the bulk and  $N^{-2/3}$ on the edges.   
The surprising feature 
of the rigidity estimate is that even if one eigenvalue is at a slightly 
 wrong location, the probability is already  extremely small.  
We remark that, without  the 
$(\log N)^{C\log \log N}$ factor, the rigidity estimate \eqref{71.1} would be wrong 
since, at  least for the classical   GUE or GOE ensembles,
 the eigenvalues are known to fluctuate on a scale $\sqrt{\log N}/N$, see \cite{GEG1, GEG2}.
For these ensembles,  the distribution of $\lambda_j - \gamma_j$ is Gaussian in the bulk.
However, 
the rigidity estimate  \eqref{71.1} in this strong probabilistic
form was not available even for the classical 
Gaussian ensembles.

\section{Main results}

Let $H=(h_{ij})_{i,j=1}^N$  be an $N\times N$  hermitian or symmetric matrix  where the
 matrix elements $h_{ij}=\ov{h}_{ji}$, $ i \le j$, are independent 
random variables given by a probability measure $\nu_{ij}$ 
with mean zero and variance $\sigma_{ij}^2\ge 0$:
\be
  \E \, h_{ij} =0, \qquad \sigma_{ij}^2:= \E |h_{ij}|^2.
\label{aver}
\ee
The distribution $\nu_{ij}$ and its variance $\sigma_{ij}^2$ may depend on $N$,
 but we omit this fact in the notation. 
Denote by $B:=\{ \sigma^2_{ij}\}_{i,j=1}^N$ the matrix of variances. 
The following  assumptions on $B$ are made throughout the paper:
\begin{description}
\item[(A)] For any $j$ fixed
\be
   \sum_{i=1}^N \sigma^2_{ij} = 1 \, .
\label{sum}
\ee
Thus $B$ is symmetric and double stochastic and, in particular, it satisfies
$-1\leq B\leq 1$.

\item[(B)]  We assume that there exists two positive constants, $\delta_-$ and
$\delta_+$, independent of $N$, such that
 \be\label{speccond}
\mbox{ 1 is a simple eigenvalue of $B$ and $\mbox{Spec}(B)
\subset [-1+\delta_-, 1-\delta_+]\cup\{1\}$.}
\ee

\item[(C)]   There is a constant $C_0$, independent of $N$, such that 
\be\label{1.3}
\max_{ij}\{\sigma_{ij}^2\}\leq \frac{C_0}{ N}.
\ee

\end{description} 

For the orientation of the reader, we  mention two special cases  that provided the main motivation
for our work.

\bigskip

\noindent {\it Example 1. Generalized Wigner matrix.}  Define $\Ci(N)$ and $\Cs(N)$ by
\be\label{defCiCs}
      \Ci(N):= \inf_{i,j}\{N\sigma^2_{ij}\}\leq \sup_{ i,j}\{N\sigma^2_{ij}\}=:\Cs(N).
\ee
The ensemble is called  {\it generalized Wigner ensemble} provided that 
\be\label{VV}
	0<C_-\le \Ci(N)\leq \Cs (N) \le C_+ <\infty,
\ee
for some $C_\pm$ independent of $N$. 
In this case, one can easily prove that $1$ is a simple eigenvalue of $B$
and \eqref{speccond} holds with some
\be\label{de-de+2}
	\delta_\pm\ge C_-,
\ee 
i.e., apart from the trivial eigenvalue, the spectrum of $B$ is separated away $\pm 1$
by  positive constants that are independent of $N$. 
The special case  $\Ci=\Cs=1$ reduces  to the standard Wigner matrices.

\bigskip

\noindent {\it Example 2. Certain band matrices with bandwidth of order $N$.}  
Band matrices are characterized by the property that $\sigma_{ij}^2$ is
a function of $|i-j|$ on scale $W$, which is called the bandwidth. 
 More precisely, the variances of a band matrix with bandwidth $1\le W\le N/2$ are given by 
\be\label{BM}
   \sigma^2_{ij} = W^{-1} f\Big(\frac{ [i-j]_N}{W}\Big),
\ee
where  $f:\bR\to \bR_+$ is a bounded nonnegative symmetric function with 
$\int f =1$ 
and we defined $[i-j]_N\in\bZ$  by the property
that  $[i-j]_N\equiv i-j \; \mbox{mod}\,\,\, N$ and
$-\frac{1}{2}N < [i-j]_N \le\frac{1}{2}N $.
We often consider the case when $W=W(N)$, i.e. the bandwidth is a function of $N$.
The condition {\bf (A)} holds only asymptotically as $W(N)\to\infty$ 
but it can be remedied by an irrelevant rescaling. If the bandwidth is
comparable with $N$, then we also have to assume that $f(x)$ is supported in $|x|\le N/(2W)$.

It is easy to see that many band matrices satisfy the spectral
 assumption \eqref{speccond}. The lower spectral bound, $-1+\delta_-\le B$ with  some
$\delta_->0$ depending only on $f$, holds for any sufficiently large $W$, see Appendix A
of \cite{ESYY}.   The parameter $\delta_+$ in the upper spectral bound typically 
behaves as of order $(W/N)^2$. Thus, for the condition {\bf (B)} to hold,
we need to assume that the bandwidth is comparable with $N$, i.e.,
it satisfies $W\ge cN$ with some positive constant $c$.
The same assumption  also guarantees that condition {\bf (C)}  holds.

We remark that the special case $W=N/2$ and $f(x)\ge c>0$ for $|x|\le 1$
was already covered by Example 1, but Example 2 allows more general band matrices that
 may have vanishing variances.  For example, with the choice of 
 $f(x) = \frac{1}{2}\cdot {\bf 1}(|x|\le 1)$,
the ensemble  with variances
\be
\sigma_{ij}^2=(N/2)^{-1}{\bf 1}\left([i-j]_N\leq N/4 \right)
\ee
is  a band matrix with bandwidth $W=N/4$.

\bigskip

Define  the Green function of $H$ by 
\be\label{green}
G_{ij}(z) =\left(\frac1{H-z}\right)_{ij}, \qquad z=E+i\eta, \qquad E\in \bR, \quad \eta>0.
\ee
The Stieltjes transform of the empirical 
eigenvalue distribution of  $H $ is given by  
\be
 m(z)=\mg (z): =   \frac{1}{N}    \sum_j G_{jj}(z) = \frac{1}{N} \tr\, \frac{1}{H-z}\, .\,\,\, 
\label{mNdef}
\ee
Define $m_{sc} (z)$ as the unique solution of
\be\label{defmsc} 
m_{sc} (z)  + \frac{1}{z+m_{sc} (z)} = 0,
\ee
with positive imaginary part for all $z$ with $\im z > 0$, i.e.,
\be\label{temp2.8}
m_{sc}(z)=\frac{-z+\sqrt{z^2-4}}{2},
\ee
where the square root function is chosen with a branch cut in the segment
$[-2,2]$ so that asymptotically $\sqrt{z^2-4}\sim z$ at infinity.
This guarantees that the imaginary part of $m_{sc}$ is non-negative for   $\eta=\im  z > 0$ and 
in the $\eta\to 0$ limit  it is the 
Wigner semicircle distribution
\be
   \varrho_{sc}(E) : = \lim_{\eta\to 0+0}\frac{1}{\pi}\im \, m_{sc}(E+i\eta)
 = \frac{1}{2\pi}
  \sqrt{ (4-E^2)_+}.
\label{def:sc}
\ee
 The Wigner semicircle law \cite{W} states that  $m_N(z) \to m_{sc} (z) $
for any fixed $z$, i.e.,
 provided that $\eta=\im z>0$ is independent of $N$. 
 Let $z=E+i\eta$ $(\eta>0)$ and denote   $\kappa:= ||E|-2|$
the distance of $E$ to the spectral edges  $\pm 2$.
 We have proved \cite{EYY2} a  local version of this result for generalized Wigner
matrices in the form of  the following probability estimate: 
\be\label{mainlsresultfake}
\P\left(|\mg(z)-m_{sc}(z)| \geq
 \frac{N^\e}{{N\eta}\, \kappa}\right)\leq \frac{C(\e, K)}{N^K}
\ee
that holds for  any fixed positive constants $\e$ and  $K$ and for
 any  $z= E+i\eta$ such that $ |E|\le 10, \; N\eta\kappa^{3/2}\ge N^\e$.
Note that this estimate deteriorates near the spectral edges as $\kappa\ll 1$. 

\medskip

In this paper we prove  the following 
local semicircle law that provides essentially the optimal estimate
{\it uniformly} in  $E = \re z$. We will estimate not only 
the deviation of $m(z)$ from $m_{sc}(z)$,  but also
the deviation of each diagonal matrix element of the resolvent, $G_{kk}(z)$,
from $m_{sc}(z)$. Moreover, we show that the off-diagonal
elements of the resolvent are small.

Let
$$
  v_k := G_{kk}-m_{sc}, \qquad m:=\frac{1}{N}\sum_{k=1}^N G_{kk}, \qquad  \barv:=\frac{1}{N}
  \sum_{k=1}^N v_k = m-m_{sc} . 
$$
Our goal is to  estimate the following  quantities
\be\label{defLambda}
  \Lambda_d:=\max_k |v_k| = \max_k |G_{kk}-m_{sc}|, \qquad
 \Lambda_o:=\max_{k\ne \ell} |G_{k\ell}|, \qquad \Lambda:=|m-m_{sc}|,
\ee
where the subscripts refer to ``diagonal'' and ``off-diagonal'' matrix elements. 
All these quantities depend on the spectral parameter $z$ and on $N$ but for simplicity we often omit
this fact from the notation.

\begin{theorem}[Strong local semicircle law] \label{45-1} 
Let $H=(h_{ij})$ be a hermitian or symmetric $N\times N$ random matrix, $N\ge 3$,
with $\E\, h_{ij}=0$, $1\leq i,j\leq N$,  and assume that the variances $\sigma_{ij}^2$ 
satisfy Assumptions {\bf (A)}, {\bf (B)}, {\bf (C)}, i.e. \eqref{sum}, \eqref{speccond} and 
 \eqref{1.3}. 
 Suppose that the distributions of the matrix elements have a uniformly 
  subexponential decay
in the sense that  there exists a constant  $\ttau>0$, independent 
of $N$, such that for any $x\ge 1$ and $1\le i,j \le N$ we have
\be\label{subexp}
\P(|h_{ij}|> x \sigma_{ij})\le \ttau^{-1} \exp{ \big( - x^\ttau \big)} .
\ee
 Then  there exist positive constants   $A_0 > 1$,  $ C, c$ and $\phi < 1$  
depending only on $\ttau$, on $\delta_\pm$ from Assumption {\bf (B)} and
on $C_0$ from Assumption {\bf (C)}, such that
 for all  $L$ with 
\be
A_0\log\log N\le L\le \frac{\log (10 N)}{10\log\log N}
\label{Lbound}
\ee
 the following estimates hold for any sufficiently large $N\ge N_0(\ttau, \delta_\pm, C_0)$:

(i) The Stieltjes transform of the empirical 
eigenvalue distribution of  $H $  satisfies 
\be\label{Lambdafinal} 
\P \Big ( \bigcup_{z\in \bS_L} \Big\{ \Lambda(z) 
 \ge \frac{(\log N)^{4L}}{N\eta} \Big\}   \Big )\le  C\exp{\big[-c(\log N)^{\phi L} \big]}, 
\ee
where
\be
{\bf  S}:={\bf  S}_L=\Big\{ z=E+i\eta\; : \;
 |E|\leq 5,  \quad  N^{-1}(\log N)^{10L} < \eta \le  10  \Big\}.
\label{defS}
\ee

(ii) The individual  matrix elements of
the Green function  satisfy that 
\be\label{Lambdaodfinal}
\P \left  ( \bigcup_{z\in \bS_L} \left\{ \Lambda_d(z)  + \Lambda_o (z) \geq 
(\log N)^{4L} \sqrt{\frac{\im m_{sc}(z)  }{N\eta}} + \frac{(\log N)^{4L}}{N\eta}
   \right\}    \right)
\leq  C\exp{\big[-c(\log N)^{\phi L} \big]}.
\ee

(iii)  The largest eigenvalue of $H$ is bounded by $2+N^{-2/3}(\log N)^{ 9L} $ in the sense that 
\be\label{443}
\P \Big (  \max_{j=1, \ldots, N} |\lambda_j  |\ge 2+N^{-2/3}(\log N)^{ 9 L} 
  \Big) 
  \le   C\exp{\big[-c(\log N)^{{\phi L}  } \big]}.  
\ee

\end{theorem}

The subexponential decay condition \eqref{subexp} can be
 weakened  if we are not aiming 
at error estimates faster than any power law of $N$. 
This can be easily carried out and we will not pursue  it in this paper. 
We also note that the upper bound on $L$ originates from the natural requirement that
$\bS_L\neq\emptyset$.  
\medskip

Prior  to our results in \cite{EYY} and \cite{EYY2}, a central limit theorem  for
the semicircle law on macroscopic scale for band matrices was established
by Guionnet \cite{gui} and Anderson and Zeitouni \cite{AZ}; a
 semicircle law  for Gaussian band matrices  was proved 
by Disertori, Pinson and Spencer \cite{DPS}.
 For a review on band matrices, see  the recent  article \cite{Spe} by Spencer.

The local semicircle estimates imply that the empirical counting function of
the eigenvalues is close to the semicircle counting function 
and that the locations of the eigenvalues are close to their classical 
location in mean square deviation sense.
Recall that $\la_1\le \la_2 \le \ldots \le \la_N$ are the ordered
eigenvalues of $H$.
We define the {\it normalized empirical counting function}  by
\be
 {\mathfrak n}(E):= \frac{1}{N}\# \{ \lambda_j\le E\}.
\label{deffn}
\ee
Let
\be
 n_{sc}(E) :  = \int_{-\infty}^E \varrho_{sc}(x)\rd x
\label{nsc}
\ee
be the distribution function
of the semicircle law and recall that $\gamma_j =\gamma_{j,N}$ denote the 
classical location of the $j$-th point
under the semicircle law, see \eqref{def:gamma}.

\begin{theorem} [Rigidity of Eigenvalues] \label{7.1}
Suppose that Assumptions {\bf (A)}, {\bf (B)}, {\bf (C)}  and  the condition \eqref{subexp} hold. 
Then  there exist positive constants   $A_0 > 1$, $C, c$ and $\phi < 1$  
depending only on $\ttau$, on $\delta_\pm$ from Assumption {\bf (B)} and
on $C_0$ from Assumption {\bf (C)}
such that  for  any $L$ with 
$$
A_0\log\log N\le L\le \frac{\log (10 N)}{10\log\log N}
$$
we have
\be\label{rigidity}
\P \Bigg\{  \exists j\; : \; |\lambda_j-\gamma_j| 
\ge (\log N)^{ L}  \Big [ \min \big ( \, j ,  N-j+1 \,  \big) \Big  ]^{-1/3}   N^{-2/3} \Bigg\}
 \le  C\exp{\big[-c(\log N)^{\phi L} \big]}
\ee
and
\be
   \P\Bigg\{ \sup_{|E|\le 5} \big| {\mathfrak n} (E)-n_{sc}(E)\big| \,  \ge
 \frac{(\log N)^{L}}{N} \Bigg\}\le  C\exp{\big[-c(\log N)^{\phi L } \big]}  
\label{nn}
\ee
for any sufficiently large $N\ge N_0(\ttau, \delta_\pm, C_0)$. 
\end{theorem}

For standard Wigner matrices, \eqref{nn} with the factor $N^{-1}$ replaced by  $N^{-2/5}$
 (in a weaker sense with some modifications in the statement)  was established in \cite{BMT}
and  a  stronger $N^{-1/2}$ control was proven for $\E \fn (E)-n_{sc}(E)$.
 If we replaced $(\log N)^{L }$ factor by $N^\delta$ for arbitrary $\delta>0$, 
\eqref{nn} was proved in \cite{EYY2} (Theorem 6.3)  with some deterioration near the 
spectral edges 
and with a slightly weaker probability estimate. 
In Theorem 1.3 of a recent preprint \cite{TV4}, the following estimate (in our scaling)
\be\label{tv}
\Big (\E \big[  |\lambda_j-\gamma_j|^2 \big] \Big )^{1/2} \le 
\Big [ \min \big ( \, j ,  N-j+1 \,  \big) \Big  ]^{-1/3}   N^{-1/6 - \e_0}
\ee
with some small positive $\e_0$ 
was  proved under the assumption that the third moment of the matrix element vanishes and all variances 
of the matrix elements are identical, i.e., for the standard Wigner matrices 
with vanishing third moment.
In the same paper, it was conjectured that the factor $N^{-1/6 - \e_0}$ on 
the right hand side of \eqref{tv} should be replaced by $N^{-2/3 + \e}$. 
Prior to the work \cite{TV4}, the estimate \eqref{rigidity} 
away from the edges with a slightly weaker probability estimate 
and with the $(\log N)^{L}$ factor replaced by $N^\delta$ for arbitrary $\delta>0$
was proved in \cite{EYY2} (see the equation before (7.8) in \cite{EYY2}).  
For Wigner matrices whose matrix element distributions matching the 
standard Gaussian random variable  
up to the third  moment, it was proved in \cite{TV} that 
$ |\lambda_j-\gamma_j| \le N^{-1 + \e}$ holds in the bulk in  probability
(Theorem 32). More detailed behavior can be obtained if one assumes 
further that the fourth moment also matches  
the standard Gaussian random variable, see Corollary 21 of \cite{TV}
 for more details. Near the  edge,  \eqref{rigidity}
with  $N^{-2/3}$  replaced by $N^{-1/2}$ and  the probability estimate on the right side 
replaced by a Gaussian type estimate
was proved in \cite{Alon}.

We remark that all results in this paper are stated for both the hermitian or symmetric case, but
the statements hold for quaternion self-dual random matrices as well
(see, e.g., Section 3.1 of \cite{ESYY}). The proofs will be
presented for the hermitian case for definiteness
but with obvious modifications
they are valid for the other two cases as well.

We will frequently use the notation $C$ and $c$
for generic positive constants and $N_0$ for the lower threshold for $N$
in this paper. We adopt the convention that, unless stated otherwise,
 these constants
and also the implicit constants in the  $O(\cdot)$ notation
may depend on the basic parameters of our model, namely on $\ttau$,
$\delta_\pm$ and $C_0$. The values of these generic constants may
change from line to line.

\subsection{Bulk Universality}

We now  use Theorem \ref{7.1}   to establish the speed of convergence 
 for local statistics  of Dyson Brownian motion.
 In fact, we will replace the Brownian motion in the definition
of Dyson Brownian motion by an Ornstein-Uhlenbeck  process. We thus
consider a flow of random matrices $H_t$ satisfying 
the following matrix valued stochastic differential equation
\be
   \rd H_t = \frac{1}{\sqrt{N}}\rd\beta_t - \frac{1}{2}H_t \rd t,
\label{OUflow}
\ee
where $\beta_t$ is a hermitian matrix valued process 
whose diagonal matrix elements are standard real Brownian motions
and the off-diagonal elements are independent standard complex Brownian motions;
with all Brownian motions being independent. The initial condition $H_0$
is the original  hermitian Wigner matrix.
For any fixed $t\ge 0$, the distribution of $H_t$ coincides with that of
\be\label{matrixdbm}
e^{-t/2} H_0 + (1-e^{-t})^{1/2}\, V,
\ee
where 
$V$ is an independent  GUE 
 matrix whose matrix elements are centered Gaussian random variables with variance $1/N$. 
For the symmetric case, the matrix elements of $\beta_t$ in \eqref{OUflow} 
are real Brownian motions and $V$ in \eqref{matrixdbm} is a GOE matrix.
 It is well-known that the eigenvalues of $H_t$ follow a process
that is also called the Dyson Brownian motion
 (in our case with a drift but we will still call it Dyson Brownian motion).
 
More precisely, let 
\be\label{H}
\mu=\mu_N(\rd{\bf x})=
\frac{e^{-\cH({\bf x})}}{Z_\beta}\rd{\bf x},\qquad \cH({\bf x}) =
N \left [ \beta \sum_{i=1}^N \frac{x_{i}^{2}}{4} -  \frac{\beta}{N} \sum_{i< j}
\log |x_{j} - x_{i}| \right ]
\ee
be the probability measure of the eigenvalues of the general $\beta$ ensemble,
with $\beta\ge 1$ ($\beta=2$ for GUE, $\beta=1$ for GOE). Here $Z_\beta$ is the normalization factor so that $\mu$ is probability measure. 
In this section, we often use the notation $x_j$ instead of $\lambda_j$ for 
the eigenvalues to follow the notations of \cite{ESYY}.
 Denote the distribution of 
the eigenvalues  at  time $t$
by $f_t ({\bf x})\mu(\rd {\bf x})$.
Then $f_t$ satisfies
\be\label{dy}
\partial_{t} f_t =  \cL f_t.
\ee
where 
\be
\cL=   \sum_{i=1}^N \frac{1}{2N}\partial_{i}^{2}  +\sum_{i=1}^N
\Bigg(- \frac{\beta}{4} x_{i} +  \frac{\beta}{2N}\sum_{j\ne i}
\frac{1}{x_i - x_j}\Bigg) \partial_{i}.
\label{Lgen}
\ee
For any $n\ge 1$ we define the $n$-point
correlation functions (marginals) of the probability measure $f_t\rd\mu$ by
\be
 p^{(n)}_{t,N}(x_1, x_2, \ldots,  x_n) = \int_{\R^{N-n}}
f_t(\bx) \mu(\bx) \rd x_{n+1}\ldots
\rd x_N.
\label{corr}
\ee
With a slight abuse of notations, we will sometimes  also use $\mu$ to denote
the density of the measure $\mu$ with respect to the Lebesgue measure.
The correlation functions of the equilibrium measure  are denoted by
\be
 p^{(n)}_{\mu,N}(x_1, x_2, \ldots,  x_n) = \int_{\R^{N-n}}
\mu(\bx) \rd x_{n+1}\ldots
\rd x_N.
\label{correq}
\ee

The main result in \cite{ESYY} concerning Dyson Brownian motion, Theorem 2.1, 
 states that the  local ergodicity of Dyson Brownian motion  holds for 
 $t \ge N^{-2\fa+\delta}$  for any $\delta > 0$ 
provided that the  Assumption III \eqref{assum3} holds.
 In fact,   the  estimate on the relaxation to the local equilibrium  \cite{ESYY} is not restricted to 
Dyson Brownian motion;  it applies to all flows satisfying { four general assumptions, labelled as Assumption I-IV 
in \cite{ESYY}. Instead of repeating  these assumptions in their general  forms, we will give only 
simple sufficient conditions.
Assumption I requires that the probability density of the global equilibrium measure} is given by 
a  Hamiltonian  of the form 
\be
  \cH = \cH_N(\bx)  = \beta\Big[ \sum_{j=1}^N U(x_j)
-\frac{1}{N}\sum_{i<j} \log|x_i-x_j| \Big], 
\label{ham}
\ee
where $\beta\ge 1$ and  the function $U:\bR\to \bR$ is smooth with
$U'' \ge \delta $ {for some positive $\delta$.}
 This is clearly satisfied since the equilibrium measures  are either GUE or GOE
in the setting of this paper.
 Assumption II requires  a limiting continuous density for the eigenvalue distribution. 
In our case, the density is given by the semicircle law. 
Assumption IV  asserts that the { local density of eigenvalues is bounded down to scale $\eta=N^{-1 + \sigma}$} for
any $\sigma > 0$. This assumption 
follows from the large deviation estimate 
\eqref{Lambdafinal} since a bound on $\Lambda(z)$, 
 $z=E+i\eta$, can be easily used to prove
an upper bound  
on the local density of eigenvalues in a window of size $\eta$ about $E$. 
 As usual, the additional
condition in \cite{ESYY} on the entropy  $S_\mu(f_{t_0})\le CN^m$ for $t_0= N^{-2\fa}$
holds due to the regularization property of the Ornstein-Uhlenbeck process. Thus  for a 
given $0<\e'<1$, choosing  $\fa = 1/2 - \e'/2, \; A = \e'$ 
in the second part of Theorem 2.1 in \cite{ESYY} and using \eqref{rigidity},
we have the following theorem. 

\begin{theorem}  [Strong local ergodicity of Dyson Brownian motion] \label{thm:main} 
Let $H$ be a hermitian or symmetric $N\times N$ random matrix with $\E\, h_{ij}=0$
and suppose that Assumptions {\bf (A)}, {\bf (B)}, {\bf (C)}  and   \eqref{subexp} hold
with parameters $\delta_\pm, C_0$ and $\vartheta$.
 Then for any $\e'> 0$, $\delta>0$, $c>0$
positive numbers, for
any integer $n\ge 1$ and for any compactly supported continuous test function
$O:\bR^n\to \bR$ there exists a constant $C$ depending on all these parameters and on $O$
such that
\be
\begin{split}
&  \sup_{t\ge N^{-1 + \delta + \e'}}  \Big | \int_{E-b}^{E+b}\frac{\rd E'}{2b}
\int_{\R^n}   \rd\alpha_1
\ldots \rd\alpha_n \; O(\alpha_1,\ldots,\alpha_n) 
 \frac{1}{\varrho(E)^n} \\
& \quad \times   \Big ( p_{t,N}^{(n)}  - p_{\mu, N} ^{(n)} \Big )
\Big (E'+\frac{\alpha_1}{N\varrho(E)},
\ldots, E'+\frac{\alpha_n}{ N\varrho(E)}\Big) \Big | 
\leq 
C N^{2 \e'} \Big[ b^{-1}  N^{ - 1+\e' } +  b^{-1/2} N^{-\delta/2} \Big],
\label{abstrthm2}
\end{split}
\ee
holds for any fixed $E\in [2 -c, 2 + c] $ and for
any  $b=b_N\in (0,c/2)$ that may depend on $N$.
Here  $ p_{t,N}^{(n)}$ and  $ p_{\mu,N}^{(n)}$, \eqref{corr}--\eqref{correq},
 are the correlation functions of the
eigenvalues of the Dyson Brownian motion flow \eqref{matrixdbm}  and those of
the equilibrium measure, respectively. 
\end{theorem}

Besides a weaker version of Theorem \ref{thm:main} was proved in \cite{EYY2},
a similar result, with no error estimate,  was obtained in \cite{EPRSY} for the hermitian case
by using an explicit formula related to 
Johansson's formula \cite{J}.
 Theorem \ref{thm:main}, however,
 contains  explicit estimates 
and is valid for a time range much bigger than the previous results.
In particular,  we mention  the following  
three special cases: 

\begin{itemize}
\item If we  choose $\delta = 1- 2 \e'$ and thus $t = N^{- \e'}$,   
then we can choose $b \sim N^{-1}$ and 
the universality is valid with essentially no averaging in $E$.

\item If we choose  the energy
window of size $b\sim 1$ and the time $t = N^{- \e'}$, then 
the error estimate is of order  $\sim N^{-1/2}$.

\item  If we choose $b \sim 1$, then the smallest time scale for which
we can prove the universality is 
$t = N^{-1 + \e'}$. This scale, up to the arbitrary small exponent $\e'$, is optimal 
in accordance 
with the time scale to local equilibrium conjectured by Dyson \cite{Dy}.

\end{itemize}

For generalized Wigner matrices with a subexponential
decay, i.e. assuming \eqref{VV}
in addition to  the conditions of Theorem \ref{thm:main}, 
the universality result with no explicit error 
estimate  holds  for any time $t\ge0$. More precisely,
for any fixed $b>0$ we have
\be
\begin{split}
 \lim_{N\to0}\sup_{t\ge 0}\; & \Big | \int_{E-b}^{E+b}\frac{\rd E'}{2b}
\int_{\R^n}   \rd\alpha_1
\ldots \rd\alpha_n \; O(\alpha_1,\ldots,\alpha_n) 
 \frac{1}{\varrho(E)^n} \\
& \quad \times   \Big ( p_{t,N}^{(n)}  - p_{\mu, N} ^{(n)} \Big )
\Big (E'+\frac{\alpha_1}{N\varrho(E)},
\ldots, E'+\frac{\alpha_n}{ N\varrho(E)}\Big) \Big | 
=0.
\label{abstrthm21}
\end{split}
\ee
This result, with slightly stronger conditions on
the distributions of the ensemble,  was already proved   in  \cite{EYY2}.
Similarly to \cite{EYY2}, the extension of the universality from a small positive time
to zero time requires a different method, the Green function comparison 
theorem \cite{EYY}   in our approach. The reasons of  universality  for zero time and
 time bigger than $1/N$ are very different:
Theorem \ref{thm:main} shows that 
the local correlation functions have already reached their equilibrium 
under the Dyson Brownian motion flow for any time larger than $1/N$. 
For time smaller than $1/N$,  in particular the important case $t=0$, the universality is valid 
because  we can compare the local correlation functions at time $t=0$ with 
the ones generated by the flow at time  $t=N^{-\e}$ with 
{\it specially adjusted initial data} (see, e.g., the Matching Lemma 3.4 \cite{EYY2}).
The same argument as in Section 3 of \cite{EYY2} can be used to prove \eqref{abstrthm21} from 
\eqref{abstrthm2}.
In fact, since our new version of  the strong local ergodicity of Dyson Brownian motion, 
Theorem \ref{thm:main},   holds for very short times,  the two ensembles to be compared 
are already  very close to each other. 
 Furthermore,  
effective error estimates instead of a limiting statement \eqref{abstrthm21} can also be obtained
and the parameter $b$ may also be chosen $N$-dependent. For the case that $b$ is $N$-independent, 
the time to local equilibrium as remarked above is $N^{-1+ \e}$. Hence the condition \eqref{VV} can be replaced by
the following condition: there are  constants $c, \e  > 0$ such that 
\be
| \{ (i, j): N \sigma_{ij}^2  \le c \} | \le N^{2 - \e}  \; .
\ee
 Since these extensions require only minor modifications of the current method, 
 we will not pursue these directions in this paper.

\subsection{Edge distribution}

Recall that $\lambda_N$ is the largest eigenvalue of the random matrix.  The 
probability distribution  functions of $\lambda_N$ for the classical Gaussian ensembles are
identified by Tracy and Widom  \cite{TW, TW2} to be  
\be\label{Fb}
\lim_{N \to \infty} \P( N^{2/3} ( \lambda_N -2) \le s ) =  F_\beta (s), 
\ee
where the function $F_\beta(s)$ can be computed in terms 
of Painlev\'e equations and $\beta=1, 2, 4$
 corresponds to the standard classical  ensembles. The distribution of $\lambda_N$ 
is believed to be universal and independent of the Gaussian structure. 
 The strong local semicircle law, Theorem \ref{45-1},
combined with a modification of the Green function comparison theorem
(Theorem \ref{GFCT})
implies  the following 
version of universality of the extreme eigenvalues.

\bigskip

\begin{theorem}[Universality of extreme eigenvalues] \label{twthm}  
Suppose that we have 
two  $N\times N$  matrices, $H^{(v)}$ and $H^{(w)}$, with matrix elements $h_{ij}$
given by the random variables $N^{-1/2} v_{ij}$ and 
$N^{-1/2} w_{ij}$, respectively, with $v_{ij}$ and $w_{ij}$ satisfying
the uniform subexponential decay condition \eqref{subexp}. Let $\P^\bv$ and
$\P^\bw$ denote the probability and $\E^\bv$ and $\E^\bw$ 
the expectation with respect to these collections of random variables.
Suppose that Assumptions {\bf (A)}, {\bf (B)}, {\bf (C)}  hold for both ensembles.  If 
the first two moments of
 $v_{ij}$ and $w_{ij}$ are the same, i.e.
\be\label{2m}
    \E^\bv \bar v_{ij}^l v_{ij}^{u} =  \E^\bw \bar w_{ij}^l w_{ij}^{u},
  \qquad 0\le l+u\le 2,
\ee
then there is an $\e>0$ and $\delta>0$
depending on $\vartheta$ in \eqref{subexp} 
 such that 
for any real parameter $s$ (may depend on $N$)   
we have
 \be\label{tw}
 \P^\bv ( N^{2/3} ( \lambda_N -2) \le s- N^{-\e} )- N^{-\delta}  
  \le   \P^\bw ( N^{2/3} ( \lambda_N -2) \le s )   \le 
 \P^\bv ( N^{2/3} ( \lambda_N -2) \le s+ N^{-\e} )+ N^{-\delta}  
\ee 
for $N\ge N_0$ sufficiently  large, where $N_0$ is independent of $s$. 
Analogous result holds for the smallest eigenvalue $\lambda_1$.

\end{theorem}

Theorem \ref{twthm} can be extended to finite correlation functions of  extreme eigenvalues.
 For example, 
we have the following extension to \eqref{tw}:
 \begin{align}\label{twa}
&  \P^\bv \Big ( N^{2/3}  ( \lambda_N -2) \le s_1- N^{-\e}, \ldots, N^{2/3} ( \lambda_{N-k} -2) 
\le s_{k+1}- N^{-\e} \Big )- N^{-\delta}   \nonumber \\
& 
 \le   \P^\bw \Big ( N^{2/3} (  \lambda_N -2) \le s_1,  \ldots, N^{2/3} ( \lambda_{N-k} -2) 
\le s_{k+1}  \Big )  \\
 &  \le 
 \P^\bv \Big ( N^{2/3} ( \lambda_N -2) \le s_1+ N^{-\e}, \ldots, N^{2/3} ( \lambda_{N-k} -2)
 \le s_{k+1}+ N^{-\e}   \Big )+ N^{-\delta}  \nonumber 
\end{align}  
for all $k$ fixed and $N$ sufficiently  large. The proof of \eqref{twa} is similar to that
 of \eqref{tw} and we will 
not provide details except stating  the general form of the Green function comparison 
theorem (Theorem \ref{GFCT2}) 
needed in this case.  We remark  that edge universality is usually
formulated in terms of joint distributions of edge eigenvalues
in the form \eqref{twa} with fixed parameters $s_1, s_2, \ldots $. Our result 
holds uniformly in these parameters, i.e., 
they may depend on $N$. 
However, the interesting regime is $|s_j|\le (\log N)^{C\log\log N}$, otherwise
the rigidity estimate \eqref{rigidity} gives a stronger control than
\eqref{twa}.

The edge universality for Wigner matrices was first  proved 
via the moment method by Soshnikov \cite{Sosh} (see also the earlier work \cite{SS})
for Hermitian and orthogonal  ensembles with  symmetric distributions
to ensure that all odd moments vanish.  
By combining the moment 
method and Chebyshev polynomials, Sodin 
proved edge universality of band matrices and some special class
of sparse matrices \cite{So1, So2}.

The removal of the symmetry assumption was not straightforward. 
The approach of  \cite{So1, So2} is  restricted to ensembles with symmetric distributions. 
The symmetry assumption  was partially removed in 
\cite{P1, P2} and significant progress was made in  \cite{TV2} which assumes only  
that the first three moments of two Wigner ensembles  are identical. In other words, 
the symmetry assumption was replaced 
by the vanishing third moment condition for Wigner matrices.  For a special class of
 ensembles, the Gaussian divisible  Hermitian ensembles, edge universality was proved \cite {J1} 
under the sole condition that the  fourth moment is finite, which
in our scaling means  that $\E |\sqrt{N}h_{ij}|^4$ is a positive constant. 
Using this result \cite{J1}, one can remove the vanishing third moment condition in 
\cite{TV2}  for   Hermitian Wigner ensembles.

In comparison with these results,  Theorem 
\ref{twthm} does not imply  the edge universality of band matrices or sparse matrices 
\cite{So1, So2},  but it implies in particular that, for the purpose  to identify 
the distribution of the top 
eigenvalue for a generalized Wigner matrix with the subexponential decay condition, 
it suffices to consider generalized Wigner ensembles with Gaussian distribution. 
Since the distributions 
of the top eigenvalues of the  Gaussian Wigner ensembles 
are given by $F_\beta$ \eqref{Fb},  Theorem 
\ref{twthm} implies  the edge universality of the standard 
Wigner matrices under the subexponential decay assumption alone. 
We remark that one can use Theorem \ref{7.1} as an input
in the approach of \cite{J1} to prove that the distributions of the top eigenvalues of
the generalized hermitian Wigner ensembles with Gaussian distributions  are given by $F_2$.
Therefore the  Tracy-Widom distribution also holds for any generalized hermitian
Wigner ensemble with subexponential decay.
But for ensembles in different symmetry classes (e.g., symmetric Wigner ensembles),
 there is no corresponding  result to identify  the distribution  of the top 
eigenvalue with $F_\beta$ if the variances are allowed to vary. 

Finally, we comment that  the subexponential  decay 
assumption in our approach, though can be weakened, is far from optimal,  
see  \cite{ABP, BBP, Ruz, Sosh2}
for discussions on optimal  moment assumptions.
  Our approach based on the local semicircle law, however,  gives both the bulk and 
edge universality and the symmetry 
of the distribution of matrix elements  plays no role.

\section{Apriori bound for the strong local semicircle law}\label{ld-semi}

We first prove a weaker form of Theorem \ref{45-1},
and in Section~\ref{sec:optimal}
we will use this apriori bound to obtain the stronger form
as claimed in  Theorem \ref{45-1}.

\begin{theorem}\label{thm:detailed}
Let $H=(h_{ij})$ be a hermitian  $N\times N$ random matrix, $N\ge 3$,
with $\E\, h_{ij}=0$, $1\leq i,j\leq N$,  and assume that the variances $\sigma_{ij}^2$ 
satisfy Assumptions {\bf (A)}, {\bf (B)}, {\bf (C)}
and assume the uniform subexponential decay \eqref{subexp}. 
Then there exist constants  $0<\phi<1 $, $C\ge 1$ and $c>0$,
 depending only on $\ttau$ from \eqref{subexp}, $\delta_\pm$ from Assumption {\bf (B)} and
on $C_0$ is from Assumption {\bf (C)}  such that 
 for any $\ell$ with $4/\phi\le \ell\le C\log N/\log\log N$ 
 and for any   $z=E+i\eta \in \bS_\ell$   we have 
\be\label{mainlsresult}
\P\left\{ \max_{i}|G_{ii}(z)-m_{sc}(z)|\geq 
\frac{(\log N)^\ell }{(N\eta)^{1/3}}\,\right\}\leq  C\exp{\big[-c(\log N)^{\phi  \ell } \big]}
\ee
and
\be\label{mainlsresult2}
\P\left\{\max_{i\ne j}|G_{ij}(z)|\geq 
\frac{ (\log N)^\ell  }{(N\eta)^{1/2}}\right\}\leq C\exp{\big[-c(\log N)^{\phi \ell } \big]}
\ee
 for any sufficiently large $N\ge N_0(\theta, \delta_\pm, C_0)$.
\end{theorem}

 We remark that the probabilistic estimates
in Theorem \ref{thm:detailed} are stated for fix $z\in {\bf S}_\ell $, but it is
easy to deduce from them probabilistic 
statements that hold simultaneously for all $z$, e.g.
$$
   \P\left( \bigcup_{z\in \bS_\ell } \Big\{ \max_{i}|G_{ii}(z)-m_{sc}(z)|\geq 
\frac{(\log N)^\ell  }{(N\eta)^{1/3}}\, \Big\}\right)\leq C\exp{\big[-c(\log N)^{\phi  \ell } \big]}.
$$
 This holds true because in the set $\bS_\ell$ 
the Green function and $m_{sc}(z)$ 
are Lipschitz continuous  in $z$
with a Lipschitz constant bounded by $\eta^{-2}\ll N^2$; for
 example $|\partial_z G_{ij}(z)|\le 
|\im z|^{-2}\le N^2$.  Consider  an $N^{-10}$-net 
in the compact set ${\bf S}_\ell $, i.e., a set of points $\{ z_k\} \subset \bS_\ell $
such that $\min_k |z-z_k|\le N^{-10}$ for any $z\in \bS_\ell $ and such that
the cardinality of  $\{ z_k\}$ is at most $CN^{20}$.
Using that the estimates \eqref{mainlsresult}--\eqref{mainlsresult2} hold
simultaneously for all points $z_k$ (since
these estimates decay faster than any polynomial in $N$ by $\phi \ell>1$),
we see that  similar estimates,  with a smaller $c$, hold
simultaneously for any $z\in {\bf S}_\ell$.

\medskip

We will follow the self consistent perturbation ideas  initiated in \cite{EYY, EYY2}.
We first introduce some notations. 

\begin{definition}\label{basicd}
 Let ${\T}=\{k_1$, $k_2$, $\ldots$, $k_t\}\subset \{1,2, \ldots ,N\}$
 be  an unordered set of $|\bT|=t$ elements and let 
 $H^{(\bT)}$ be the $N-t$ by $N-t$ minor of $H$ after removing the
 $k_i$-th $(1\leq i\leq t)$ rows and columns. For $\bT=\emptyset$, we define $H^{(\emptyset)}=H$.
Similarly, we define $\ba^{(j;\,\,{\T})}$ to be
 $j$-th column of $H$ with the $k_i$-th $(1\leq i\leq t)$
 elements removed. Sometimes, we just use the short notation 
$\ba^{j}$=$\ba^{(j;\,\,{\T})}$. Note that the $\ell$-th entry of $\ba^j$ is
 $\ba^j_\ell= h_{\ell j}$ for $\ell\not\in \T$.
 For any ${\T}\subset \{ 1, 2, \ldots , N\}$ we introduce the following notations:
 \begin{align}
 G^{({\T})}_{ij}:=&(H^{({\T})}-z)^{-1}(i,j) ,\qquad i,j\not\in\T\non\\
 Z^{({\T})}_{ij}:=&\ba^{i}\cdot(H^{({\T})}-z)^{-1}\ba^{j}=\sum_{k,\ell\notin {\T}}
\overline{\ba^{\,i}_k} G^{({\T})}_{k \ell}\ba^{j}_{\ell\,} \non \\
\wH^{({\T})}_{ij}:= & h_{ij}-z\delta_{ij}-Z^{({\T})}_{ij}. 
 \end{align}
These quantities depend on $z$, but we mostly neglect this dependence in the notation. 
\end{definition}
\bigskip

The following  formulas  were proved in
  Lemma 4.2  of \cite{EYY}. 

\begin{lemma}[Self-consistent perturbation formulas] \label{basicIG} Let 
$\bT\subset \{ 1, 2, \ldots, N\}$. For  simplicity, we use the
 notation $(i \,{\T})$ for $(\{i\}\cup {\T})$ and $(i j \,{\T})$
 for $(\{i,j\}\cup {\T})$. 
 Then we have the following identities:
\par\begin{enumerate}    
\item For any  $i\notin {\T}$ 
\be\label{GiiHii} 
 G^{({\T})}_{ii}=(\wH^{(i\,{\T})}_{ii})^{-1}.
 \ee
 \item For $i\neq j$ and $i,j\notin {\T}$
\be\label{GijHij} 
 G^{({\T})}_{ij}=-G^{({\T})}_{jj}G_{ii}^{(j\,{\T})}\wH^{(ij\,\,{\T})}_{ij}=
-G^{({\T})}_{ii}G_{jj}^{(i\,{\T})}\wH^{(ij\,\,{\T})}_{ij}.
 \ee
 	\item  For $i\neq j$ and $i,j\notin {\T}$
  \be\label{GiiGjii}
  G^{({\T})}_{ii}-G^{(j\,\,{\T})}_{ii}=
G^{({\T})}_{ij}G^{({\T})}_{ji}(G^{({\T})}_{jj})^{-1}.
  \ee
  \item  For any indices  $i$, $j$ and $k$ that are different  and 
  $i,j,k \notin {\T}$
 \be\label{GijGkij}
G^{({\T})}_{ij}-G^{(k\,\,{\T})}_{ij}=G^{({\T})}_{ik}G^{({\T})}_{kj}
(G^{({\T})}_{kk})^{-1} . \ee 
 \end{enumerate}
 
\end{lemma}

The following  large deviation
estimates concerning independent random variables were proved in 
Appendix B of \cite{EYY}. 

\begin{lemma}\label{generalHWT}
Let $a_i$ ($1\leq i\leq N$) be independent complex random  variables with mean zero, 
variance $\sigma^2$  and having a uniform  subexponential decay 
\be\label{subexp1}
\P(|a_{i}|\geq x \sigma)\leq \ttau^{-1} \exp{\big( - x^\ttau\big)}, \qquad \forall \; x\ge 1,
\ee
with some $\ttau>0$.
Let $A_i$, $B_{ij}\in \C$ ($1\leq i,j\leq N$). 
Then   there exists a constant $0< \phi<1$, depending on $\ttau $,
 such that for any $\zeta > 1$ 
we have
 \begin{align}
\P\left\{\left|\sum_{i=1}^N a_iA_i\right|\geq  (\log N)^{\zeta}
 \sigma \,\Big(\sum_{i}|A_i|^2\Big)^{1/2}\right\}\leq &\; \exp{\big[-(\log N)^{\phi \zeta  } \big]},
\label{resgenHWTD} \\
\P\left\{\left|\sum_{i=1}^N\overline a_iB_{ii}a_i-\sum_{i=1}^N\sigma^2 B_{ii}\right|\geq 
(\log N)^{\zeta} \sigma^2 \Big( \sum_{i=1}^N|B_{ii}|^2\Big)^{1/2}\right\}\leq &\;
  \exp{\big[-(\log N)^{\phi \zeta } \big]},\label{diaglde}\\
\P\left\{\left|\sum_{i\neq j}\overline a_iB_{ij}a_j\right|\geq (\log N)^{\zeta} \sigma^2 
\Big(\sum_{i\ne j} |B_{ij}|^2 \Big)^{1/2}\right\}\leq & \; \exp{\big[-(\log N)^{\phi \zeta } \big]}
 \label{resgenHWTO}
\end{align}  
for any sufficiently large $N\ge N_0$, where $N_0=N_0(\ttau)$ depends on  $\ttau$.
\end{lemma}

The following lemma (Lemma 4.2 from \cite{EYY2}) collects elementary properties of the 
Stieljes transform  of the semicircle law.  As a technical note, we use the notation $f\sim g$ 
 for two positive functions
in some domain $D$  if
there is a positive universal constant $C$ 
such that $C^{-1}\le f(z)/g(z) \le C$ holds for all $z\in D$.

\begin{lemma}\label{lm:msc}  We have for all $z$ with $\im z>0$ that
\be
   |m_{sc}(z)| = |m_{sc}(z)+z|^{-1}\le 1.
\label{zmsc2}
\ee
{F}rom  now on, let $z=E+i\eta$ with $|E|\le 5$ and $0<\eta \le 10$
and we set $\kappa= \big| \, |E|-2\big|$.
Then we  have
\be
   |m_{sc}(z)|\sim 1, \qquad   |1-m_{sc}^2(z)|\sim \sqrt{\kappa+\eta}
\label{smallz}
\ee
and the following two bounds: 
\begin{align}\label{esmallfake}
\im m_{sc} (z)\sim & \left\{\begin{array}{cc} 
 \frac{\eta}{\sqrt{\kappa+\eta}} & \mbox{if  $\kappa\ge\eta$ and $|E|\ge 2$} \\  & \\
\sqrt{\kappa+\eta} & \mbox{if $\kappa\le \eta$ or $|E|\le 2$.}
\end{array}
\right. \\ \non
\end{align}
\qed

\end{lemma}

\subsection{Self-consistent perturbation equations}

Following \cite{EYY, EYY2}, we define the following quantities: 
\begin{align}
   A_i: & =\sigma^2_{ii}G_{ii}+\sum_{j\neq i}\sigma^2_{ij}\frac{G_{ij}G_{ji}}{G_{ii}}
\label{defA} \\
  Z_i:& =  \sum_{k,l\ne i} 
\Big[ \ov{\ba^i_k} G^{(i)}_{k\,l}\ba^i_l-\E_{\ba^i} \ov{\ba^i_k} 
G^{(i)}_{k\,l}\ba^i_l\Big]=Z_{ii}^{(i)}-\E_{\ba^i} Z_{ii}^{(i)}, \label{defZi} \\
\Upsilon_i:& =A_i +
\left(\wH^{(i)}_{ii}-\E_{\ba^i}\wH^{(i)}_{ii}\right) = A_i +h_{ii}-Z_i ,
\label{seeqerror}
\end{align}
where $\E_{\ba^i}$ indicates the expectation with respect to the
matrix elements in the $i$-th column.
Using \eqref{GiiHii} from 
Lemma \ref{basicIG}, we obtain  the following  system of self-consistent  equations for the
deviation from $m_{sc}$ of the diagonal matrix elements of the resolvent;
\be\label{1}
  v_i = G_{ii} - m_{sc}
= \frac{1}{-z- m_{sc}- \Big(\sum_{j}\sigma^2_{ij}v_j-\Upsilon_i\Big)}
  - m_{sc} .
\ee
For the off-diagonal terms, we will use the equation \eqref{GijHij}. 
All the quantities defined so far depend on the spectral parameter $z=E+i\eta$, but
we will mostly omit this fact from the notation.

The key  quantities $\Lambda$, $\Lambda_d$ and $\Lambda_o$ \eqref{defLambda}
 appearing in Theorem \ref{thm:detailed} 
will be typically small and  we will prove in this section that
 their size is less than $( N \eta)^{-1/3}$,
modulo logarithmic  corrections. We thus define the exceptional (bad) event
\be\label{B}
 {\bf B} ={\bf B}(z):= \Big\{ \Lambda_d(z) + \Lambda_o(z)
\ge   (\log N)^{-2} \Big\}.
\ee
We will always work in the complement set
${\bf B}^c$, i.e.,  we will have
\be\label{15}
   \Lambda_d(z) + \Lambda_o(z)
\le (\log N)^{-2}.
\ee
We collect some basic properties of the Green function in the  following elementary lemma.

\begin{lemma}\label{62}
Let $\T$ be a subset of  $\{1, \ldots, N \}$.
 Then there exists a
constant $C=C_\T$  depending 
on $C_0$ from \eqref{1.3} and on $|\T|$,  the cardinality of $\T$,
 such that the following hold in ${\bf B}^c$ 
\begin{align}
 |G_{kk}^{(\bT)}-m_{sc}|\le & \;  \Lambda_d +  C\Lambda_o^2
 \qquad \mbox{for all $k\not \in \bT $,}
\label{Gkkm} \\
  \frac{1}{ C }\le & \; |G_{kk}^{(\bT)}|\le C  \qquad \mbox{for all 
$k\not\in \bT$,}
\label{Gkk} \\
 \max_{k\ne l} |G_{k\,l}^{(\bT)}|\le &\; C\Lambda_o,  \qquad
\label{41} \\
  \max_i |A_i|\le &\; \frac{C}{N} + C\Lambda^2_o  
\label{Aest}
\end{align}
for any fixed $|\T|$ and
for any sufficiently large $N$. We recall that all quantities
depend on the spectral parameter $z$ and
the estimates are uniform in  $z=E+i\eta$
as long as $|E|\le 5$, $0<\eta\le 10$.
\end{lemma}

{\it Proof.} For $\T=\emptyset$, the estimates \eqref{Gkkm} and \eqref{41}
follow directly from the definitions \eqref{defLambda}. The bound \eqref{Gkk}
follows from \eqref{15} and that $|m_{sc}(z)|\sim 1$, see \eqref{zmsc2}.
Finally, \eqref{Aest} follows from inserting \eqref{Gkk}, \eqref{41}, \eqref{sum}
and \eqref{1.3} into \eqref{defA}.
The general case can be proved by induction on $|\bT|$ and using the 
formulas \eqref{GiiGjii} and \eqref{GijGkij} that guarantee that
\be
    |G^{(\T)}_{k\ell} - G^{(\T')}_{k\ell}|\le C^* \Lambda_o^2
\label{lo}
\ee
holds for any $\T' = \T \cup \{m\}$, where $C^*$ depends on the
constant $C_\T$ for the induction hypothesis. In the set ${\bf B}^c$ and
for sufficiently large $N$,
depending on $|\T|$, the estimate \eqref{lo} together with 
$|m_{sc}(z)|\sim 1$ guarantees that the lower bound in \eqref{Gkk}
continues to hold for $\T'$. The other estimates for $\T'$ follow
from \eqref{lo} directly. \qed

\subsection{Estimate of the exceptional events}\label{sec:excep}

The following lemma is a modification of Lemma 4.5 in \cite{EYY2}. It improves the 
estimate in the sense that 
 the control parameter  depends only on $\Lambda$ but not on
$\Lambda_d$ and $\Lambda_o$ (see \eqref{defLambda} for definitions). Since $\Lambda$, being an
average quantity, behaves better, this  yields
a stronger estimate.  

For any $\ell>0$ we
define the  key control parameter $\Psi$,  which is  random variable, by
\be\label{1.9}
\,\,\, \Psi(z):= (\log N)^\ell  \sqrt{\frac{\Lambda(z)+\im m_{sc}(z) }{N\eta}}.
\ee
We also define the  events 
\begin{align} \label{eq:excep}
\Omega_h &:= \left \{  \max_{1\leq i,j\leq N}|h_{ij}|\geq
(\log N)^{\ell/10} |\sigma_{ij}| \right \} \cup \left\{ \Big| \sum_{i=1}^N h_{ii}\Big|\ge (\log N)^{\ell/10}\right\}
     \\
\Omega_d(z) &:= \left \{   \max_{i}|Z_i(z)|\ge  \frac{1}{2}\Psi (z)\right \} 
 \non \\
\Omega_o(z) &:= \left \{    \max_{i\ne j}|Z_{ij}^{(ij)}(z)| \ge 
\frac{1}{2} \Psi (z)\right \}, \non  
\end{align}
and we let 
\be
  \Omega (z):= 
\Omega_h \cup \Big[\big( \Omega_d(z) \cup \Omega_o(z) \big)\cap {\bf B}(z)^c \Big]
\label{defOmega}
\ee 
 be the set of  exceptional events. 
These definitions depend on the parameter $\ell$ that we omit from the notation.

 The main reason that $\Psi$ emerges as the key controlling parameter can be seen from the following 
consideration.  In order to estimate  the off-diagonal term $G_{ij}$,   
we need to bound \eqref{GijHij} $K_{ij}$ and thus $Z_{ij}$. 
By the large deviation estimate,  \eqref{resgenHWTO}, we have 
\be\label{z1}
|Z_{ij}^{(ij)}|
 \leq   C(\log N)^{\ell/3}  \sqrt{\sum_{k,l\not= i, j }\left|\sigma_{ik}
G^{(ij)}_{k\,l}\sigma_{lj}\right|^2} \le C(\log N)^{\ell/3}  \sqrt{   \frac  1 { N^2} \sum_{k,l\not= i, j }  \left| 
G^{(ij)}_{k\,l}   \right|^2}
\ee
holds with  high probability.  Here we have used that 
$\sigma_{il}^2 \le C_0/ N$ from \eqref{1.3}.

For any   normal matrix $A$, we have 
\be\label{wid1}
\sum_{j} |A_{ij}|^2 =  (A A^*)_{ii} =  (|A|^2)_{ii} 
\ee
where $|A|^2:= A A^*$. 

Applying this identity   to the Green function $G = [H-z]^{-1}$, we obtain 
the following ``Ward identity": 
\begin{align} \label{wid}
\sum_{l }\left| 
G_{k\,l} \right|^2 = \sum_{\alpha}
\frac{|u_\alpha(k)|^2}{|\lambda_\alpha-z|^2} =   \frac{\im G_{kk}}{\eta},
\end{align}
where $u_\alpha$ and $\lambda_\al$ are 
 the eigenvectors and eigenvalues of $H$.
The term "Ward identity" comes from quantum field theory and it represents an identity 
derived from a conservation law or symmetry of a system. In our case, the symmetry is generated by 
the global phase multiplication  $e^{i \theta}$, but this connection is not important for our purpose.

Applying \eqref{wid} to estimate the last term in \eqref{z1} and {\it neglecting the superscript $(ij)$}, we can 
bound $ Z_{ij}^{(ij)}$ by
$$
|Z_{ij}^{(ij)}(z)|
 \leq  C(\log N)^{\ell/3}    \sqrt{\frac{  N^{-1} \sum_k  \im G_{kk} }{N \eta}  }
 \leq  C(\log N)^{\ell/3}    \sqrt{\frac{ \Lambda(z) +  {\im m_{sc}(z) }}{N \eta}  }
$$
where we have used the definition of $\Lambda$ in the last inequality. Notice that the control parameter 
$\Psi$ appears naturally in this estimate. Furthermore, it is $\im m_{sc}(z)$ which appears in the numerator, not 
$m_{sc}(z)$. This is the fundamental reason that we are able to obtain optimal estimate up to the edges 
of the spectrum. Near the edges, $\im m_{sc}(z)$  is small while $|m_{sc}(z)|$  stays near $1$.

\begin{lemma}\label{selfeq1} 
 There exist a constant $0<\phi < 1$, depending on 
$\ttau$ \eqref{subexp}, and universal  constants  $C>1$, $c>0$,
such that 
 for any $\ell$ with 
$4/\phi \le \ell \le  C \log N/\log \log N $ and for any $z\in {\bf S}_\ell $ 
 we have 
\be\label{B12}
\P (\Omega (z))    \le  C\exp{\big[-c(\log N)^{\phi \ell } \big]},
\ee
and we also have the pointwise statement 
 \be\label{21}
(\log N)^{\ell/2}\Lambda_o(z) + \max_i  |\Upsilon_i(z)|
 \le   \Psi(z)     \quad \text {in}  \;\;\;  \Omega (z)^c \cap  {\bf B}(z)^c
\ee
for any sufficiently large $N\ge N_0(\ttau)$. 
\end{lemma}

\bigskip

{\it Proof.} There exists $0 <\phi< 1$,  depending on $\ttau$,
such that the following two estimates hold
for any $\ell \ge 4/\phi$: 
$$
   \P \left \{   |h_{ij}|\geq
(\log N)^{\ell/10} |\sigma_{ij}|  \right \}  \le  C\exp{\big[-(\log N)^{\phi\ell} \big]},
 \qquad\forall i, j
$$
 by \eqref{subexp}, and
$$
  \P \left\{ \Big| \sum_{i=1}^N h_{ii}\Big|\ge (\log N)^{\ell/10}  \right\} 
 \le    C\exp{\big[-(\log N)^{\phi \ell } \big]} 
$$
by \eqref{1.3} and  the large deviation principle for the sum of independent random
variables  (e.g., \eqref{resgenHWTD}).  Thus
\be\label{resboundhij}
\P\left(\Omega_h \right)\leq   C\exp{\big[-c(\log N)^{\phi \ell } \big]},
\ee
so we can work on the complement set $\Omega_h^c$.  Note that  
\be
    \Omega^c \cap  {\bf B}^c = \Omega_h^c \cap \Omega_d^c\cap \Omega_o^c\cap {\bf B}^c.
\label{capp}
\ee

Fix $z\in \bS_\ell$ and we
will prove, possibly with a smaller $\phi$, that for $\ell \ge 4/\phi$ we have
\be\label{POdz}
\P\Big(\Omega_h^c\cap  \Omega_d(z) \cap {\bf B}^c(z) \Big)\leq
 C\exp{\big[-c(\log N)^{\phi \ell } \big]}
\ee
and
\be\label{Poz}
\P\Big(\Omega_h^c\cap\Omega_o(z)\cap {\bf B}^c(z) \Big)\leq 
C\exp{\big[-c(\log N)^{\phi \ell }\big]},
\ee
and this will prove \eqref{B12}.

To prove the diagonal estimate \eqref{POdz}, we can choose 
a sufficiently small $\phi>0$
(depending on $\ttau$) and 
apply the large deviation bound \eqref{diaglde} from 
Lemma \ref{generalHWT} to obtain that 
for any fixed $i$
 \be\label{BI23}
|Z_{i}|
 \leq  (\log N)^{\ell/3} \sqrt{\sum_{k,l\not= i }\left|\sigma_{ik}
G^{(i)}_{k\,l}\sigma_{li}\right|^2}
\ee
holds with a probability larger than $1- C\exp{\big[-c(\log N)^{\phi \ell } \big]}$ 
for sufficiently large $N$.
{F}rom the Ward identity \eqref{wid} and $\sigma_{il}^2 \le C_0/ N$  (by \eqref{1.3} and \eqref{Gkkm}), 
we have 
\begin{align} \label{B14}
\sum_{k, l\not= i }\left|\sigma_{ik}
G^{(i)}_{k\,l}\sigma_{li}\right|^2 
\le 
\frac{C_0}{ N} \sum_{k\not= i } \frac{\im G^{(i\,)}_{kk}}{N \eta}.
\end{align}
Since we are in the set ${\bf B}^c$, we have $ \Lambda_d + \Lambda_o
\le (\log N)^{-2}$. Thus from \eqref{GiiGjii} and \eqref{Gkk} we have  that
\be\label{5}
 0 <   \im G_{kk}^{(i)} \le  \im G_{kk}+ |G_{kk}^{(i)} - G_{kk}|\le 
\im G_{kk} +C   |G_{ik}|^2 \le  \im G_{kk} + C \Lambda_o^2.
\ee
The last term of \eqref{B14} is bounded by 
\begin{align} \label{B14y}
\frac {C_0^2} { N^2 }\sum_{k\not= i } \frac{\im G^{(i\,)}_{kk}}{\eta}
& \le C 
\frac{ \Lambda+ \Lambda_o^2  +  {\im m_{sc} }}{N \eta} 
\qquad \mbox{in ${\bf B}^c$}.
\end{align}

We have thus proved that for any $z\in {\bf S}_\ell$
 \be\label{BI24y}
|Z_{i}(z)|
 \leq C(\log N)^{\ell/3}  \sqrt{\frac{ \Lambda (z)+ \Lambda_o^2(z)  +  {\im m_{sc}(z) }}{N \eta}  }
 \qquad \mbox{in ${\bf B}^c(z)$}.
\ee
holds with a probability larger than $1- C\exp{\big[-c(\log N)^{\phi \ell } \big]}$ for sufficiently large $N$.

Similarly, for the off-diagonal estimate \eqref{Poz}, for any fixed   $i\ne j$,
 we have from \eqref{resgenHWTO} that 
 \be\label{BI244}
|Z_{ij}^{(ij)}|
 \leq   C(\log N)^{\ell/3}  \sqrt{\sum_{k,l\not= i, j }\left|\sigma_{ik}
G^{(ij)}_{k\,l}\sigma_{lj}\right|^2}
\ee
holds with a probability larger than  $1- C\exp{\big[-c(\log N)^{\phi \ell } \big]}$
 for sufficiently large $N$.
Similarly to the proof of \eqref{BI24y} for $Z_i$, we have 
 \be\label{BI244y}
|Z_{ij}^{(ij)}(z)|
 \leq  C(\log N)^{\ell/3}    \sqrt{\frac{ \Lambda(z)+ \Lambda_o^2(z)  +  {\im m_{sc}(z) }}{N \eta}  }
 \qquad \mbox{in ${\bf B}^c(z)$}
\ee
holds for any $z\in {\bf S}_\ell$
with a probability larger than $1- C\exp{\big[-c(\log N)^{\phi \ell } \big]}$ for sufficiently large $N$.

\medskip

Using Lemma \ref{62}, we   have $|G_{ii} |\le C$ and $|G_{jj}^{(i)}| \le C$
in the set ${\bf B}^c$. {F}rom \eqref{GijHij},
we can thus estimate  the off-diagonal term $G_{ij}$ by 
\be\label{6.4}
  |G_{ij}|  = |G_{ii}| |G_{jj}^{(i)}| |K_{ij}^{(ij)}|
  \le C\left(|h_{ij}| + |Z_{ij}^{(ij)}|\right), \qquad i\ne j, \qquad  \mbox{in ${\bf B}^c$}.
\ee
Hence we have that in the event ${\bf B}^c\cap\Omega_h^c$
\be
\Lambda_o = \max_{i \not = j} |G_{ij}|\le \frac{C(\log N)^{\ell/10} }{\sqrt{N}}+C(\log N)^{\ell/3}  \sqrt { 
\frac{ \Lambda+ \Lambda_o^2 +  {\im m_{sc}}}{N \eta}} 
\label{gij1}
\ee
holds with a probability larger than  $1- C\exp{\big[-c(\log N)^{\phi \ell } \big]}$ for sufficiently large $N$.

Recall that $N\eta\ge (\log N)^{10 \ell}$ on the set $\bS_\ell $ and
since  $\ell \ge 4/\phi\ge 4$, we have $(\log N)^{\ell/3}\ll \sqrt{N\eta}$, thus
the $\Lambda_o$ term on the right hand side of \eqref{gij1} can be absorbed into
 the left side for sufficiently large $N$.
Furthermore, 
by \eqref{esmallfake}, we have  $\im m_{sc}(z) \ge c\eta$ with
a universal positive constant $c$ for any $z\in \bS_\ell$.
 Thus   the first term on the right hand side of \eqref{gij1}  can be bounded by 
\[
\frac{C(\log N)^{\ell/10}  }{\sqrt{N}} \le    (\log N)^{\ell/3}  \sqrt { 
\frac{   {\im m_{sc}(z) }}{N \eta}}
\]
for large enough $N$, and thus it can be absorbed into  the second  term. We conclude that
\be
   \P \Big\{ \Lambda_o\le C(\log N)^{-2\ell/3}  \Psi, \quad {\bf B}^c\cap\Omega_h^c \Big\} 
  \ge  1- C\exp{\big[-c(\log N)^{\phi \ell } \big]}.
\label{conl}
\ee
  Inserting this bound into \eqref{BI24y} and
\eqref{BI244y}, 
 we have proved  \eqref{POdz} and \eqref{Poz}.
Finally, the estimate \eqref{21} 
for $\Upsilon$ and $\Lambda_o$ is a simple consequence of 
\eqref{conl}, the definition \eqref{seeqerror}, the
bound \eqref{Aest}, the
definition of $\Omega_d$  and that $\Omega^c\cap {\bf B}^c\subset \Omega_h^c$.
This completes the proof of Lemma \ref{selfeq1}.
\qed 

\subsection{Analysis of the self-consistent equation}\label{sec:selfcons}

Now we start using the self-consistent equation \eqref{1}. 
Since 
$$
    \Big| \sum_j \sigma_{ij}^2 v_j - \Upsilon_i\Big| \le \Lambda_d + |\Upsilon_i|,
$$
the bound \eqref{zmsc2} allows us to expand the denominator in 
\eqref{1} as long as $ \Lambda_d + |\Upsilon_i| \le \frac{1}{2}$. 
In this case,
using \eqref{defmsc},  we obtain the following equation for $v_i$
\be\label{1.18}
 v_i =  m^2_{sc} \Big(\sum_{j}\sigma^2_{ij}v_j-\Upsilon_i\Big)+
m^3_{sc}\Big(\sum_{j}\sigma^2_{ij}v_j-\Upsilon_i\Big)^2+
  O \Big(\sum_{j}\sigma^2_{ij}v_j-\Upsilon_i\Big)^3 .
\ee
Recall that $B$ denotes the $N\times N$ matrix of covariances, $B= (\sigma_{ij}^2)$.
Thus we can rewrite the last equation as 
$$
 [(1 -  m^2_{sc} B) \bv ]_i =   - m^2_{sc} \Upsilon_i +
m^3_{sc}\Big( (B  \bv)_i-\Upsilon_i\Big)^2+
  O \Big( ( B \bv)_i -\Upsilon_i\Big)^3 .
$$ 
We will first use this equation  to estimate $v_i -\barv$, i.e. the deviation
of $v_i$ from its average (Lemma \ref{59}). In the second step,  we will 
add up \eqref{1.18} for all $i$ and obtain an equation
for $\barv$ (Lemma  \ref{591}). Finally, we use a dichotomy
argument to estimate $\Lambda= |\barv|$ in Lemma \ref{592}.

\medskip

{By normalization assumption $\sum_j \sigma^2_{ij}= 1$,    the vector ${\bf e}=(1,1,\ldots , 1)$
is the (unique) eigenvector of $B$ with  eigenvalue 1. }
We  introduce the notation
\be\label{defq}
 q=q(z):=\max\{\delta_+, |1-\re m^2_{sc}(z)|\},
\ee
and we recall  the following elementary lemma that was proven in \cite[Lemma 4.8]{EYY2}.

\begin{lemma}\label{infinf} The matrix $\mathbb I-m_{sc}^2(z)B$ is
invertible on the subspace orthogonal to ${\bf e}$. 
 Let  $\bf u$ be a vector which is  orthogonal to ${\bf e}$ and let
$$
   \bw = (\mathbb I-m_{sc}^2(z)B)\bu,
$$
then 
$$
\|{\bf u}\|_\infty\leq \frac{C\log N}{q(z)}\|{\bf w}\|_\infty
$$
for some constant $C$ that only depends on $\delta_-$ in \eqref{speccond}. \qed
\end{lemma}

The following lemma estimates the deviation of $v_i$ from its average $\barv$:

\begin{lemma}\label{59}  Suppose that $4\le \ell\le C\log N/\log\log N$.
 Fix  the spectral parameter $z\in \bS_\ell$ and
 we will omit it from the notations.
  Suppose that in some set $\Xi$ it holds that
 \be\label{342}
 \Lambda_d\le \frac{q}{(\log N)^{3/2}},
 \ee
 then  in the set  $\Xi\cap\Omega^c \cap {\bf B}^c$  we have  
\be\label{llld2}
\max_i\left|v_i-\barv\right|\leq 
 \frac{C \log N}{q}
\left(\Lambda^2 +  \Psi+ \frac{(\log N)^2}{q^2}\Psi^2\right)
\le  \frac{C \log N}{q^3}
\left(\Lambda^2 +  \Psi\right)
\ee
for some constant $C$  depending only on $\delta_-$
and for sufficiently large $N$.
\end{lemma}

\medskip

{\it Proof.}  For $z\in \bS_\ell$, $q(z)$ and $\im m_{sc}(z)$ are bounded.
Combining \eqref{342} with the definitions of $\Psi(z)$, $\bS_\ell$ and with $\ell \geq 4$, we obtain  
that $\Lambda_d(z)$, $\Lambda(z)$ are bounded by $C(\log N)^{-3/2}$
and $\Psi(z)$  is bounded by $C(\log N)^{-2}$.
Thus the expansion \eqref{1.18} holds true in the set $\Xi\cap\Omega^c\cap {\bf B}^c$,
by using \eqref{21}. We can estimate
the second and third order terms in \eqref{1.18} by $C(\Psi+\Lambda_d)^2$
and we obtain
\begin{align}\label{344}
v_i & =  m_{sc}^2 \sum_{j}\sigma^2_{ij}v_j+ \e_i, \qquad
\mbox{with} \quad     
\e_i= O(\Psi) + O(\Lambda_d^2) \quad \text { in  } \; \Xi\cap\Omega^c \cap {\bf B}^c .
\end{align}
Taking the average over $i$, we have
$$
 (1-m_{sc})\barv = \frac{1}{N}\sum_i \e_i =  O(\Psi) + O(\Lambda_d^2),
$$
and thus it follows from \eqref{344} that
$$
    v_i - \barv = m_{sc}^2\sum_j \sigma^2_{ij} (v_j -\barv ) + O(\Psi) + O(\Lambda_d^2).
$$
 Applying Lemma \ref{infinf} for $u_i=v_i-\barv$, we obtain 
\be\label{347}
\max_i\left|v_i-\barv\right|\leq \frac{C \log N}{q} \left(\Lambda_d^2 +       \Psi\right) ,
\ee
hence
$$
   \Lambda_d \le \Lambda + \frac{C \log N}{q} \left(\Lambda_d^2 +       \Psi\right).
$$
With \eqref{342}, this inequality implies 
\be\label{71}
\Lambda_d\leq \Lambda + \frac{C \log N}{q} (\Lambda^2+\Psi).
\ee
Using \eqref{71} to bound $\Lambda_d^2$ in \eqref{347},  we have  proved  the
first inequality of \eqref {llld2}, the second one follows from $\Psi\le C(\log N)^{-2}$. 
This completes the proof of Lemma \ref{59}. \qed

\medskip

In this paper we assumed that the positive constants $\delta_\pm$ are independent
of $N$ (see \eqref{speccond}), thus $q$ is bounded and the  condition
\eqref{342} is automatically satisfied in the set ${\bf B}^c$, see \eqref{15}, 
and  therefore \eqref{llld2} can be written as
\be\label{llld}
\max_i\left|v_i-\barv\right|\leq C(\log N) \left(\Lambda^2 +       \Psi\right) 
 \qquad \mbox{in $\Omega^c \cap {\bf B}^c$},
\ee
in particular,
\be\label{71d}
\Lambda_d\leq \Lambda + C (\log N) (\Lambda^2+\Psi) \qquad \mbox{in $\Omega^c \cap {\bf B}^c$},
\ee
with some constant $C$ depending only on $\delta_\pm$.

\medskip

\begin{lemma}\label{591}     Suppose that $4\le \ell\le C\log N/\log\log N$.
  Fix  the spectral parameter $z\in \bS_\ell$ and
 we will omit it from the notations.
Then in the set $\Omega^c \cap {\bf B}^c$  we have 
\be\label{661}
(1- m^2_{sc}) \barv =  m_{sc}^3 {\barv}^2+  m^2_{sc}\barZ   
 +  O\Big( \frac{\Lambda^2}{\log N}\Big)+   O\Big( (\log N)\Psi^2 \Big),
\ee
where $\barZ: = N^{-1} \sum_{i=1}^N Z_i$. The implicit constants in 
the error terms depend only on $\delta_\pm$ and $C_0$.
\end{lemma}

\medskip 

{\it Proof.} {F}rom the choice   $\ell \ge 4$, and
from $\Lambda  \le (\log N)^{-2}$ in the set ${\bf B}^c$, we have
\be\label{95}
 \Psi  \le (\log N)^{-8} .
\ee
Moreover, for $z\in \bS_\ell$, we have $\im m_{sc}(z)\ge c\eta$ with
some universal positive constant $c$ (see Lemma \ref{lm:msc}), we also have
\be
  \Psi\ge \frac{(\log N)^\ell}{\sqrt{N}}.
\label{Psibelow}
\ee
By the definition of $\Upsilon_i$ \eqref{seeqerror}, by the estimates \eqref{Aest} and 
 \eqref{21},  we have
\be\label{511}
\Upsilon_i= A_i +h_{ii}-Z_i = h_{ii} - Z_i + O(\Lambda_o^2 + N^{-1} ) = h_{ii}- Z_i + O(\Psi^2 )
\qquad \mbox{in $\Omega^c\cap {\bf B}^c$}.
\ee
The size of the last term of \eqref{1.18} is less than
 $O(    \Psi^3+\Lambda_d^3)$ which is bounded by $O(    \Psi^2+\Lambda^3)$
using \eqref{71d} and \eqref{95}.
Thus we have,  from \eqref{21} and \eqref{1.18},
\be
v_i = m^2_{sc} \Big(\sum_{j}\sigma^2_{ij}v_j + Z_i -h_{ii}+ 
O(\Psi^2 ) \Big)+m^3_{sc} \Big(\sum_{j}\sigma^2_{ij}v_j+  O(\Psi )  \Big)^2+
   O \Big(     \Psi^2+\Lambda^3\Big)\qquad \mbox{in $\Omega^c\cap {\bf B}^c$}.
\ee
Summing  up $i$ and dividing by $N$, we obtain
\be\label{66}
 \barv =m^2_{sc}\barv + m^2_{sc}\barZ   + O(      \Psi^2+  \Lambda^3)
+\frac{m_{sc}^3}N\sum_i\bigg(\sum_{j}\sigma^2_{ij}v_j+  O(\Psi )  \bigg )^2  
\qquad \mbox{in $\Omega^c\cap {\bf B}^c$}.
\ee
Here we used that in  the set $\Omega^c\cap {\bf B}^c\subset\Omega_h^c$,
we have $N^{-1}|\sum_i h_{ii}|\le (\log N)^{\ell/10}N^{-1} \le \Psi^2$ by \eqref{Psibelow}.
Writing $v_j = (v_j -\barv)+\barv$, the last term in \eqref{66} can be estimated
using \eqref{llld}
$$ 
\frac{m_{sc}^3}N\sum_i\bigg(\sum_{j}\sigma^2_{ij}v_j+  O(\Psi )  \bigg )^2 
  = m_{sc}^3 {\barv}^2  + O\big( (\log N)\Psi(\Lambda^2 +\Psi)\big)+ O(\Lambda\Psi)+
 O(\Psi^2). 
$$
Collecting the various error terms and using \eqref{95} and
that  $\Lambda\le (\log N)^{-2}$ in
${\bf B}^c$,  we obtain \eqref{661} from \eqref{66}.
This completes the proof of  Lemma \ref{591}.  \qed

\subsection{Dichotomy estimate for $\Lambda$}

 Throughout this section we fix the parameter $\ell$ with 
$4\le \ell \le C\log N/\log\log N$.
By Lemma \ref{591} we have that in  $\Omega^c\cap {\bf B}^c$
\begin{align}\label{m-m2}
(1-m^2_{sc}) \barv - m_{sc}^3 \barv^2  =  O(     \Psi )+  O(\Lambda^2)/\log N,
\end{align}
where we have used the simple bound $\Psi \le 1/\log N$ and
that in the set $\Omega(z)^c \cap {\bf B}(z)^c$ all $Z_i$, hence $\barZ$ can be bounded
by $\Psi$ (see \eqref{capp} and the definition of $\Omega_d$).

We introduce the following notations:
\be\label{58}
 \alpha  := \Big|\frac { 1-m^2_{sc}}{m_{sc}^{3}}\Big|
\qquad  \beta : =  \frac{ (\log N)^{2 \ell} }{ (N \eta)^{1/3}}, \quad
\mbox{with}\quad \eta =\im z, 
\ee
where $\al=\al(z)$ and $ \beta= \beta(z)$ depend on the spectral parameter $z$.
For
 any $z\in \bS_\ell$ we have the bound 
$\beta(z) \le (\log N)^{-4}$,
by $\ell \ge 4$.
{F}rom Lemma \ref{lm:msc} it also follows that there is a universal
constant $K\ge1$ such that
\be
   \frac{1}{K}\sqrt{\kappa+\eta}\le  \alpha(z)\le K\sqrt{\kappa+\eta}
\label{Kdef}
\ee
for any $z\in \bS_\ell$.

By definition of  $\Psi=\Psi(z)$ \eqref{1.9}, we have 
\begin{align}
    \Psi = &    (\log N)^\ell   \sqrt{\frac{\Lambda +\im m_{sc} }{N\eta}}    \non \\
\le &    (\log N)^\ell     {\frac{\Lambda +\im m_{sc} }{(N\eta)^{1/3}}} +    
 (\log N)^\ell    (N \eta)^{-2/3} \le \beta \Lambda + \alpha \beta + \beta^2, 
\end{align}
where, in the last step, we have
 used that  $\alpha(z) \sim \sqrt{\kappa+\eta}$, see \eqref{Kdef}, and thus
$ \im m_{sc}(z)   \le C \alpha(z)$ (see Lemma \ref{lm:msc}).
We conclude from \eqref{m-m2} and $|m_{sc}|\sim 1$ that
\be
  \Big| \frac{1-m^2_{sc}}{m_{sc}^3} \barv - \barv^2\Big|  \le
 C^*\big(\beta \Lambda + \alpha \beta + \beta^2\big) +
   O(\Lambda^2)/\log N \qquad \mbox{in $\Omega^c\cap {\bf B}^c$}
\label{dich}
\ee
with some  constant $C^*$.   

Neglecting the error term and replacing $\barv$ by $\Lambda$,  we roughly have the equation 
\be
 \Big| \alpha \Lambda  -\Lambda^2\Big|  \le
 C^*\big(\beta \Lambda + \alpha \beta + \beta^2\big).
\label{dichint}
\ee
This inequality provides certain estimates on $\Lambda$ depending on whether $\al \lesssim \beta$ or not.

Since $\alpha$ and $\beta$ are functions of $z$ ($\beta(z)$
depends only on $\eta=\im z$), 
we will fix  $E=\re z$ and vary $\eta=\im z$ from $\eta=10$ down
to $\eta = (\log N)^{10\ell}/N$. 
{ Thanks to \eqref{Kdef}, $\al(z)$ is essentially monotone increasing
in $\eta$, up to universal constants. The function $\beta(z)$ is monotonically
decreasing. Therefore there exists a threshold $\wt\eta$ such that
for $\eta\le\wt\eta$ we have $\alpha \lesssim \beta$
and for  $\eta\ge\wt\eta$ we have $\alpha \gtrsim \beta$.
To implement precisely the idea of dividing the estimate 
according to the relative size of $\alpha$ and  $\beta$, 
we will need to choose a large but fixed constant 
 $U>1$ depending only on $C^*$. Let $\wt\eta=\wt\eta(U, E)$ be the solution to 
$\sqrt{\kappa+\eta}= 2 U^2K \beta(z) $ where  $\kappa = \big| |E|-2\big|$.
Note that up to a constant factor, this equation is the
same as $\alpha(z)=\beta(z)$. } Since  $\sqrt{\kappa+\eta}$ is increasing
 while $\beta(z) $ is decreasing in $\eta$, 
the solution is unique and  one can easily prove that 
\be\label{teta}
\wt\eta \le N^{-1/3}
\ee
 for sufficiently large $N$,
depending on $U$. 
The   implementation of this idea and  precise estimates on $\Lambda$ is given by the following Lemma:

\medskip 
\begin{lemma}\label{592}{\bf [Dichotomy Lemma]}    Suppose that 
$4\le \ell\le C\log N/\log\log N$. 
Then there is a  constant $U_0= U_0(\delta_\pm, C_0)\ge 1$
such that
for any $U\ge U_0$,  there exists a constant $C_1(U)$, depending only on $U$,
 such that  for any spectral parameter $z\in \bS_\ell$ the following 
estimates hold
\begin{align}\label{81}
\Lambda(z) & \le   U  \beta(z)   \quad \text{ or }  \quad  \Lambda(z)  \ge  
 \frac {\alpha(z)}{  U }     & \text{ if } \;\; \im z\ge \wt\eta(U,\re z)\\
\Lambda(z) & \le C_1(U) \beta(z)     & \text{ if } \;\;  \im z< \wt\eta(U,\re z)
 \label{82}
\end{align}
 in the set $\Omega(z)^c \cap {\bf B}(z)^c$ and for any sufficiently
large $N\ge N_0(\delta_\pm, C_0)$.

\end{lemma}

{\it Proof.}  We will set $U_0= 9(C^*+1)$ and let $U\ge U_0$ where $C^*$ is the constant appearing in \eqref{dich}.
Depending on the relative size of $\beta$ and $\al$, 
which is determined by $z$, we will 
either express $\barv$ or $\barv^2$ from \eqref{dich}.
This will correspond to the two cases in Lemma \ref{592}.
Recalling that $|\barv|= \Lambda$, the last error term
in \eqref{dich} can be easily absorbed for sufficiently large $N$
 and we will
get a quadratic inequality for $\Lambda$.

\medskip

\noindent
{\it Case 1:  $\eta=\im z \ge \wt\eta(U,E)$.}
By the definition of $\wt\eta$, in this case $\sqrt{\kappa+\eta}\ge 2U^2K\beta(z)$, i.e.,
\be\label{newalz}
\al(z)\ge 2U^2\beta(z)
\ee by \eqref{Kdef}.
 {F}rom the choice of $U_0$ and $U\ge U_0$ we get
that $\alpha  \ge \beta$
and $\frac{1}{2}\al\ge C^*\beta$. Expressing
 $\barv$ from \eqref{dich} and absorbing 
the $C^*\beta\Lambda$ term into the left hand side,
 we obtain
\be\label{13}
\frac{1}{2} \alpha  \Lambda  
 \le   2  \Lambda^2     +  2C^*  \alpha\beta .
\ee
Thus either
$$
   \frac{1}{4} \alpha  \Lambda  
 \le   2  \Lambda^2, 
$$
i.e. $\Lambda\ge \al/8$ which is larger than $\al/U$, or 
$$
   \frac{1}{4} \alpha  \Lambda  
 \le    2C^*  \alpha\beta,
$$
i.e. $\Lambda \le 8C^*\beta \le U\beta$, which proves  \eqref{81}.

\medskip
\noindent
{\it Case 2:  $\eta=\im z < \wt\eta(U,E)$.}
In this case $\sqrt{\kappa+\eta}\le 2U^2K\beta(z)$, i.e.,
$\al(z)\le 2U^2K^2\beta(z)$.
We  express $\barv^2$ from
\eqref{dich} and we get
\be\label{13-2}
 \Lambda^2   \le   2 \alpha \Lambda      + 
 2 C^* \big [ \,  \beta \Lambda  + \beta \alpha + \beta^2  \, \big ]   
 \le  C'   \beta \Lambda  + C'  \beta^2   
\ee
with a constant $C'$ depending on $U$.
This quadratic inequality immediately implies that $\Lambda\le C_1(U)\beta$
with some $U$-dependent constant $C_1(U)$.
Hence we have proved   Lemma \ref{592}.
 
\qed

\subsection{Initial estimates for large $\eta$}

In this section we show that Theorem \ref{thm:detailed}
holds for $\eta=\im z = 10$, i.e. on the upper boundary of
$\bS_\ell$. This will serve as an initial step for the
continuity argument. The proof for $\eta=10$ is similar
to the arguments in Sections \ref{sec:excep} and \ref{sec:selfcons}
but much easier. In particular, no apriori
assumption similar to  \eqref{15} or no bad set ${\bf B}$
are  necessary.
We start with the analogue of Lemma \ref{selfeq1} which actually holds
uniformly for any $z$ with $0< \eta=\im z\le 10$
 and not only for $z\in \bS_\ell$.  Note that these estimates
are very weak for small $\eta$, but we will use 
them only for $\eta=10$.

\begin{lemma}\label{selfeq1-1} 
For any $z\in \bC$ with $0< \eta=\im z\le 10$, define  the exceptional events 
\begin{align} 
\Theta_d(z) &:= \left \{   \max_{i}|Z_i(z)|\ge  \frac {(\log N)^\ell  } {\sqrt{N} \eta} \right \} 
 \non \\
 \Theta_o(z) &:= \left \{    \max_{i\ne j}|Z_{ij}^{(ij)}(z)| 
\ge\frac {(\log N)^\ell }  {\sqrt{N} \eta} \right \} \non  \\
   \Theta (z) &:= 
 \Omega_h \cup \Theta_d(z) \cup \Theta_o(z), 
\label{defTheta}
\end{align}
where we recall the definition of $\Omega_h$ in \eqref{eq:excep}.
Then   there exists constants $0<\phi < 1$, $C>1$, $c>0$,  depending on 
$\ttau$ \eqref{subexp},
such that  for any $\ell$ with $4/\phi \le \ell \le  C \log N/\log \log N $
 and for any $z\in {\bf S}_\ell $  we have 
\be\label{B12-1}
\P (  \Theta (z)  )    \le C\exp{\big[-c(\log N)^{\phi \ell } \big]},
\ee
and the pointwise bound 
\be\label{21-1}
\max_i  |\Upsilon_i(z)|
 \le CN^{-1/3}  \eta^{-3}    \quad \text {in}  \;\;\;   \Theta(z) ^c
\ee
for sufficiently large $N\ge N_0(\ttau, C_0)$. 
Furthermore, for $\eta \ge 3$ we have the estimate 
\be\label{12-3}
\Lambda_d(z) +\Lambda_o(z)\le CN^{-1/3}  \quad \text {in}  \;\;\;   \Theta(z) ^c.
\ee
for sufficiently large $N\ge N_0(\ttau, C_0)$.
\end{lemma}

{\it Proof.}  Given the estimate \eqref{resboundhij}, 
for the proof of \eqref{B12-1} 
 it is sufficient to estimate the probability of
$\Theta_d$ and $\Theta_o$. The estimate  \eqref{B14} still holds, but
  we can now  bound the last term in  \eqref{B14} simply  by 
\be\label{sim}
\sum_{k, l\not= i }\left|\sigma_{ik}
G^{(i)}_{k\,l}\sigma_{li}\right|^2\le
\frac1N \sum_{k\not= i } \sigma_{ik}^2 \frac{\im G^{(i\,)}_{kk}(z)}{\eta} 
\le \frac 1 { N \eta^2}, 
\ee
for  any  $z$,  using the trivial  deterministic estimate 
\be\label{trivial}
|G_{ij}^{(\T)}| \le  \eta^{-1} 
\ee
that holds for any $i,j$ and for any  $\T$.
Combining \eqref{sim} with
the large deviation bound \eqref{diaglde} from 
Lemma \ref{generalHWT} as in 
\eqref{BI23}, we obtain $\P (\Theta_d)\le C\exp{\big[-c(\log N)^{\phi \ell} \big]}.$
The same argument  holds for the exceptional set $\Theta_o$
involving the off-diagonal elements and this proves \eqref{B12-1}.

From \eqref{GijHij} and the trivial estimate \eqref{trivial},
we can estimate  the off-diagonal term $G_{ij}$ in 
 the set $  \Theta(z) ^c$ by 
\be\label{6.4-1}
  |G_{ij}|  = |G_{ii}| |G_{jj}^{(i)}| |K_{ij}^{(ij)}|
  \le \eta^{-2} \left(|h_{ij}| + |Z_{ij}^{(ij)}|\right)
 \le  (\log N)^\ell  \bigg [ \frac 1 { \sqrt N \eta^2}  +  
 \frac {1 } {\sqrt{N} \eta^3}  \,  \bigg ] 
\le   N^{-1/3}  \eta^{-3}   , \quad i\ne j,
\ee
for sufficiently large $N$.
Moreover, the same argument gives
$$
   \frac{|G_{ij}|}{|G_{ii}|} =  |G_{jj}^{(i)}| |K_{ij}^{(ij)}|\le
 N^{-1/3}\eta^{-2},  \quad i\ne j,
$$
which can be inserted in the definition of $A$, \eqref{defA}, and with $N\eta\gg 1$,
we get
$$
  |A_i| \le \frac{C_0}{N\eta}  + \frac{1}{ N^{1/3}\eta^3} \le\frac{2}{ N^{1/3}\eta^3}
$$
 for sufficiently large $N$.
In the set $\Theta^c$
a similar bound holds for $h_{ii}$ and $Z_i$ using $\eta\le 10$.
Recalling that $\Upsilon_i = A_i + h_{ii} - Z_i$, 
 and this proves  \eqref{21-1}.

For the proof of \eqref{12-3} it is sufficient to bound only 
$\Lambda_d$, the necessary estimate for $\Lambda_o$ is given in
\eqref{6.4-1}. 
We define  $\Upsilon=\max_i|\Upsilon_i|$
and  note that for $\eta\ge 3$ we have
$\Upsilon\le  CN^{-1/3}$ in the set $\Theta^c$ by  \eqref{21-1}.
{F}rom the self consistent equation \eqref{1} and the 
defining equation \eqref{defmsc} of $m_{sc}$, we have
\be\label{temp1.47}
v_n=\frac{\sum_{i}\sigma^2_{ni}v_i+O(\Upsilon)}{(z+m_{sc}+\sum_{i}\sigma^2_{ni}
v_i+O(\Upsilon))(z+m_{sc})}, \,\,\,1\leq n\leq N.
\ee
Using
 $|G_{ii}|\leq \eta^{-1}$ from \eqref{trivial}
 and $|m_{sc}(z) | =\big| \int \varrho_{sc}(x)/(x-z)\rd x|\leq \eta^{-1}$, we obtain for $\eta\ge 3$ that 
 \be\label{vileq2}
\Lambda_d=\max_i |v_i|\leq 2/\eta\leq 2/3.
 \ee 
By  \eqref{zmsc2}, we have  $|z+m_{sc}(z)| = |m_{sc}(z)|^{-1}\ge 3 $.
 Together with
  \eqref{vileq2}, we obtain from \eqref{temp1.47} that
\be
|v_n|\le\frac{\max_i|v_i|}{|z+m_{sc}(z)|-\max_i|v_i|}+O(\Upsilon).
\ee
Maximizing over $n$, we have  
\be\label{temp1.49}
\Lambda_d = \max_n |v_n|\leq \frac{\Lambda_d }{|z+m_{sc}|-\Lambda_d }+O(\Upsilon).
\ee
Since the  denominator satisfies $|z+m_{sc}(z)|-\Lambda_d\geq 3-2/3=7/3$
by $\Lambda_d\le 2/3$, the estimate \eqref{12-3} 
follows from \eqref{temp1.49}
and \eqref{21-1}. This completes the proof of Lemma \ref{selfeq1-1}.  
\qed

\subsection{Continuity  argument : conclusion of the
 proof of Theorem \ref{thm:detailed}} \label{continuity}

Fix an energy $E$ with $|E| \le 5$ and choose a decreasing
finite sequence $\eta_k\in \bS_\ell$,  $k=1,2,\ldots, k_0$, with $k_0\le CN^{8}$ such that  
$|\eta_k-\eta_{k+1} |\le N^{-8}$ and $\eta_1 = 10$, $\eta_{k_0}= N^{-1}(\log N)^{10 \ell }$.
Denote by $z_k= E + i \eta_k$. We will first show that  Theorem \ref{thm:detailed}
holds for any $z=z_k$.

Throughout this section fix any $U\ge U_0$ from Lemma \ref{592}
and recall the definition of $\wt\eta(U,E)$ from before this lemma.
Consider first the case of $z_1$. Since $\eta_1\ge \wt\eta(U,E)$, see \eqref{teta}, we are in the
first case \eqref{81} in  Lemma \ref{592}.
By Lemma \ref{selfeq1-1}, we have $\Lambda_d(z_1)+\Lambda_o(z_1)\le CN^{-1/3}$
in the set $\Theta(z_1)^c$, in particular, $\Theta(z_1)^c \subset{\bf B}(z_1)^c $.
Moreover, by $\Lambda(z_1)\le CN^{-1/3}$ in the set $\Theta(z_1)^c$, and \eqref{Kdef}, 
  the second alternative of \eqref{81} cannot hold
and therefore $\Lambda(z_1)\le U\beta(z_1)$
in the set $\Theta(z_1)^c\cap \Omega(z_1)^c\cap {\bf B}(z_1)^c =\Theta(z_1)^c\cap \Omega(z_1)^c $.
Using the probability estimates \eqref{B12} and \eqref{B12-1},
we have proved that 
\be
\P \bigg [  \Lambda (z_1)  \ge U \beta(z_1)
  \bigg ] + \P\big( {\bf B}(z_1)\big) \le C\exp{\big[-c(\log N)^{\phi \ell } \big]}.
\label{eta1}
\ee

For a general $k$ we have the following:

\begin{lemma}\label{lm:ind} 
There exist constants $0< \phi < 1$, $C'>1$, $c>0$, depending on $\ttau$,
such that if $\ell$ satisfies $4/\phi\le \ell \le C'\log N/\log\log N$
and $U$ is chosen $U\ge  U_0(\delta_\pm, C_0)$ (see  Lemma \ref{592})
 then  the following hold for any $k\le k_0$
and for any sufficiently large  $N\ge N_0(\ttau,\delta_\pm, C_0, U)$: 
\medskip
\noindent
{\it Case 1.} 
If $\eta_k \ge \wt\eta(U,E)$, then
\be\label{61}
\P \bigg [ 
  \Lambda (z_k)  \ge U\beta(z_k)    \bigg ] \le  C'k\exp{\big[-c(\log N)^{\phi \ell} \big]}
\quad \mbox{and}\quad
 \P\big( {\bf B}(z_k)\big) \le  C'k\exp{\big[-c(\log N)^{\phi \ell} \big]}.
\ee  
{\it Case 2.} If  $\eta_k < \wt\eta(U,E)$, then
\be\label{61uj}
\P \bigg [ 
  \Lambda (z_k)  \ge C_1(U)\beta(z_k)    
   \bigg ]\le C'k\exp{\big[-c(\log N)^{\phi \ell} \big]}
 \quad \mbox{and}\quad  \P\big( {\bf B}(z_k)\big) \le C'k\exp{\big[-c(\log N)^{\phi \ell} \big]},
\ee   
where $C_1(U)$ is given from Lemma \ref{592}.
\end{lemma}

{\it Proof.} We proceed by induction on $k$, the case $k=1$ has been checked in \eqref{eta1}.
First consider Case 1, when $k<k_0$ is such that $\eta_k \ge \wt\eta(U,E)$, i.e.
\eqref{61} holds by the induction hypothesis.
By the definition of the sequence $z_k$, we  have 
\be\label{con}
\Big |  G_{ij}(z_k)  -   G_{ij} (z_{k+1}) \Big |
\le |z_k-z_{k+1}|  \sup_{z\in \bS_\ell}
\Big | \frac{\partial G_{ij} (z) }{\partial z} \Big | \le N^{-8} \sup_{z\in \bS_\ell}
\frac{1}{|\im z|^2} \le  N^{-6}
\ee
for any $i,j$.  Hence $|\Lambda(z_k)-\Lambda(z_{k+1})|\le N^{-6} \le \frac{1}{2}U\beta(z_{k+1}) $
 and thus
\be\label{611}
\P \bigg [     \Lambda (z_{k+1}) \ge \frac{3}{ 2} U \beta(z_{k+1})  \bigg ] 
\le C'k\exp{\big[-c(\log N)^{\phi \ell} \big]}.
\ee  
In other words, the estimate on $ \Lambda (z_{k+1}) $ is deteriorated by a factor $3/2$,
but it will be gained back by the dichotomy estimate in Lemma \ref{592}.

Using \eqref{con} we also have, in $\Omega(z_k)^c\cap {\bf B}(z_k)^c$,
\begin{align}\label{68}
 \Lambda_d (z_{k+1}) + \Lambda_o (z_{k+1})  & \le \Lambda_d (z_{k}) + 
\Lambda_o (z_{k})  + 2N^{-6}  \non \\
 &  \le 
(\log N)^\ell   (\Lambda(z_k)^2+\Psi (z_k)) + \Lambda(z_k) + 2N^{-6}   \non \\
& \le (\log N)^{2 \ell} \sqrt{\frac{U\beta(z_k)
 +\im m_{sc}(z_k) }{N\eta_k}} + 2U\beta(z_k) + 2N^{-6} . 
\end{align}
Here in the second line
we used the bounds \eqref{21} and \eqref{71d} that hold on
the set $\Omega(z_k)^c\cap {\bf B}(z_k)^c$, in the last line
we used $\Lambda(z_k)\le U\beta(z_k)\le (\log N)^{-\ell}$. All these
estimates hold on an event with probability at least 
$1-C'(k+\frac{1}{2})\exp{\big[-c(\log N)^{\phi \ell} \big]}$ using \eqref{B12} and 
the estimate on $\P({\bf B}(z_k))$ from \eqref{61}.
Here we assumed that the constant $C'$ is larger than twice
the constant $C$ in \eqref{B12}.

By the choice of $\ell\ge 4$ and the definition of $\beta$ from 
\eqref{58},  the last line of \eqref{68}
 is bounded by $   (\log N)^{-2}$ and thus we have 
\be\label{Bk+1}
\P ( {\bf B}(z_{k+1})) \le C'\Big(k+\frac{1}{2}\Big)\exp{\big[-c(\log N)^{\phi \ell} \big]}.
\ee

Suppose now that $k+1$ falls into the first case,
$\eta_{k+1}\ge \wt\eta(U,E)$, then, from \eqref{newalz},
$$
   \frac{3}{2} U\beta(z_{k+1})<\frac{\al(z_{k+1})}{U},
$$
so by the dichotomy estimate \eqref{81}, $\Lambda(z_{k+1})\le
\frac{3}{2}U\beta(z_{k+1})$  from \eqref{611} implies  $\Lambda(z_{k+1})\le
U\beta(z_{k+1})$ on the set $\Omega(z_{k+1})^c\cap {\bf B}(z_{k+1})^c$.
Thus \eqref{B12}, \eqref{611} and \eqref{Bk+1} imply that
\be\label{612}
\P \bigg [     \Lambda (z_{k+1})  \ge  U \beta(z_{k+1})
\bigg ] \le C'(k+1)\exp{\big[-c(\log N)^{\phi \ell} \big]}
\ee  
by using $C'\ge 2C$ where $C$ is the constant from \eqref{B12}.
This proves  \eqref{61}, i.e.  the induction step
if $\eta_{k+1}$ is in the first case.
If $\eta_{k+1}$ falls into  the second case, i.e.,  $\eta_{k+1}\le \wt\eta(U,E)$, 
then \eqref{611} gives directly the
induction step, i.e. \eqref{61uj} for $k+1$.

So far we considered Case 1, i.e., we assumed that $\eta_k\ge \wt\eta(U,E)$. Now consider
Case 2, when  $\eta_k< \wt\eta(U,E)$ and therefore the induction hypothesis is
\eqref{61uj}. The argument is very similar to the previous case
but $U\beta(z_k)$ is replaced with $C_1(U)\beta(z_k)$ everywhere in  
 \eqref{611}, \eqref{68} and we   still obtain \eqref{Bk+1}.
Since $\eta_{k+1} <\eta_k\le \wt\eta(U,E)$, we can
directly refer  to \eqref{82} to obtain the
induction step, i.e. \eqref{61uj} for $k+1$. This completes the proof
of Lemma \ref{lm:ind}. \qed

\medskip

Choosing a sufficiently  large but fixed $U$, e.g. $U=U_0(\delta_\pm, C_0)$, 
we have thus proved that $\Lambda(z_k) \le C\beta(z_k)$  for all  $k\le k_0$
with a constant depending on $\delta_\pm$ and $C_0$, 
in particular $\Psi(z_k)\le C\beta(z_k)$ by the definition of $\Psi$ \eqref{1.9}.
Using \eqref{21} and \eqref{71d} 
 we have proved  Theorem \ref{thm:detailed} for all $z_k$, $k\le k_0$ and 
any fixed energy $E$  with $|E| \le 5$. 
 For any $z=E+i\eta\in \bS_\ell$
there is a $z_k =E+i\eta_k$ with $|z-z_k|\le N^{-8}$. Using the
Lipschitz continuity of $G_{ij}(z)$ and $m_{sc}(z)$ with Lipschitz
constant at most $N^2$, we easily conclude the proof of Theorem \ref{thm:detailed}
for  any $z\in \bS_\ell$.  Note that in order to accommodate the higher 
$(\log N)$-power in $\beta$ and the additional logarithmic factors
in  \eqref{21} and \eqref{71d}
 with the final formulation of the result in Theorem \ref{thm:detailed},
 we needed to redefine $\ell \to \ell/3$ which results
in a decreased $\phi$ in the final statement.
\qed

\section{Optimal error bound in the strong local semicircle law}\label{sec:optimal}

We have proved Theorem \ref{thm:detailed} which is weaker than the main result Theorem \ref{45-1}
but it will be used as an apriori bound for the improvement.  
The key ingredient for the stronger result is the following lemma which shows that 
$\barZ$, the average of $Z_i$'s, is much smaller than 
the size of typical $Z_i$.
(Notice that in the proof of  Theorem \ref{thm:detailed}, 
$\barZ$ was estimated in \eqref{m-m2} by the same quantity,
 $\Psi$, as each individual $Z_i$.)

For $z\in \bS_\ell$ define  
\be
   \Gamma =\Gamma(z):=\Omega_h\cup {\bf B}(z),
\qquad \Delta=\Delta(z):=\Omega(z)\cup {\bf B}(z),
\label{defGamma}
\ee
where $\Omega_h, \Omega$ were defined in  \eqref{eq:excep}--\eqref{defOmega}
and ${\bf B}$ was given in  \eqref{B}.  Recall that $\Omega_h$ and $\Omega$ depend
on $\ell$ and thus $\Gamma$ and $\Delta$ also depend on $\ell$ but we omit this fact
from the notation.
We remark that 
Theorem \ref{thm:detailed}  shows that there exists a positive constant $\phi>0$ such that
for any $4/\phi\le\ell\le \log N/\log\log N$ we have 
\be
   \P \big( {\bf B}(z)\big) = \P \Big( \Lambda_d(z)+ \Lambda_o(z) \ge (\log N)^{-2}\Big)
  \le   C\exp{\big[-c(\log N)^{\phi \ell} \big]}, \qquad z\in \bS_\ell,
\label{Best}
\ee
since 
the error bar $(\log N)^\ell/(N\eta)^{1/3}$ in Theorem \ref{thm:detailed}
is much smaller than $(\log N)^{-2}$. 
Combining \eqref{Best} with \eqref{B12} and $\Gamma\subset\Delta$, we get
that 
\be
   \P \big( \Gamma(z) \big) \le \P \big( \Delta(z)\big)
\le  C\exp{\big[-c(\log N)^{\phi \ell} \big]}, \qquad z\in \bS_\ell,
\label{Gammaest}
\ee
with positive constants $C,c$ depending only on  $\ttau$ in \eqref{subexp},
$\delta_\pm$ from Assumption {\bf (B)} and  $C_0$ from
Assumption {\bf (C)}.

With this notation, and recalling 
that $\Lambda_o(z)=\max_{i\ne j} |G_{ij}(z)|$, we then have the following
lemma whose proof will be given separately in Section  \ref{sec:Z}.

\begin{lemma}\label{motN} There exist positive constants $D\ge 1$, $A_0\ge 1$, 
 and $\psi \le \min\{ 1/10, \phi\}$,
depending on $\ttau$,  
 such that for any
$\ell$ with 
\be
    A_0\log\log N \le \ell \le \frac{\log N}{\log\log N},
\label{Lboud}
\ee
 for any  $p\le (\log N)^{\psi \ell-2}$ positive even number
 and
for any fixed $z\in\bS_\ell$   we have
\be\label{52}
\E \Bigg[ {\bf 1}\big(\Gamma^c(z)\big) \Big|\frac1N\sum_{i=1}^N Z_i(z)\Big|^{p} \Bigg]\leq
  (Dp)^{Dp} 
  \E \Big[{\bf 1}\big(\Gamma^c(z)\big)  \big[ \Lambda_o(z)^{2} + N^{-1}\big]^{p}\Big]
\ee
for any sufficiently large $N\ge N_0(A_0, \psi)$.
\end{lemma}

 The first version of this lemma was presented in Lemma 5.2 of \cite{EYY2}
where the $p$-dependence of the constant in \eqref{52} was not carefully tracked
and the effect of the exceptional event $\Gamma$ was estimated less
precisely. This was sufficient since in \cite{EYY2} we applied the result for 
an exponent $p$ independent of $N$; as a consequence, in particular, the
probability estimates for the local semicircle law
were only power law and not subexponential in $N$ as here.
In the current paper we allow $p$ to depend
on $N$ which requires the more precise form as stated in Lemma~\ref{motN}.
Furthermore, here we give a new proof that relies on a
different organization of partially independent terms.
 The main difference
is that here we separate dependences on individual
matrix elements, while in \cite{EYY2} we separated
entire rows and columns. The new method is therefore
more robust, but combinatorially more demanding.

\medskip
Recalling the notation 
$$
   \barZ = \barZ(z)= \frac{1}{N}\sum_{i=1}^N Z_i(z),
$$
we will apply Lemma~\ref{motN} in the following form:

\begin{corollary}\label{corZ}  There exist positive constants $D\ge 1$, $A_0\ge 1$, 
 and $\psi \le \min\{ 1/10, \phi, 1/D\}$,
depending on $\ttau$,  
 such that for any
 $\ell$ satisfying \eqref{Lboud},
for any $p\le (\log N)^{\psi \ell-2}$  positive even  
number and
for any fixed $z\in\bS_\ell$ \eqref{defS} 
we have  for any set $\Xi$ in the probability space 
\be\label{eq:corZ}
\E \Big[ {\bf 1}\big(\Gamma^c\big)|\barZ(z)|^p\Big] \le  
  \E \Big[ {\bf 1}\big(\Gamma^c\cap \Xi^c\big) \Psi(z)^{2p}\Big]+  (Dp)^{D p} 
 \big [\P(\Omega(z)) + \P (\Xi)  \big ]
\ee
where $\Omega(z)$ is defined in Lemma \ref{selfeq1}.
\end{corollary}

{\it Proof.}  On the right hand side of \eqref{52} we can split the set $\Gamma^c$
as
$$
  \Gamma^c = \Omega_h^c \cap {\bf B}^c =
\big [ \Omega^c \cap {\bf B}^c \cap \Xi^c \big ] \cup \big [ \Omega^c \cap {\bf B}^c \cap \Xi \big ]
 \cup \big[ (\Omega_h^c\setminus\Omega^c) \cap {\bf B}^c \big]
$$
On the set $\big [ \Omega^c \cap {\bf B}^c \cap \Xi \big ]
 \cup \big[ (\Omega_h^c\setminus\Omega^c) \cap {\bf B}^c \big]
 \subset {\bf B}^c$, we  estimate $\Lambda_o$ trivially by
\be\label{x1}
\Lambda_o\le (\log N)^{-2}\le 1. 
\ee
Since $\Omega_h^c\setminus\Omega^c \subset \Omega$,  we have 
\be\label{eq:corZnew}
\E \Big[ {\bf 1}\big(\Gamma^c\big)|\barZ(z)|^p\Big] \le  
 (Dp)^{Dp}
  \E \Big[ {\bf 1}\big(\Omega^c \cap {\bf B}^c   \big) 
 \big[ \Lambda_o(z)^{2} + N^{-1}\big]^{p}\Big] + (Dp)^{Dp}  
   \big [ \P (\Xi)  +\P(\Omega)  \big ].
\ee 
Choosing $\psi\le 1/D$ we see that $(Dp)^D\le (\log N)^\ell$. Thus
we can use
 $(\log N)^\ell N^{-1} \le C\Psi^2(z)$ 
for $z\in \bS_\ell$ (by $\im m_{sc}(z)\ge c\eta$) and that  $(\log N)^\ell\Lambda_o^2\le \Psi^2$ on 
$\Omega^c\cap {\bf B}^c$, see \eqref{21}, to absorb the $(Dp)^{Dp}$ prefactor 
in the first term in \eqref{eq:corZnew}.
This concludes the proof of Corollary \ref{corZ}. \qed

\bigskip

\begin{lemma}\label{selfKK}  Fix two numbers $\ell$ and $L$
 that satisfy $4\le \ell \le L\le\frac{\log (10 N)}{10\log\log N}$, in
particular $\bS_L\subset\bS_\ell$,
and let $0< \tau\le 1$ be an arbitrary constant.
For any $z=E+i\eta$ define
\be
  \gamma = \gamma(z): =  \frac{(\log N)^{3\ell+2}}{(N\eta)^\tau}.
\label{gammadef}
\ee
Suppose that  for all $z\in {\bf S}_L$ we have 
\be\label{spo1}
 \Lambda (z) \leq  \gamma(z) 
\ee
and 
\be\label{Zbund}
|[Z](z)|\leq (\log N)^{3 \ell}\left(\frac{\gamma(z) +\im m_{sc}(z) }{N\eta}\right).
\ee
Suppose that  $\Lambda(z)= o(1)$ for $\eta=10$, $|E|\le 5$.
 Then  in the set $\Omega^c \cap {\bf B}^c$ we have 
\be\label{spo2}
\Lambda(z)\leq (\log N)^{3 \ell  +2}(N\eta)^{-(\tau+ 1)/2}
\ee
for any $z\in {\bf S}_L$.
Furthermore, if $\Lambda(z)\leq \al(z)/2$ and \eqref{Zbund} hold
for some $z\in{\bf S}_L$, then 
\be\label{spo3}
\Lambda(z) \leq   C(\log N)^{3 \ell+1}\left(\frac{\gamma(z)+\im m_{sc}(z) }{\al(z) N\eta}\right),
\ee
in the set $\Omega^c \cap {\bf B}^c$,
where $\alpha$ was defined in \eqref{58}. 
\end{lemma}

{\it Proof:} 
In the first part of the proof $z\in \bS_L$ is fixed
so we drop the $z$-dependence of various quantities. 
Recall  \eqref{58}, \eqref{Kdef} and Lemma \ref{lm:msc} for $m_{sc}$ 
and $\alpha  \sim \sqrt {\kappa + \eta}$.  
From  Lemma \ref{591} and using \eqref{Zbund},    in the set $\Omega^c \cap {\bf B}^c$ we have, 
  with $w:= [v]$, the estimate 
\be\label{115}
\frac{(1-m_{sc}^2)}{m_{sc}^3}w -w^2=O\left(\frac{|w|^2}{\log N}\right)+O\Big [( \log N)^{3 \ell+1}
\left(\frac{\gamma+\im m_{sc} }{N\eta}\right)\Big ] 
\ee 
where we have used \eqref{spo1}, the definition
of $\Psi$ \eqref{1.9} and that $|w| = \Lambda$.  We can complete the square of the left side 
and obtain the inequality 
\be
\Lambda\leq 2\al +C(\log N)^{ \frac{3\ell+1}{2}}\left(\frac{\gamma+ \alpha }{N\eta}\right)^{1/2},
\label{sqcom}
\ee
where we have used that $\im m_{sc} \le C  \alpha$.  
We claim that in fact
\be\label{120}
\Lambda\leq 2\al +C(\log N)^{ \frac{3 \ell+1}{2}}\left(\frac{\gamma}{N\eta}\right)^{1/2}
\ee 
also holds; indeed this is trivial if $\Lambda\le 2\al$, and 
if $\Lambda\geq 2\al$ then by assumption 
\eqref{spo1}
$\gamma\ge \Lambda \ge 2 \alpha$, so $\al$ can be absorbed into $\gamma$
in \eqref{sqcom}.

Define
\be\label{algeClo}
\al_0=\al_0(z): =   T (\log N)^{ \frac{3 \ell+1}{2}}\left(\frac{\gamma}{N\eta}\right)^{1/2}
 =  T (\log N)^{ 3 \ell+ 3/2} (N\eta)^{-\frac { 1 + \tau } 2}
\ee
with  a large  parameter $T$ (independent of $N$) to be specified later,
and note that  $\al_0\le\gamma$ for sufficiently large $N$.

Suppose that  $\Lambda\leq \al/2$. In this case the $w^2$ terms are smaller than the leading term $\al w$
in the left hand side of \eqref{115}, therefore we can express $|w|=\Lambda$ and estimate
it by 
\be\label{122-1}
\Lambda \leq C(\log N)^{3 \ell+1}\left(\frac{\gamma+\im m_{sc} }{\al N\eta}\right)
\leq C(\log N)^{3  \ell+1}\left(\frac{\gamma}{\al N\eta}  + \frac{1  }{ N\eta}\right).
\ee
In the second step also used $\im m_{sc} \le C  \alpha  $.
In particular, the first inequality proves \eqref{spo3}.

Assume now that $\Lambda \le \alpha/2$ and  $\alpha \ge \alpha_0$. 
Plugging the lower bound 
 \eqref{algeClo} on $\alpha$ into  \eqref{122-1}
and using the definition of $\gamma$ we obtain
\be
\Lambda 
\leq  C T ^{-1} (\log N)^{ \frac{3 \ell+1}{2}} 
\left(\frac{\gamma}{N\eta}\right)^{1/2} = CT^{-2}\alpha_0 . 
\label{dic}
\ee
Choosing $T$ as a sufficiently large constant we obtain that
\be\label{122}
\Lambda\le \frac{\al}{4}
\ee
under the condition that $\Lambda \le \alpha/2$ and  $\alpha \ge \alpha_0$.
 Therefore, as long as $\al\ge \al_0$, we
have a dichotomy: either $\Lambda \ge \al/2$ or $\Lambda\le \al/4$. 

We now fix $E$ and we continuously decrease $\eta$ from $\eta=10$ to
 $\eta = N^{-1}(\log N)^L$, the
lower point in $\bS_L$.  Since  $\Lambda(z)\ll 1$ and $\alpha(z)$ is bounded away from zero
for $\eta=10$, $|E|\le 5$,  we know that $\Lambda\le \al/2$ holds for $\eta=10$.
Since $\Lambda(z)$ is continuous function, by the dichotomy 
we have that $\Lambda\le \al/4$  for all $\eta$ as long as
$\alpha \ge \alpha_0$. In particular, $\Lambda\le  CT^{-2}\alpha_0$ from \eqref{dic}
which proves  \eqref{spo2} in the case $\al\ge \al_0$.
 
Finally, for $\alpha \le \alpha_0$, we can estimate $\Lambda$ directly via \eqref{120} 
and this proves that 
\be\label{122-2}
\Lambda 
\leq  C  (\log N)^{ \frac{3 \ell+1}{2}} \left(\frac{\gamma}{N\eta}\right)^{1/2}
\ee
from which  \eqref{spo2} follows
and we have thus completed the proof. 
\qed

\medskip

{\it Proof of Theorem \ref{45-1}.}  First we explain the idea.
We will prove, by  an induction on the exponent $\tau$,  
that $\Lambda\le (N\eta)^{-\tau}$ holds modulo logarithmic factors
with a high probability. Notice that we proved this statement
for $\tau =1/3$ in Theorem \ref{thm:detailed}.
Lemma \ref{selfKK} asserts that if this statement
is true for some $\tau$, then it also holds for $\frac{1+\tau}{2}$ assuming
a bound on $[Z]$. This bound can be obtained from Corollary \ref{corZ}
with a high probability.
Repeating the induction step for $O(\log \log N)$ times, we will obtain 
that $\tau$ is essentially one, i.e. we get Theorem \ref{45-1}.
However, we have to keep track of the increasing logarithmic factors
and the deteriorating probability estimates of the exceptional sets.

Throughout the proof we fix $L$ satisfying \eqref{Lbound}
with the constant  $A_0$ obtained from Corollary \ref{corZ} and
we also fix $\psi$ from the same Corollary.
We will also use a moving exponent $\ell$ 
whose value will always
satisfy $L/2\le \ell \le  L$, in particular $\bS_L\subset \bS_\ell$.

We recall the definition
\be
   \gamma=\gamma(z, \tau, \ell)=\frac{(\log N)^{3 \ell +2}  }{(N\eta)^{\tau}},
\label{gamdef}
\ee
where we now emphasize the dependence on $\tau$ and $\ell$. 
Define the events 
\be\label{y1}
  R_{\tau, \ell}  := \bigcup_{z\in \bS_L}  R_{\tau,\ell}(z), \;\;\; 
R_{\tau,\ell} (z) :=  \Big\{ \Lambda (z) \geq \gamma(z,\tau,\ell) \, \Big\}.
\ee
Then \eqref{mainlsresult} in
 Theorem \ref{thm:detailed} states  that there is a $\psi$ with $0<\psi<1/10$
such that for any $\ell_0:=L$ we have
\be\label{y1-0}
   \P\left( R_{\tau, \ell_0} \right)\leq \exp{\big[-(\log N)^{\psi \ell_0 } \big]},
 \ee
with   $ \tau = 1/3$ and for any $N\ge N_0(\ttau, \delta_\pm, C_0)$.
 Notice that we have used  a weaker form of  Theorem \ref{thm:detailed}
 by making the threshold $\gamma$  larger, the 
restrictions for $\ell_0$ stronger 
 and  reducing the exponent $\phi$ to $\psi$
since  this weaker form will be preserved in  the iterative procedure.   
By setting a sufficiently large lower threshold on $N$, we could remove the
constants $C, c$ from \eqref{mainlsresult}.
The general iteration step is included in the following lemma.

\begin{lemma} \label{y4} There exists a sufficiently large $N_0=N_0(\ttau,\delta_\pm, C_0)$
such that for any $N\ge N_0$ the following implication holds.
If for some $0< \tau < 1$ and for 
some  $\ell$ with $L/2\le \ell\le L$
\be\label{y1-1}
   \P\left( R_{\tau, \ell} \right)\leq \exp \big[-(\log N)^{\psi  \ell } \big] ,
\ee
 then 
\be\label{y2}
   \P\left( R_{\tau', \ell'} \right)\leq \exp{\big[-(\log N)^{  \psi  \ell'} \big]},
\ee
where 
\be\label{y3}
\tau' = \frac {\tau + 1} 2 , \quad \ell' = \ell - \frac 3 { \psi}.
\ee
\end{lemma}

\begin{proof} 
 Define 
\be \label{1.9x}
\Phi = \Phi(z,\tau, \ell) : = (\log N)^\ell  \sqrt{\frac{\gamma(z,\tau,\ell)
+\im m_{sc}(z) }{N\eta}}.
\ee
Fix $z\in \bS_L$, then 
from  Corollary \ref{corZ} with the choice of $\Xi = R_{\tau,\ell}$ we have 
\begin{align}
\E \Big[ {\bf 1}\big( \Gamma^c\big)|\barZ|^p\Big]  & \le
 \E \Big[ {\bf 1}\big( R_{\tau,\ell}^c \big) \Psi^{2p}\Big] + 
  (Dp)^{Dp}    \exp \Big [  - c(\log N) ^{ \psi \ell } \Big ] \\ 
 &\le\Phi^{2p} +   (Dp)^{Dp}  \exp \Big [  - c(\log N) ^{ \psi \ell }   \Big ], 
\end{align}
where we have used \eqref{y1-1} and \eqref{B12} to bound the probability of $\Xi$ and $\Omega$
and we used that $\Lambda\le \gamma$ on $R_{\tau,\ell}^c$ to estimate $\Psi\le \Phi$.
 We will choose $p = (\log N)^a$ with 
\be
a = \psi \ell  - 3.
\ee
{F}rom Markov's inequality and \eqref{Gammaest} we obtain that
\begin{align}
\P  \Big (  |[Z]| \ge  \frac{1}{2}(\log N )  \Phi^2 \Big ) 
& \le 2^p( \log N)^{-  p }  \Phi^{-2p}  \Big [  \Phi^{2p} +  
  (Dp)^{Dp}    \exp \big [  - (\log N) ^{ \psi \ell } \big]\Big ]    
 + C\exp [- c(\log N)^{\phi\ell}]
\non\\
\label{2.12}
& \le  2^p( \log N)^{-   p }  +       \exp 
\big[ Dp\log (2Dp)+ p(\log N) -  (\log N)^{a+3}  \big ]  + C\exp [- c(\log N)^{\phi\ell}]    
\non\\
& \le
\exp \Big [  -   3(\log N)^{a}  \Big ].
\end{align}
Here in the second line we used $\Phi\ge N^{-1/2}$ from $\Im\, m_{sc}(z)\ge c\eta$
to estimate $\Phi^{-2p}$. In the final estimate we used that
$\log p = a\log \log N\le \psi\ell \log\log N\le \psi\log N$
and that $\psi\le \phi$.
This estimate was for any fixed $z\in\bS_L$.
By choosing a grid of $z$-values in
 $\bS_L$ with spacing of order $N^{-c}$, 
with some large $c$, we can use
the Lipschitz  continuity of $[Z](z)$ and $\Phi(z)$
to conclude that essentially the same estimate holds
simultaneously for all $z\in \bS_L$.

Combining this with \eqref{y1-1}, we have
\be\label{x2} 
 |[Z]|\le  (\log N )   \Phi^2  \le  (\log N )^{ 3 \ell } 
\Big (  \frac{\gamma +\im m_{sc} }{N\eta}\Big ) \qquad \mbox{and} \qquad \Lambda\le \gamma,
\ee
for all $z\in \bS_L$ with a probability at least $1- \exp \big [  -   2(\log N)^{a}  \big ]$.
We can now apply Lemma \ref{selfKK} so that 
\be\label{spo21}
\Lambda(z)\leq (\log N)^{3 \ell + 2}(N\eta)^{-(\tau+ 1)/2}
\ee
hold  for any $z\in {\bf S}_L $
with  a probability 
bigger than  $1- \exp \big [  -   (\log N)^{a}  \big ]$.  
Here we have used that  
 $\P (\Omega \cup {\bf B}) \le \exp \big [  -   2(\log N)^{a}  \big ]$
 from  \eqref{Gammaest}. 
  We have thus proved \eqref{y2} and Lemma~\ref{y4}.
\end{proof} 

\bigskip

Returning to the proof of Theorem  \ref{45-1},  we choose  
 $\tau_0= 1/3$ and $\ell_0=L$ as the initial values of the iteration.
The input condition \eqref{y1-1} in Lemma \ref{y4}
for the initial step has been checked in \eqref{y1-0}.
Iterating  Lemma \ref{y4} yields a sequence of $(\tau_n, \ell_n)$
so that $\tau_{n+1}=\tau_n'$ and $\ell_{n+1}=\ell_n'$ via \eqref{y3}, more precisely
$$
   \tau_n = 1- 2^{-n}\cdot\frac{2}{3} \ge 1-2^{-n}, \qquad \ell_n = L - 3n/\psi,
$$
such that
\be\label{Lambdafinal-3} 
\P \Big ( \bigcup_{z\in \bS_L} \Big\{ \Lambda(z) 
 \ge \frac{(\log N)^{ 3\ell_n+2 } }{ (N\eta)^{1 -2^{-n} } } \Big\}   \Big )
\le  \exp{\big[-(\log N)^{ \psi\ell_n } \big]}.  \quad 
\ee
We run the iteration until $n =2\log\log N$ so  that
$$
   (N\eta)^{2^{-n}} \le N^{2^{-n}} \le e.
$$
If $A_0=20/\psi$, i.e. $L\ge (20/\psi)\log\log N$,  then 
$\ell_n \ge 2L/3$ and thus
\be\label{Lambdafinal-1} 
\P \Big ( \bigcup_{z\in \bS_L} \Big\{ \Lambda(z) 
 \ge \frac{e(\log N)^{ 3L+2 } }{ N\eta } \Big\}   \Big )
\le  \exp{\big[-(\log N)^{ 2\psi L/3 } \big]}.  \quad 
\ee
This proves
\eqref{Lambdafinal} after renaming  $2 \psi/3$ to a new $\phi$. 
The proof of \eqref{Lambdaodfinal} follows from the estimate
on $\Lambda$,  from \eqref{21}, \eqref{71d} and 
\eqref{Gammaest}.

\medskip

Finally, to prove \eqref{443}, we need the following Lemma. 

\begin{lemma}\label{etaf} Let  $L\ge4$ satisfy \eqref{Lbound} 
and define the set
\be\label{etafix}
\bU_L:= \Big \{z=E+i\eta \; : \;
 5\geq |E|\geq 2+N^{-2/3}(\log N)^{8L+8},\,\,\,\,\,\eta=N^{-2/3}(\log N)^{2L+1}  \Big \}.
\ee
Then for $A_0$ large enough in \eqref{Lbound}, we have
\be\label{434}
\P \Big (\bigcup_{z\in \bU_L}
 \Big\{ 
  \Lambda(z)\leq (\log N)^{-1}(N\eta)^{-1}   
  \Big\}\Big)
  \ge 1-   C\exp{\big[-c(\log N)^{\psi L/2  } \big]}.  \quad 
\ee
\end{lemma}
\bigskip

{\it Proof of Lemma \ref{etaf}: } 
For $z\in \bU_L$ we have $\kappa = N^{-2/3}(\log N)^{8L+8} \ge \eta$ and thus 
 we have  (see \eqref{Kdef})
$$
  \al(z)\ge  c\sqrt{\kappa+\eta} \ge c N^{-1/3}(\log N)^{4L+4}.
$$
Therefore
$\Lambda \le \alpha /2 $ holds on the event $
\Lambda (z) \le \frac {e(\log N)^{3 L+2}} { N \eta} $ for any $z\in \bU_L$.
 Since $\bU_L\subset\bS_L$,
the  probability of this event is
 bigger than $1- \exp{\big[-(\log N)^{2\psi L/3} \big]}$
by \eqref{Lambdafinal-1}. Combining this bound on $\Lambda$
with the estimate \eqref{B12} for $\ell=L$, we know that
$$
   |[Z](z)|\le  (\log N)^{2L} \frac{ \gamma(z) + \im m_{sc}(z)}{N\eta}
$$
holds with a probability bigger than $1- 2\exp{\big[-(\log N)^{2\psi L/3} \big]}$.
Here we used $\gamma(z)= (\log N)^{3L+2}(N\eta)^{-1}$
with the choice of $\tau =1$ and $\ell =L$, see \eqref{gamdef}.

 We can now  use \eqref{spo3} from Lemma \ref{selfKK} with $\ell =L$ and $\tau=1$ to have 
\be\label{l1}
\Lambda 
\leq C(\log N)^{3 L+1}\left(\frac{(\log N)^{3L+2}(N\eta)^{-1}+\frac{\eta}{\sqrt{\kappa}} }
{\sqrt{\kappa} N\eta}\right)
\ee
with probability larger than $1- 3\exp{\big[-(\log N)^{2\psi L/3} \big]}$.
Here we used the probability estimate \eqref{Gammaest} on $P(\Omega\cup{\bf B})$ and
the first bound in \eqref{esmallfake}.
Then using the values of $\kappa$ and $\eta$ in the set   \eqref{etafix}, we obtain 
$$
\Lambda\leq (\log N)^{-1}(N\eta)^{-1} 
$$
from \eqref{l1} and this proves Lemma \ref{etaf}.  \qed

\medskip 

We now prove  \eqref{443}. On the set $\bU_L$ we have 
\be
\im m_{sc}=O(\frac{\eta}{\sqrt\kappa})\leq  (\log N)^{-1}(N\eta)^{-1}.
\ee
Combining it with \eqref{434}, we obtain that 
\be\label{440}
\P \Big (\bigcup_{z\in \bU_L}
 \Big\{ 
\im m(z)\leq 2(\log N)^{-1}(N\eta)^{-1}   \Big\}  \Big )
  \ge 1-   C\exp{\big[-c(\log N)^{ \psi L/2  } \big]}.  \quad 
\ee

Fix $z=E+i\eta\in \bU_L$ and define the event
$$
   W(z): =\{ \exists j \; : \; |\lambda_j-E|\le\eta\}.
$$
Recalling the definition of $m$,
\be
\im m(z)=\frac1N\sum_{j=1}^N\frac{\eta}{(E-\lambda_j)^2 +\eta^2},
\ee  
it is clear that $\im m(z)\ge \frac{1}{4}(N\eta)^{-1}$ on the set $W(z)$.
Using \eqref{440} we obtain that 
\be\label{ext}
   P\Big( \exists j\; : \; 2+N^{-2/3}(\log N)^{8L+8}\le |\lambda_j|\le 5\Big)
  \le C\exp{\big[-c(\log N)^{ \psi L/2  } \big]}.
\ee
Finally, we need to control the probability of a very large eigenvalue. For example, the following 
(not optimal) estimate was  proved in, e..g,  
 Lemma 7.2 of \cite{EYY}. We formulate the results for the largest eigenvalue $\lambda_N$,
but analogous results hold for the smallest eigenvalue $\lambda_1$ as well.

\begin{lemma}\label{N-1/6}
 Let $H$ satisfy Assumptions {\bf (A), (B), (C)} and the subexponential decay condition
\eqref{subexp}.
Then for some $\e>0$, depending on $\ttau$, we have 
\be\label{5.2}
\P ( \lambda_N \ge K )\leq e^{-N^\e\log K}
\ee
for any $K\ge 3$.
\end{lemma}

Combining this lemma with  \eqref{ext} we 
 completed the proof of \eqref{443}.  \qed

\bigskip

\section{Estimates on the location of eigenvalues}

{\it Proof of Theorem \ref{7.1}.}
We now translate the information on the Stieltjes transform obtained
in Theorem~\ref{45-1} to prove Theorem \ref{7.1} on the location of the eigenvalues. 
We will need the following Lemma~\ref{lm:HS1}
which is a special case  of Lemma 6.1 proved in \cite{EYY2}
with the choice  $A=0$. The conditions (6.1) and (6.2) stated in 
Lemma 6.1  of  \cite{EYY2} are not sufficient.  Instead,  the following slightly stronger assumption is necessary:
\be\label{strongerbern}
    |m^\Delta (x+iy)|\le \frac{CU}{y(\kappa_x+y)^A}, \qquad
\mbox{for} \qquad 1\ge y >0,\quad |x|\le K+1,
\ee
i.e., it is not sufficient to control only the imaginary part of $m^\Delta$.
This stronger condition is needed in (6.7) of \cite{EYY2}, where
the imaginary part of $m$ is changed to its real part after an integration by parts. 
With  the condition \eqref{strongerbern},   the proof of  Lemma 6.1  in  \cite{EYY2}  remains otherwise unchanged.
This immediately proves the following lemma as a special case:

\begin{lemma}\label{lm:HS1}
Let $\varrho^\Delta$ be a signed measure on the real line with $\mbox{supp} \;\varrho^\Delta \subset
[-K,K]$ for some fixed constant $K$. 
For any $E_1, E_2 \in [-3,  3]$ and $\eta>0$ we define $f(\lambda)=f_{E_1,E_2,\eta}(\lambda)$
to be a characteristic function of $[E_1, E_2]$ smoothed on scale $\eta$, i.e.,
$f\equiv 1$ on $[E_1, E_2]$, $f\equiv 0$ on $\R\setminus [E_1-\eta, E_2+\eta]$
and $|f'|\le C\eta^{-1}$, $|f''|\le C\eta^{-2}$. 
Let $m^\Delta$ be the Stieltjes transform of $\varrho^\Delta$.
Suppose for some positive number $U$ (may depend on $N$)
we have
\be
   | m^\Delta (x+iy)|\le \frac{CU}{N y } \qquad \mbox{for} \qquad
1\ge y>0, \quad |x|+y\le K.
\label{trivv1}
\ee
Then
\be
   \left|\int_\R f_{E_1,E_2,\eta}(\lambda)\varrho^\Delta(\lambda)\rd\lambda \right|  \le
  \frac{CU|\log \eta|}{ N }
\label{genHS1}
\ee
with some constant $C$ depending on $K$. \qed
\end{lemma}

We will apply this lemma with the choice that the  signed measure is
the difference of the empirical density and the semicircle law,
$$
  \varrho^\Delta(\rd \la)=\varrho(\rd\la) - \varrho_{sc}(\la)\rd \la,\qquad
\varrho(\rd\la):=\frac{1}{N}\sum_i \delta(\lambda_i-\la).
$$

First we prove \eqref{nn}. Choose $L:=A_0 \log\log N$, where $A_0$ is given
in  Theorem \ref{45-1}, and we define 
$$
    T_N:= (\log N)^L = (\log N)^{A_0\log\log N}
$$
for simplicity.
 By  Theorem \ref{45-1}, the assumptions of Lemma \ref{lm:HS1} hold for 
the difference  $m^\Delta =m-m_{sc}$ with $K=10$ and $U=T_N^4$
 if $y\ge y_0 := T_N^{10}/N $.
For $y\le y_0  $,  set $z=x+iy$, $z_0=x+iy_0$ and
estimate
\be\label{mmmsc}
   |  m(z)-m_{sc}(z)|\le   |  m(z_0)-m_{sc}(z_0)|
  + \int_y^{y_0} \big| \partial_\eta 
\big(m(x+i\eta)-m_{sc}(x+i\eta)\big)\big|\rd \eta.
\ee
Note that
\begin{align}
   |\partial_\eta m(x +i\eta)| 
   = & \Big|\frac{1}{N}\sum_j \partial_\eta G_{jj}(x+i\eta)\Big|
 \\ \le & \frac{1}{N}\sum_{jk} |G_{jk}(x+i\eta)|^2 =
 \frac{1}{N\eta }\sum_j \im G_{jj}(x+i\eta) = \frac{1}{\eta}\im m(x+i\eta),
\end{align}
and similarly
$$
  |\partial_\eta m_{sc}(x+i\eta)| = 
\Big|\int \frac{\varrho_{sc}(s)}{(s-x-i\eta)^2}\rd s\Big|
\le \int \frac{\varrho_{sc}(s)}{|s-x-i\eta|^2}\rd s
 = \frac{1}{\eta} \im m_{sc}(x+i\eta).
$$
Now we use the fact that the functions $y\to y\im m(x+iy)$ and
$y\to y\im m_{sc}(x+iy)$ are monotone increasing for any $y>0$
since both are Stieltjes transforms of a positive measure.
Therefore the integral in \eqref{mmmsc} can be bounded by
\be\label{intbb}
   \int_y^{y_0} \frac{\rd \eta}{\eta} \big[ \im  m(x+i\eta) +  
\im m_{sc}(x+i\eta)\big] 
\le  y_0\big[ \im m(x+iy_0) +  
\im m_{sc}(x+iy_0)\big]  \int_y^{y_0} \frac{\rd \eta}{\eta^2}
\ee

By definition, $\im m_{sc}(x+iy_0)  \leq |m_{sc}(x+iy_0)| \le C$. By the choice of $y_0$ and  Theorem \ref{45-1},  we have 
\be
     \im \, m(x+iy_0) \le
   \im m_{sc}(x+ i y_0) + \frac {T_N^4} {N y_0} 
\le C 
\label{imest}
\ee
with very high probability.
 Together with \eqref{intbb} and \eqref{mmmsc}, 
this proves that \eqref{trivv1} holds for $y\le y_0  $
as well if $U$ is increased to $U=T_N^{10}$.

The application of Lemma \ref{lm:HS1} shows that for any $\eta\ge 1/N$
\be
   \left|\int_\R f_{E_1,E_2,\eta}(\lambda)\varrho(\la) \rd\lambda
 -\int_\R f_{E_1,E_2,\eta}(\lambda) \varrho_{sc}(\lambda)\rd\lambda \right|  \le
  \frac{C (\log N)T_N^{10}}{N}.
\label{genHS2}
\ee
With the fact: $y\to y\im m(x+iy)$ is  monotone increasing for any $y>0$,  \eqref{imest} implies a crude upper bound on the empirical density. Indeed,  for any interval $I:=[x-\eta, x+\eta]$, with $\eta=1/N$, we have 
\be
   {\mathfrak n}(x+\eta)- {\mathfrak n}(x-\eta ) \le C\eta\,\im \, m\big( x+ i\eta\big)
 \le Cy_0\,\im \, m\big( x+ iy_0\big)\leq \frac{CT_N^{10}}{N}.
\label{density}
\ee
This bound can be used to estimate the difference between
the characteristic function of the interval $[E_1,E_2]$ and
the smoothed function $ f_{E_1,E_2,\eta}$.

Since the probability to have eigenvalues outside the interval $[-3, 3]$ are
 extremely small, we consider only the case 
that all eigenvalues are inside $[-3, 3]$. Let $E_1= -4$ and $E_2:=E \in [-3,3]$. Then 
from \eqref{genHS2}  and \eqref{density} we have that
\be
   \Big| {\mathfrak n}(E) - n_{sc}(E)\Big|
  \le \frac{C (\log N)T_N^{10} }{N}
\label{fixE}
\ee
holds for any fixed $E\in [-3,3]$ with an overwhelming probability. 
The supremium over $E $ is a standard argument for 
extremely small events and we omit the details.  
This completes the proof of \eqref{nn} after possibly increasing $L$ (hence $A_0$)
 and decreasing $\phi$ in order to replace the $(\log N)T_N^{10}$ with $(\log N)^L$.

\medskip

Now we turn to the proof of \eqref{rigidity}. Let $L$ as before.
Fix any $ 1\le j \le N/2 $ and let $E = \gamma_j$, $E'= \lambda_j$.
Setting $t_N=(\log N)T_N^{10}=(\log N)^{10L+1}$ for simplicity, from \eqref{fixE}  we have 
\be
n_{sc}(E) = \fn(E') = n_{sc}(E') + O(t_N/N). 
\label{nnn}
\ee
Clearly $E\le 1$, and using \eqref{fixE} $E'\le 1$ also holds with an overwhelming
probability.
First, using \eqref{443} and 
\be\label{nscxs}
n_{sc}(x) \sim (x+2)^{3/2},
\,\,\,\mbox{for}\,\,\,-2\leq x\leq1, 
\ee 
i.e. 
$$
  n_{sc}(E)=n_{sc}(\gamma_j) = \frac{j}{N} \sim (E +2)^{3/2},
$$
we know that \eqref{rigidity} holds  (with a possibly increased power of $\log N$
in the left hand side)
 if 
\be
E,E'\leq -2+t_N N^{-2/3}.
\ee
The correct power $(\log N)^L$ can be restored by increasing $L$ (hence $A_0$)
and decreasing $\phi$, as before.

Hence, we can assume that one of $E$ and $E'$ is 
in the interval $[-2+t_N N^{-2/3},1]$. With  \eqref{nscxs}, this assumption implies that
at least one of
 $n_{sc}(E)$ and $n_{sc}(E')$ is larger than $t_N^{3/2}/ N$. 
 Inserting this information into \eqref{nnn}, we obtain that both  $n_{sc}(E)$ and $n_{sc}(E')$
are positive and 
$$ 
   n_{sc}(E) = n_{sc}(E') \big[ 1 + O(t_N^{-1/2})\big],
$$
in particular, $E+2\sim E'+2$. Using that $n_{sc}'(x)\sim (x+2)^{1/2}$ for $-2\le x\le 1$,
we obtain that $n_{sc}'(E) \sim n_{sc}'(E')$, and in fact $n_{sc}'(E)$ is
comparable with $n_{sc}'(E'')$ for any $E''$ between $E$ and $E'$.
Then with Taylor's expansion, we have 
\be\label{88}
|n_{sc}(E')-n_{sc}(E)|\le C |n_{sc}'(E)| |E'-E| .
\ee
Since $n_{sc}'(E) = \rho_{sc}(E) \sim \sqrt \kappa$ and $n_{sc}(E) \sim \kappa^{3/2}$,
moreover, by $E=\gamma_j$ we also have $n_{sc}(E) =j/N$,
 we obtain
from \eqref{nnn} and \eqref{88} that
\[
|E'-E| \le \frac{C|  n_{sc}(E')- n_{sc}(E)|}{n_{sc}'(E)}\le
 \frac {Ct_N} { N n_{sc}'(E) }   \le  \frac {C t_N} { N (n_{sc}(E))^{1/3} }
  \le \frac {C t_N} { N^{2/3}  j^{1/3} },
\]
which proves \eqref{rigidity}, 
again, after increasing $L$ and decreasing $\phi$ to achieve the claimed
 $(\log N)^L$ prefactor.
 This concludes the proof of Theorem \ref{7.1}.

\qed

\section{Edge Universality  }

In this section, we prove the edge universality, i.e.,  Theorem \ref{twthm}. 
At the end of Section \ref{sec:green} we will give a
 heuristic explanation  why matching the second moments is sufficient but
we first  need  some preparation and to introduce various notations.
We will consider the largest eigenvalue $\lambda_N$, but
the same argument applies to the lowest eigenvalue $\lambda_1$ as well.

For any $E_1\le E_2$ let
$$
   \cN(E_1, E_2): = \#\{ E_1\le \la_j\le E_2\}
$$
denote the number of eigenvalues in $[E_1, E_2]$.
By  Theorem \ref{7.1} (rigidity of eigenvalues), there exist positive constants $A_0$, $\phi$,
  $C$ and $c>0$,  
depending only on  $\ttau$, $\delta_\pm$  and $C_0$ such that with setting 
\be
L:= A_0\log\log N
\label{defL}
\ee
we have 
 \be\label{6-1}
 \P \left\{ \abs{N^{2/3} ( \la_N -2) }\geq (\log N)^{L} \right\} 
\leq C\exp{\big[-c(\log N)^{\phi L} \big]}
\ee 
and 
\be\label{6-2}
 \P\left\{ \cN\left(2-\frac{2(\log N)^L}{ N^{2/3}},\;2+\frac{2(\log N)^L}{ N^{2/3}}\right)\geq
  (\log N)^{L} \right\} \leq C\exp{\big[-c(\log N)^{\phi L} \big]}
\ee 
for sufficiently large $N\ge N_0(\ttau, \delta_\pm, C_0)$.
 These estimates hold for both the $\bv$ and $\bw$ ensembles. 
Using these estimates,  we can assume    that $s$ in \eqref{tw} satisfies  
\be\label{25}
-(\log N)^{L} \le s \le (\log N)^{L}.
\ee
With  $L$ from \eqref{defL}, we set 
\be\label{defEL}
E_L:= 2+2(\log N)^{L}N^{-2/3}.
\ee
For any $E\le E_L$ let 
\[
\chi_E := {\bf 1}_{[E, E_L]}
\]
be the characteristic function of the interval $[E, E_L]$. For
any $\eta>0$ we 
define
\be\label{thetam}
\theta_\eta(x):=\frac{\eta }{\pi(x^2+\eta^2)} = \frac{1}{\pi} \im \frac{1}{x-i\eta}
\ee
to be an approximate delta function on scale $\eta$.
In the following elementary lemma we compare the sharp counting function
$\cN(E, E_L)= \tr \chi_E(H)$ by its approximation smoothed  on scale $\eta$.

\begin{lemma}\label{lem:21} 
Suppose that the assumptions of Theorem \ref{twthm} hold
and $L$, $\phi$ satisfy \eqref{6-1} and \eqref{6-2}. 
For any  $\e>0$, set $\ell_1:=N^{-2/3-3\e}$ and $\eta:=N^{-2/3-9\e}$. 
Then there exist constants $C, c$  such that
for any $E$ satisfying
\be\label{E-2N}
|E-2|N^{2/3}\leq \frac 3 2(\log N)^L
\ee
we have
\be\label{6.10}
\P \Big\{  |\tr \chi_E(H) - \tr  \chi_E  \ast \theta_\eta (H)|  \le C\left( N^{-2\e}  
 +   \cN (E-\ell_1, E+\ell_1)  \right) \Big\} \ge 1- C \exp[-c(\log N)^{\phi L}]
\ee
for  sufficiently large $N$. This estimate holds
for both  the $\bv$ and $\bw$ ensembles.
\end{lemma}

\noindent 
{\it Proof of Lemma \ref{lem:21}.}  By \eqref{defEL} and \eqref{E-2N} we have 
\be
   \eta\ll \ell_1\ll E_L-E\le CN^{-2/3}(\log N)^L.
\label{EEL}
\ee
Since $\chi_E$ is the characteristic function of  $ [E, E_L]$, 
for any $x\in \R$, we have  
\[
|\chi_E(x) - \chi_E  \ast \theta_\eta (x)| 
=  \Big |  \Big (\int_\R  \chi_E(x) - \int_{E-x}^{E_L-x} \Big )   
\theta_\eta(y)\rd y  \Big | .
\]
Let  $d =d(x):= |x-E| +\eta$ and $d_L=d_L(x) := |x-E_L| +\eta$. Using that $\int \theta_\eta =1$ and
the estimate
$$
   \int_\al^\infty \theta_\eta(y)\rd y = \frac{1}{\pi}\int_\al^\infty 
\frac { \eta} {  y^2 + \eta^2} \rd y 
  \le \frac{C\eta}{\al + \eta}, \qquad \al>0,
$$
an elementary calculation shows that
\be\label{casess}
|\chi_E(x) - \chi_E  \ast \theta_\eta (x)|  \le  C\eta\Big[
\frac {E_L-E} {d_L(x) d(x)} + \frac {\chi_E(x)}{d_L(x)+ d(x)}\Big] 
\ee
for some constant $C>0$. It is easy to check that if $\min\{d, d_L\}\leq \ell_1$, then 
the right side of \eqref{casess} is bounded by a constant and  if $\min\{d, d_L\}\geq \ell_1$,
 then it is less than $O(\eta/\ell_1)=O(N^{-6\e})$. 
Hence  we have 
\be\label{6-20} |\tr \chi_E(H) - \tr \chi_E  \ast \theta_\eta (H)| 
 \le C\left( \tr f(H)  + \frac {\eta}{\ell_1}\,\,\cN (E, E_L)
 +   \cN (E-\ell_1, E+\ell_1) +  \cN (E_L-\ell_1, \infty) \right),
\ee
where
\be
f(x) := \frac {\eta (E_L-E)} {d_L(x) d(x)} {\bf 1}\left(x\leq E-\ell_1\right).
\ee
With the assumption \eqref{E-2N}, $\cN (E, E_L)$ and $\cN (E_L-\ell_1, \infty)$ 
can be bounded by using \eqref{6-2} and \eqref{6-1}. Hence it follows from \eqref{6-20} that 
\be\label{6-22} |\tr \chi_E(H) - \tr  \chi_E  \ast \theta_\eta (H)|  \le C\left( \tr f(H)   
 +   \cN (E-\ell_1, E+\ell_1) +N^{- 5\e}  \right)
\ee
holds with a probability larger than   $1-C \exp[-c(\log N)^{\phi L}]$, 
for some constants $C$ and $c$ and for sufficiently large $N$, uniformly in $E$ with \eqref{E-2N}.
Set
\be
g(y): = \frac 1 { y^2 + \ell_1^2},
\ee
and notice that
\be
  \frac{1}{a^2} \le C (g  \ast \theta_{\ell_1}) (a)     \qquad \mbox{ if }    |a| \ge \ell_1,
\ee 
which implies 
\be
\frac{f(x)}{\eta(E_L-E) } =\frac{  {\bf 1}\left(x\leq E-\ell_1\right)}{d_L(x)d(x)}
  \le \frac{C\cdot {\bf 1}\left(x\leq E-\ell_1\right)}{|E-x|^2} \le
C  (g  \ast \theta_{\ell_1}) (E - x).  
 \ee
Recalling from \eqref{mNdef} and \eqref{thetam}
 that
$$  
\frac{1}{N}\tr \theta_{\ell_1}(H-E) = \frac{1}{\pi N} \im \tr \frac{1}{H-E-i\ell_1}=
\frac{1}{\pi}\im m(E+i\ell_1),
$$
we obtain
\begin{align}\nonumber
\tr f(H)  & \le CN  \eta(E_L-E)   \int_\R  \frac 1 {y^2 + \ell_1^2 }    \im  m( E-y+i\ell_1)  \rd y  
   \\\label{6-27}
&  \le CN^{1/3}  \eta (\log N)^L \int_\R  \frac 1 {y^2 + \ell_1^2 }  \Big[
 \im m_{sc}( E-y+i\ell_1) +  \frac {(\log N)^{CL}} { N \ell_1} \Big] \rd y,    
\end{align}
where, by \eqref{Lambdafinal},  the second inequality holds with a probability larger
 than  $1-C \exp[-c(\log N)^{\phi L}]$ and we also used \eqref{EEL}.
The integral of the second term in the r.h.s is bounded by 
\be\label{6-28}
CN^{1/3}\eta(\log N)^L\int_\R  \frac 1 {y^2 + \ell_1^2 }  \frac {(\log N)^{CL}} { N \ell_1}  \rd y 
\le N^{-2/3}\eta (\log N)^{CL}\ell_1^{-2}\leq N^{-2\e},
\ee
by using  the definitions of $\ell_1$ and $\eta$.

For the first term in the r.h.s of \eqref{6-27} we use the elementary estimate
$$
   \im m_{sc}( E-y+i\ell_1)\le C\sqrt{\ell_1+\big||E-y|-2\big|}.
$$
The integral in the region
$$
A:=\left \{ \big| |E-y|-2\big|\ge\ell_1\right \} 
$$
can be bounded by
$$
   \int_A  \frac {\im m_{sc}( E-y+i\ell_1)}{y^2 + \ell_1^2 }
  \rd y \le C \int_A  \frac{\big||E-y|-2\big|^{1/2}}
 {y^2 + \ell_1^2 }\rd y
  \le C \int_\R \frac {|y|^{1/2} + |E-2|^{1/2}} {y^2 + \ell_1^2 }\rd y
  \le C\Big( \frac{1}{\sqrt{\ell_1}} + \frac{|E-2|^{1/2}}{\ell_1}\Big) .
$$
On the complementary region we have
$$
   \int_{A^c}  \frac 1 {y^2 + \ell_1^2 }
 \im m_{sc}( E-y+i\ell_1) \rd y \le C \sqrt{\ell_1}
 \int_{A^c}  \frac 1 {y^2 + \ell_1^2 }\rd y \le C\ell_1^{-1/2}.
$$
Combining these estimates and using \eqref{E-2N} together with
the definitions of $\ell_1$ and $\eta$ we get
$$
CN^{1/3}  \eta (\log N)^L \int_\R  \frac 1 {y^2 + \ell_1^2 }  
 \im m_{sc}( E-y+i\ell_1)\rd y \le N^{-2\e},
$$
and therefore, together with  \eqref{6-28}, we have
 $\tr f(H)\le 2N^{-2\e}$. Considering \eqref{6-22}, we have thus
 proved Lemma
\ref{lem:21}.  
\qed

\medskip 

Let $q:\R \to\R_+$ be a smooth cutoff  function  such that
\[
q(x) = 1      \quad  \text{if} \quad |x| \le 1/9,   \qquad q(x) = 0   
   \quad  \text{if} \quad |x| \ge 2/9,
\]
and we assume that $q(x)$ is decreasing for $x\ge 0$.

\begin{corollary} \label{23} Suppose the assumptions of Lemma \ref{lem:21} hold and 
$E$ satisfies 
\be\label{condE-}
|E-2|N^{2/3}\leq (\log N)^L.
\ee
Let $\ell: = \frac{1}{2}\ell_1 N^{ 2\e} = \frac{1}{2}N^{-2/3 - \e}$.  
Then  the  inequality 
\be\label{41new}
\tr \chi_{E+ \ell  }  \ast \theta_\eta (H)  -  N^{-\e} \le  \cN (E, \infty)  \le  
\tr  \chi_{E- \ell  }  \ast \theta_\eta (H)  +  N^{-\e}
\ee
holds with a probability bigger than  $1 - C \exp [ - c(\log N)^{\phi L} ]$. Furthermore, 
we have
\be\label{67E}
\E\;  q \left(\tr \chi_{E-\ell}  \ast \theta_\eta  (H) \right)
\le \P (\cN(E,  \infty)  = 0 )  \le \E \; q \left(\tr \chi_{E+\ell}  \ast \theta_\eta  (H) \right) 
+ C\exp{\big[-c(\log N)^{\phi L} \big]}
\ee 
for sufficiently large $N$ independent of $E$ as long as \eqref{condE-} holds.  
Notice that the directions
in the inequalities \eqref{41new} and \eqref{67E} 
are opposite since $q$ is decreasing for positive arguments.

\end{corollary} 

{\it Proof.} For any $E$ satisfying  \eqref{condE-} we have
$E_L-E \gg \ell$ thus $|E-2-\ell|N^{2/3}\le\frac{3}{2}(\log N)^L$
(see \eqref{E-2N}), therefore \eqref{6.10} holds for $E$ 
replaced with $y\in [E-\ell, E]$ as well. We thus obtain
\begin{align*}
 \tr \chi_{E }  (H) & \le \ell^{-1}  \int_{E- \ell }^E  \rd y \;\tr \chi_y (H) \\
& \le  \ell^{-1}  \int_{E- \ell}^E  \rd y \;  
 \tr \chi_y  \ast \theta_\eta (H)  
  +      C\ell^{-1}   \int_{E- \ell }^E  \rd y \left [ N^{-2\e}  
 +   \cN (y-\ell_1, y+\ell_1)  \right] \\
 & \le     \tr \chi_{E- \ell  }  \ast \theta_\eta (H)  +  CN^{- 2 \e}+ 
   C  \frac{\ell_1}{\ell}  \cN (E- 2\ell , E + \ell ) 
\end{align*} 
with a probability larger than $1 - C \exp [ - c(\log N)^{\phi L} ]$. 
From \eqref{nn}, \eqref{condE-}, $\ell_1/\ell=2N^{-2\e}$ and $ \ell \le N^{-2/3}$, we can bound 
$$
  \frac{\ell_1}{\ell}  \cN (E- 2\ell , E+ \ell ) 
 \le N^{1-2\e} \int_{E-2\ell}^{E+\ell}\varrho_{sc}(x)\rd x + N^{-2\e}(\log N)^{L_1}
 \le \frac{1}{2}N^{-\e}
$$
with a very high probability, where we estimated the explicit integral
using that the integration domain is in a $CN^{-2/3}(\log N)^L$-vicinity of
the edge at 2.
We have thus proved 
\[
\cN(E, E_L) =   \tr \chi_{E }  (H) \le  \tr  \chi_{E- \ell  }  \ast \theta_\eta (H)  +  N^{-\e}.
\]
By \eqref{6-1}, we can replace $\cN(E, E_L)$ by $\cN(E,  \infty)$
 with  a change of probability of at most 
$C\exp [ -c (\log N)^{\phi L}]$. This proves the upper bound of \eqref{41new}
 and the lower bound can be proved similarly.

On the event that \eqref{41new} holds, the condition  
$\cN(E, \infty) = 0$ implies that $\tr \chi_{E+\ell}  \ast \theta_\eta (H) \leq 1/9$. Thus we have
\be\label{Pleft}
\P\left(\cN(E, \infty)=0\right)\leq \P\left(\tr  \chi_{E+\ell} 
 \ast \theta_\eta (H) \leq 1/9\right)+C\exp[-c(\log N)^{\phi L}].
\ee
Together with the Markov inequality, 
this proves the upper bound in \eqref{67E}. For the lower bound, we use
$$
\E\;  q \big(\tr \chi_{E-\ell}  \ast \theta_\eta  (H) \big)
\le \P \big( \tr \chi_{E-\ell}  \ast \theta_\eta  (H)\le 2/9\big)
\le \P \big( \cN(E,\infty)\le 2/9+N^{-\e}\big)
 =\P \big( \cN(E,\infty)=0\big),
$$
where we used the upper bound from \eqref{41new} and that $\cN$ is an integer.
This  completes the proof of the Corollary. 
\qed

\bigskip

\subsection{Green Function Comparison Theorem}\label{sec:green}

Recalling that $\theta_\eta(H)= \frac{1}{\pi}\im G(i\eta)$,
Corollary  \ref{23}  bounds the probability of $\cN(E,\infty)=0$ in terms of  the expectations
 of two functionals of Green functions. In this subsection, we show that the difference 
between the expectations of these  functionals
w.r.t. two probability distributions 
$\bv$ and $\bw$ is  negligible assuming their second moments
match. The precise statement is the following Green function
 comparison theorem on the edges. All statements are formulated
for the upper spectral edge 2, but with the same proof they hold for the lower
spectral edge $-2$ as well.

\begin{theorem} [Green function comparison theorem on the edge] \label{GFCT}
Suppose that  the assumptions of Theorem  \ref{twthm}, including \eqref{2m},  hold.
Let    $F:\R\to \R$ be a function whose derivatives satisfy  
\be\label{gflowder}
\max_{x}|F^{(\al)}(x)|\left(|x|+1\right)^{-C_1} \leq C_1,\qquad \al=1,\;2,\;3,\;4
\ee
with some constant $C_1>0$.
Then there exists $\e_0>0$ depending only on $C_1$ such that for any $\e<\e_0$
and for any real numbers $E$, $E_1$ and $E_2$ satisfying
$$
   |E-2|\leq N^{-2/3+\eps}, \qquad |E_1-2|\leq N^{-2/3+\eps}, \qquad 
 |E_2-2|\leq N^{-2/3+\eps}, \qquad
$$
and setting $\eta = N^{-2/3-\e}$,
we have
\begin{align}\label{maincomp}
\Bigg|\E^\bv  F  
\left (   
N \eta \im m (z)    \right )  &  -
\E^\bw  F  
\left (  N \eta \im m (z)      \right )  \Bigg|
\le  C N^{-1/6+C \e},\qquad z=E+i\eta,
\end{align}
and 
\be\label{c1}
\left|\E^\bv  F\left(N \int_{E_1}^{E_2} \rd y \;  \im m (y +i\eta)\right)-\E^\bw 
 F \left(N \int_{E_1}^{E_2}  \rd y \; \im m (y+i\eta ) \right)\right|\leq C N^{-1/6+C \e}
\ee 
for some constant $C$ and large enough $N$ depending only on  $C_1$,
  $\ttau$, $\delta_\pm$ and $C_0$ (in \eqref{1.3}).

\end{theorem}

\bigskip 

Theorem \ref{GFCT} holds in a much greater generality. We state the following extension 
which can be used to prove \eqref{twa}, the generalization  of Theorem \ref{twthm}. 
The class of functions $F$ in the following 
theorem can be enlarged  to allow some polynomially increasing functions
 similar to \eqref{gflowder}. 
But for the application to prove  \eqref{twa}, the following form is sufficient.
The proof of Theorem \ref{GFCT2} is similar to that of Theorem \ref{GFCT} and will be omitted.

\begin{theorem}  \label{GFCT2}
Suppose that  the assumptions of Theorem  \ref{twthm}, including \eqref{2m},  hold.
Fix any $k\in \N_+$ and
let  $F: \R^k \to \R $ be a bounded smooth function with bounded derivatives.
Then for any sufficiently small $\e$ there exists a $\delta>0$ such that
for any sequence of real numbers $E_k < \ldots < E_{1}< E_0 $
with $|E_j-2|\le N^{-2/3+\e}$, $j=0,1,\ldots, k$, we have
\be\label{c11}
\left| \Big ( \E^\bv- \E^\bw \Big )   F\left(N \int_{E_1}^{E_0} \rd y  \im m (y +i\eta),  
 \ldots, N \int_{E_k}^{E_0} \rd y  \im m (y+i\eta ) 
\right) \right|\leq N^{-\delta}.
\ee

\end{theorem}

\bigskip

Assuming that  Theorem \ref{GFCT} holds, we now prove Theorem \ref{twthm}. 

\medskip 

\noindent 
{\it Proof of Theorem \ref{twthm}.}
As we discussed in \eqref{6-1} and \eqref{6-2}, we can assume that \eqref{25} holds
for the parameter $s$.
We define $E:=2+sN^{-2/3}$ that satisfies \eqref{condE-}.
We define $E_L$ as in \eqref{defEL} with the  $L$ such that \eqref{6-1} 
and \eqref{6-2} hold.  For simplicity, we set $\xi = \phi L$ and note that $\xi\ge 2$
for sufficiently large $N$. 
 With the left side of \eqref{67E}, for any sufficiently small $\e>0$, we have 
\be
\E^\bw\;  q \left(\tr \chi_{E-\ell}  \ast \theta_\eta  (H) \right)
\le \P^\bw (\cN(E, \infty)  = 0 )
\ee
with the choice 
$$
   \ell: = \frac{1}{2}N^{-2/3-\e}, \qquad \eta:= N^{-2/3-9\e}.
$$
The bound \eqref{c1}
applying to the case $E_1=E-\ell$ and 
$E_2=E_L$
shows that  there exist $\delta>0$, for sufficiently small $\e>0$, such that
\be\label{645}
\E^\bv\;  q \left(\tr \chi_{E-\ell}  \ast \theta_\eta  (H) \right)\leq 
\E^\bw\;  q \left(\tr \chi_{E-\ell}  \ast \theta_\eta  (H) \right)+N^{-\delta}
\ee
(note that $9\e$  plays the role of the $\e$ in the Green function comparison theorem).
Then applying the right side of \eqref{67E} in Lemma \ref{23}, with $\xi=\phi L\ge 2$,
 to the l.h.s of \eqref{645}, we have
\be
\P^\bv (\cN(E-2\ell, \infty)  = 0 ) \leq 
\E^\bv\;  q \left(\tr \chi_{E-\ell}  \ast \theta_\eta  (H) \right)+C\exp{\big[-c(\log N)^2\big]}.
\ee
 Combining these inequalities, we have
\be\label{Pbv}
\P^\bv (\cN(E- 2\ell, \infty)  = 0 ) \leq\P^\bw (\cN(E, \infty)  = 0 )+2N^{-\delta}
\ee
for sufficiently small $\e>0$ and  sufficiently large $N$. Recalling that 
$E=2+sN^{-2/3}$, 
this proves the first 
inequality  of \eqref{tw} and, 
by switching the role of  $\bv, \bw$, the second inequality of \eqref{tw} as well. 
 This completes the proof of Theorem \ref{twthm}.
\qed

\bigskip 

\noindent 
{\it Proof of Theorem \ref{GFCT}.}
Notice that 
\be\label{c2} 
N \int_{E_1}^{E_2} \rd y  \, \im m (y+i\eta ) = \eta  \int_{E_1}^{E_2} \rd y   \, \tr G(z) \ov G(z),
 \quad z = y+ i \eta .
\ee
We now set up notations to  replace the matrix elements one by one. This step is identical 
for the proof of both  \eqref{maincomp} and \eqref{c1},
 and we will use the notations of the case  \eqref{maincomp}
which are less involved.

Fix a bijective ordering map on the index set of
the independent matrix elements,
\be\label{order}
\phi: \{(i, j): 1\le i\le  j \le N \} \to \Big\{1, \ldots, \gamma(N)\Big\} , 
\qquad \gamma(N): =\frac{N(N+1)}{2},
\ee
and denote by  $H_\gamma$  the generalized Wigner matrix whose matrix elements $h_{ij}$ follow
the $v$-distribution if $\phi(i,j)\le \gamma$ and they follow the $w$-distribution
otherwise; in particular $H_0 = H^{(v)}$ and $ H_{\gamma(N)} = H^{(w)}$.  
The specific choice of the ordering map \eqref{order} is irrelevant; in the following argument, 
 $\phi$ could be any bijective ordering map.
 With $\eta=N^{-2/3-\e}$, 
it was proved in \eqref{Lambdaodfinal} that for any constant $\xi>0$, 
\be\label{basic}
\P\left(\max_{0 \le \gamma \le \gamma(N)} \max_{1 \le k,l \le N}  
 \max_{E}\left |  \left (\frac 1 {  H_{\gamma}-E- i \eta} \right )_{k l }
 -\delta_{kl}m_{sc}(E+i \eta)
\right |\le 
N^{-1/3+2\e} \right)\geq 1-C\exp[{-c(\log N)^\xi}]
\ee
with some constants $C, c$ and large enough $N\ge N_0$ (may depend on $\xi$).
 The last maximum in the
formula \eqref{basic}
runs over all $E$ satisfying $|E-2|\le N^{-2/3+\e}$. When applying \eqref{Lambdaodfinal},
 we have used $(\log N)^{4L}(N\eta)^{-1}\leq N^{-1/3+2\e}$
and that 
\be\label{imscbound}
\im m_{sc}(E+i\eta) \le \sqrt{|E-2|+\eta}\le CN^{-1/3+\e/2}
\ee
for $|E-2|\le CN^{-2/3+\e}$.

\bigskip 
We set $z=E+i\eta$ where $|E-2|\le CN^{-2/3+\e}$ and $\eta=N^{-2/3-\e}$.
{F}rom \eqref{basic}, \eqref{imscbound} and the identity
$$
  \im m(z) = \frac{1}{N}\im \tr G=\frac{\eta}{N}\sum_{ij}G_{ij}\overline{G_{ij}},
$$
 we have that
 \be
\Big|\eta^2\sum_{ij}G_{ij}\overline{G_{ij}}\Big|=
 | N \eta \im m (z) | \leq CN^{2\e}
  \ee
  and
  \be
  \Big|\eta^2\sum_{i=j}G_{ij}\overline{G_{ij}}\Big|
  \le N\eta^2 \big( |m_{sc}| + CN^{-1/3+2\e}\big) \leq CN^{-1/3-2\e}
  \ee
hold with a probability larger than $1-C\exp[{-c(\log N)^\xi}]$. Since the derivative of $F$ 
is bounded as in \eqref{gflowder},  there exists $C$ depending on  $F$,  $\ttau$, $\delta_\pm$ 
and $C_0$ such that 
\be
\left|\E F\left(\eta^2\sum_{ij}G_{ij}\overline{G_{ij}}\right)-
\E F\left(\eta^2\sum_{i\neq j}G_{ij}\overline{G_{ij}}\right)\right|\leq CN^{-1/3+C\e}.
\ee
This holds for both the $\bv$ and the $\bw$ ensembles.

To show \eqref{maincomp}, we only need to  prove that for small enough $\e$, 
there exists $C$ depending on $F$,  $\ttau$, $\delta_\pm$ and $C_0$ such that
\begin{align}\label{maincompnew}
\Bigg|\E^\bv F\left (   
\eta^2 \sum_{i\neq j}G_{ij}^{(v)}\overline {G_{ji}^{(v)}}\right)   -
  \E^\bw F\left ( G^{(v)} \to  G^{(w)}\right )  \Bigg| \le CN^{-1/6+C\e},
\end{align}
where $G^{(v)}$ and $G^{(w)}$ denote the Green functions of the $H^{(v)}$ and $H^{(w)}$,
respectively.
Here the shorthand notation $F\left ( G^{(v)} \to  G^{(w)}\right ) $
means that we consider the same argument of $F$ as in 
the first term in \eqref{maincompnew}, but all $G^{(v)}$
terms are replaced with $G^{(w)}$.
In fact, the upper index notation is slightly superfluous since the Green function
is the same, only the underlying ensemble measure changes, but
we wish to emphasize the difference between the two ensembles in this way as well.

Similarly, for \eqref{c1}, we only need to  prove that for small enough $\e$, there 
exists $C$ depending on $F$,  $\ttau$, $\delta_\pm$ and $C_0$ such that
\be\label{c1new}
\left|
\E^\bv  F\left(N \int_{E_1}^{E_2} \rd y  \left (   
\eta \sum_{i\neq j}G_{ij}^{(v)}\overline {G_{ji}^{(v)}}(y+i\eta)\right)\right) -
\E^\bw  F\left ( G^{(v)} \to  G^{(w)}\right )
\right|\leq C N^{-1/6+C \e}.
\ee
Consider the telescopic sum of differences of expectations 
\begin{align}\label{tel}
\E \, F \left (\eta^2 \sum_{i\neq j}\left(\frac{1}{H^{(v)}-z}\right)_{ij}
\overline{\left(\frac{1}{H^{(v)}-z}\right)_{ji}} \right )   - 
 & \E \, F \left  ( H^{(v)}\to H^{(w) }\right )  \\
= & \sum_{\gamma=1}^{\gamma(N)}\left[  \E \, F \left ( H^{(v)}\to H_{\gamma }\right ) 
-  \E \, F \left ( H^{(v)}\to H_{\gamma-1 } \right ) \right] . \non
\end{align}
Let $E^{(ij)}$ denote the matrix whose matrix elements are zero everywhere except
at the $(i,j)$ position, where it is 1, i.e.,  $E^{(ij)}_{k\ell}=\delta_{ik}\delta_{j\ell}$.
Fix a $\gamma\ge 1$ and let $(a,b)$ be determined by  $\phi (a, b) = \gamma$. For simplicity
to introduce the notation, 
we assume that $a\neq b$. The  $a=b$ case
can be treated similarly. We note the total number 
of the diagonal terms is $N$ and the one of the off-diagonal terms is $O(N^2)$.  
We will compare $H_{\gamma-1}$ with $H_\gamma$ for each $\gamma$
and then sum up the differences according to \eqref{tel}.

Note that these two matrices differ only in the $(a,b)$ and $(b,a)$ matrix elements 
and they can be written as
\be\label{defHg1}
    H_{\gamma-1} = Q + \frac{1}{\sqrt{N}}V, \qquad V:= v_{ab}E^{(ab)}
+ v_{ba}  E^{(ba)}
\ee
$$
    H_\gamma = Q + \frac{1}{\sqrt{N}} W, \qquad W:= w_{ab}E^{(ab)} +
   w_{ba} E^{(ba)},
$$
with a matrix $Q$ that has zero matrix element at the $(a,b)$ and $(b,a)$ positions and
where we set $v_{ji}:= \ov v_{ij}$ for $i<j$ and similarly for $w$.
Define the  Green functions
\be\label{defG}
      R: = \frac{1}{Q-z}, \qquad S:= \frac{1}{H_{\gamma-1}-z}, \qquad T:= \frac{1}{H_{\gamma}-z}.
\ee

We first claim that the estimate \eqref{basic} holds for the Green function $R$ as well. 
More precisely, the probability of the event
\be
  \Omega_R:= \max_{1 \le k,l \le N}  
 \max_{E}\big| R_{k l}(E+i\eta) -\delta_{kl}m_{sc}(E+i \eta)
\big|\ge N^{-1/3+2\e}
\label{defOMR}
\ee
(where $\max_E$ is the maximum over all $E$ with $|E-2|\le N^{-2/3+\e}$)
satisfies
\be
  \P (\Omega_R)\le  C\exp{\big[-c(\log N)^\xi\big]}
\label{POMR}
\ee
for any fixed $\xi>0$.
To see this, we use the 
resolvent expansion 
\be\label{relRS}
 R = S  +   N^{-1/2}SV S +   N^{-1}(SV)^2 S+ \ldots + N^{-9/5} (SV)^9S+
N^{-5} (SV)^{10} R.
\ee
Since $V$ has only at most two nonzero elements, when
computing the $(k,\ell)$ matrix element of this matrix identity,
each term is a sum of finitely many terms (i.e. the number of summands
is $N$-independent) that  involve
matrix elements of $S$ or $R$ and $v_{ij}$, e.g.  $(SVS)_{k\ell} =S_{ki} v_{ij} S_{j\ell}
+ S_{kj} v_{ji} S_{i\ell}$.  Using the bound \eqref{basic} for the $S$ matrix elements,
the subexponential decay for $v_{ij}$ and 
 the trivial bound $|R_{ij}| \le  \eta^{-1}\leq N$, we obtain that 
the estimate \eqref{basic} holds for $R$ as well. 

\medskip

After having introduced these notations, we are in a position to give a
heuristic power counting argument that is the core of the proof.
In particular, we can explain the origin of the second moment matching condition. 
Take $F(x)=x$ for simplicity.
A resolvent expansion  analogous to \eqref{relRS}
 gives 
\be\label{Eeta}
\E \eta \sum_i \im  S_{ii} = \eta \, \E \, \im \sum_i \Big [ R_{ii}   -   N^{-1/2}( RV R)_{ii} +  
 N^{-1} ((RV)^2 R)_{ii}+ \ldots  \Big ]
\ee
which is an expansion in the order of $N^{-1/2}$ since the matrix $V$ 
contains only a few nonzero elements of size $N^{-1/2}$.
Notice that $\eta \sum_i \im  S_{ii}$ estimates the number of eigenvalues near $E$ in a window 
of size $\eta$.  For the two ensembles to have the 
same local eigenvalue distribution on scale $\eta$,
 we need the error term to be less than order one
even after performing the telescopic sum.
 In the bulk, $\eta$ has to be chosen as $\eta \sim N^{-1}$ 
and we can view $\eta \sum_i$ as order one
 in the power counting.  Since in the telescopic expansion
we will have $N^2$ terms  to sum up, 
we need that the error term of the expansion is $o(N^{-2})$
for each replacement step, i.e., for each fixed 
label $(a, b)$. This explains the usual condition of four moments to be identical
for the Green function comparison theorem in the bulk \cite{EYY}
since the first four terms in \eqref{Eeta} has to be equal.
Near the edges, i.e., at energies $E$ 
with  $|E-2|\lesssim N^{-2/3}$,
 the correct local scale is $\eta\sim N^{-2/3}$ and
 the strong local semicircle law \eqref{Lambdaodfinal}
implies that the off-diagonal Green functions are of order 
$N^{-1/3}$  and the diagonal Green functions are bounded. 
Hence the size of the third order term 
$ \eta \, \E  \sum_i N^{-3/2}((RV)^3 R)_{ii}$ is of order 
\[
\eta N N^{-3/2} N^{-2/3}  = N^{-2 + 1/6} 
\]
where we used that,  for a generic label $(a, b)$,  there are
at least two off-diagonal resolvent terms 
in $ ((RV)^3 R)_{ii}$. Notice that the error term is still larger than $N^{-2}$, required 
for summing over $a, b$ (this argument would be sufficient if we had a matching of
three moments and only the fourth order term in \eqref{Eeta} needed to be estimated).
The key observation is that  the leading term, which gives this order $N^{-2 + 1/6} $,
has actually almost zero expectation which
 improves the error to be less than $o(N^{-2})$.  This is due 
to the fact that with the help of \eqref{basic}
we are able to follow the main term in the diagonal elements of the Green functions
and thus compute the expectation fairly precisely. Notice that similar  reasons apply to the proof of 
Lemma \ref{motN} in Section \ref{sec:Z}.

\subsection{Main Lemma}\label{sec:mainlemma}

The key step to the proof of Theorem \ref{GFCT} is the following lemma:

\begin{lemma}\label{lemGamma} Fix an index $\gamma$,
 recall the definitions of $Q$, $R$ and $S$ 
from \eqref{defG} and suppose first that  $\gamma =\phi(a,b)$
with $a\ne b$. For any small $\e>0$ and
under the assumptions in Theorem \ref{GFCT} on $F$, $E$,
$E_{1}$ and $E_2$,   there exists $C$ depending 
on $F$,  $\ttau$, $\delta_\pm$ and $C_0$ (but independent of $\gamma$) and
there exist  constants $A_N$ and $B_N$, depending  on the distribution of the Green function $Q$,
denoted by $\mbox{dist}(Q)$, 
and on the second moments of
 $v_{ab}$, denoted by $m_{2}(v_{ab})$, such that
\be\label{new612}
\left|\E \, F \left (\eta^2 \sum_{i\neq j}S_{ij}\overline{S}_{ji}(z) \right ) -
\E \, F \left (\eta^2 \sum_{i\neq j}R_{ij}\overline{R}_{ji}(z) \right )-
 A_N\big(m_2(v_{ab}), \mbox{dist}(Q)\big)\right|\leq CN^{-13/6+C\e},\;
\ee
with $z=E+i\eta$, $\eta=N^{-2/3-\e}$, and
\begin{align}\label{c3} 
\Bigg|\E  F\left( \eta  \int_{E_1}^{E_2} \rd y  
\sum_{i\neq j}S_{ij}\overline{S}_{ji}(y+i\eta)\right)
-\E  F\Bigg( \eta  \int_{E_1}^{E_2} \rd y  & \sum_{i\neq j}R_{ij}\overline{R}_{ji}(y+i\eta)\Bigg)
 \\ & -B_N\big(m_2(v_{ab}), \mbox{dist}(Q) \big)  \Bigg|\leq CN^{-13/6+C\e}\non
\end{align}
for large enough $N$ (independent of $\gamma$). The constants $A_N$ and $B_N$ may also depend 
on $F$ and on the parameters  $\ttau$, $\delta_\pm$ and $C_0$, but they depend on 
the centered random variable $v_{ab}$ only through its second moments.

Finally, if $a=b$, i.e. $\gamma =\phi(a,a)$, then the bounds \eqref{new612} and \eqref{c3} hold
with $CN^{-11/6+C\e}$ standing on their right hand side.
\end{lemma}

\bigskip

The same estimates hold if $S$ is replaced by $T$ everywhere
and note that $Q$ is independent of $v_{ab}$ and
$w_{ab}$. Since $m_2(v_{ab})=m_2(w_{ab})$, we obviously 
have that $A_N\big(m_2(v_{ab}), \mbox{dist}(Q)\big)
=A_N\big(m_2(w_{ab}), \mbox{dist}(Q)\big)$.
Thus we get from Lemma \ref{lemGamma} that in case of $a\ne b$
\be\label{RSest}
  \left|\E\, F \left (\eta^2 \sum_{i\neq j}S_{ij}\overline{S}_{ji}(z) \right ) -
  \E\, F \left (\eta^2 \sum_{i\neq j}T_{ij}\overline{T}_{ji}(z) \right ) \right|
  \le CN^{-13/6+C\e}
\ee
and a similar bound for the quantity \eqref{c3}. In case of $a=b$, the
estimate is only $CN^{-11/6+C\e}$.
Recalling the definitions of $S$ and $T$ from \eqref{defG},
the bound \eqref{RSest} compares
the expectation of a function of the resolvent of $H_\gamma$ and that of $H_{\gamma-1}$.
 The telescopic summation then implies
\eqref{maincompnew} and  \eqref{c1new}
since the number of summands with $a\ne b$ is of order $N^2$ but the number
of summands with $a=b$ is only $N$. This completes
 the proof of Theorem \ref{GFCT}.
\qed 

\bigskip

{\it Proof of Lemma \ref{lemGamma}. } We will only prove the more complicated case 
\eqref{c3}; the 
proof can be adapted easily  for  \eqref{new612} which will be omitted. 
Similarly to  $\Omega_R$ from \eqref{defOMR}, 
define
$$
   \Omega_S:= \max_{1 \le k,l \le N}  
 \max_{E}\big| S_{k l}(E+i\eta) -\delta_{kl}m_{sc}(E+i \eta)
\big|\ge N^{-1/3+2\e},
$$
where $\max_E$ is the maximum over all $E$ with $|E-2|\le N^{-2/3+\e}$.
Since $S$ is the Green function of $H_{\gamma-1}$, 
we obtain from \eqref{basic} directly that
\be
   \P (\Omega_S)\le  C\exp{\big[-c(\log N)^\xi\big]}
\label{POMS}
\ee
for any fixed $\xi>0$. Finally, set
\be\label{OMDEF}
\Omega_\bv:=\{ |v_{ab}|\ge N^{\e}\sigma_{ab}\}, \qquad \mbox{and}\qquad
\Omega:= \Omega_R \cup \Omega_S\cup \Omega_\bv.
\ee
Using \eqref{POMR}, \eqref{POMS} and the subexponential decay
of $v_{ab}$, we obtain 
\be\label{28}
\P\left(\Omega\right)\le C\exp{\big[-c(\log N)^\xi\big]}.
\ee
for any fixed $\xi>0$ and large enough $N$.
Since the arguments of $F$ in \eqref{c3}
are bounded by $CN^{2+2\e}$ and $F(x)$ increases at most
polynomially, it is easy to see that the contribution
of the set $\Omega$
to the expectations  in \eqref{c3} is negligible.
We can thus concentrate on the set $\Omega^c$.

Define $x^S$ and $x^R$ by 
\be\label{defxs}
x^{S} :=   \eta  \int_{E_1}^{E_2} \rd y \sum_{i\neq j}S_{ij}\overline{S}_{ji}(y+i\eta)
,\,\,\,\,x^{R}:= \eta  \int_{E_1}^{E_2} \rd y \sum_{i\neq j}R_{ij}\overline{R}_{ji}(y+i\eta),
\ee 
and
decompose  $x^S$ into three parts
\be
x^{S} =x^{S}_2+x^S_1+x^S_0, \quad 
x^S_k:=   \eta  \int_{E_1}^{E_2} \rd y  \sum_{i\neq j, \; |\{i,j\}\cap \{a,b\}|=k } 
 S_{ij}\overline{S}_{ji}(y+i\eta),
\ee
and $x^R_k$ are defined similarly. 
Here $k=|\{i,j\}\cap \{a,b\}|$ is the number of times  $a$ and $b$ appears among the
 summation indices  $i,j$ (if $a=b$ then we count it only once); clearly $k=0, 1$ or $2$. 
The number of the  terms in the summation  of  $x^S_k$ is $O(N^{2-k})$ 
since $a$ and $b$ are fixed.  From the resolvent expansion, we have 
\be\label{SR-N}
    S =  R -  N^{-1/2} RVR+  N^{-1} (RV)^2R -  N^{-3/2} (RV)^3R+  N^{-2} (RV)^4S.
\ee
In the following formulas we will omit the spectral parameter from the notation
of the resolvents. The spectral parameter is always $y+i\eta$
with $y\in [E_1,E_2]$, in particular $|y-2|\le N^{-2/3+\e}$.

If $|\{i,j\}\cap \{a,b\}|=k$,   using \eqref{SR-N} and \eqref{basic}, we have in $\Omega^c$
\be\label{51}
\big|N^{-m/2}\big[(RV)^m R\big]_{ij}\big|\leq C_m N^{-m/2+3m\e}N^{- (2-k)/3},\,\,\,  m \in \N_+, 
\; k = 0,1, 2
\ee
for some constants $C_m$. Furthermore, we can replace the last $R$ 
by $S$, i.e.,  we also have 
\be\label{666}
\left|  N^{-2} \big[(RV)^4S\big]_{ij}\right|\leq CN^{-2-(2-k)/3+C\e}.    
\ee
Therefore,   in $\Omega^c$ we have, 
\be
 |x^S_{k}-x^R_{k}|\leq C N^{ -5/6-2k/3+C\e},  \quad k = 0,1, 2.
 \ee
Inserting these bounds into  the Taylor expansion of $F$
and keeping only the terms larger than $o(N^{-2})$,
we obtain  
\be\label{FF123}
\left|\E  [ F(x^S)- F(x^R) ] 
-\E\left( F'(x^R) (x^S_0-x^R_0) 
+\frac12 F''(x^R)  (x^S_0-x^R_0 )^2+ F'(x^R)(x^S_1-x^R_1)\right) \right|\leq CN^{-13/6+C\e},
\ee
where we used the remark after \eqref{28} to treat the contribution on the
event $\Omega$.
Since there is no $x_2$ appearing in \eqref{FF123}, we can focus on the case $k=0$ or $1$. 

For $k=0$ or $1$, we define $Q_\ell^{(k)}$ for $\ell=1$, $2$ or $3$, as the sum of the terms
 in $x^S_k-x^R_k$ in which the total number of $v_{ab}$ or $v_{ba}$ is   $\ell$, i.e., 
\begin{align}
Q_1^{(k)} & :=-N^{-1/2}\eta  \int_{E_1}^{E_2} \rd y  \sum_{|\{i,j\}\cap \{a,b\}|=k }
\left(R_{ij}\overline{(RV R)_{ji}}+(RVR)_{ij}\overline{R_{ji}}\right) \\
Q_2^{(k)} & :=N^{-1}\eta  \int_{E_1}^{E_2} \rd y \sum_{|\{i,j\}\cap \{a,b\}|=k }
\Bigg(R_{ij}\overline{((RV)^2 R)_{ji}}+((RV)^2R)_{ij}\overline{R_{ji}}+
(RVR)_{ij}\overline{(RV R)_{ji}}\Bigg) \\ \label{Q3}
Q_3^{(k)} & :=-N^{-3/2}\eta  \int_{E_1}^{E_2} \rd y \sum_{|\{i,j\}\cap \{a,b\}|=k}
\Bigg(R_{ij}\overline{((RV)^3 R)_{ji}}+\overline{R_{ji}}((RV)^3R)_{ij}+
((RV)^2R)_{ij}\overline{(RV R)_{ji}}\\\nonumber
&\quad\quad\quad\quad\quad\quad\quad\quad\quad\quad\quad\quad\quad\quad\quad\quad
+(RV R)_{ij}\overline{((RV)^2R)_{ji}}\Bigg) .
\end{align}

By these definitions and \eqref{51},  we have
\be\label{byi1232}
Q_\ell^{(k)} \leq N^{-\ell/2-1/3 -2k/3+C\e} \,\,\, \,\,\,{\rm in } \,\,\,\Omega^c.  
\ee
Furthermore, with \eqref{51} and \eqref{666}, we decompose  $x^S_k-x^R_k$  as
\be\label{yi1232}
x^S_k-x^R_k =  Q_1^{(k)}+Q_2^{(k)}+ Q_3^{(k)}+O(N^{-7/3+C\e}).
\ee
The last two terms in \eqref{Q3}  can also be bounded by using \eqref{51}, i.e.,
\be\label{spQ3k}
Q_3^{(k)}  =O(N^{-13/6+C\e})-N^{-3/2}\eta  \int_{E_1}^{E_2} \rd y 
\sum_{|\{i,j\}\cap \{a,b\}|=k}\left(R_{ij}
\overline{((RV)^3 R)_{ji}}+\overline{R_{ji}}((RV)^3R)_{ij}\right)
\,\,\, \,\,\,{\rm in } \,\,\,\Omega^c.
\ee
 Inserting \eqref{byi1232} and \eqref{yi1232} into the second term of the 
l.h.s of \eqref{FF123}, with the bounds on the derivatives of $F$, we have 
\begin{align}\label{6755}
& \E\left( F'(x^R)(x^S_0-x^R_0)+ F'(x^R)(x^S_1-x^R_1) +
\frac12 F''(x^R)(x^S_0-x^R_0)^2\right)\\\nonumber
& = B +  \E F'(x^R)  Q_3^{(0)}   + O\left(N^{-13/6+C\e}\right),
\end{align}
where 
\begin{align}\label{Adef}
B  & := \E\left( \sum_{k=0, 1}  F'(x^R) [Q_1^{(k)}+Q_2^{(k)}] 
+\frac12 F''(x^R) [Q_1^{(0)}]^2  \right) \\
& =  \E\left( \sum_{k=0, 1}  F'(x^R) \E_{v_{ab}}  [Q_1^{(k)}+Q_2^{(k)}] 
+\frac12 F''(x^R) \E_{v_{ab}} [Q_1^{(0)}]^2  \right) \non
\end{align} 
depends on $v_{ab}$ only through its expectation (which is zero) and on its
second moments.

First we give a trivial estimate on $Q_3^{(0)}$. In case $i,j$ are distinct
from $a$ and $b$, it is easy to see by writing out  terms
 in \eqref{spQ3k} that they contain 
at least three offdiagonal elements of resolvent; for example in the term
$R_{ij}\ov{ R_{ja} v_{ab} R_{ba} v_{ab} R_{ba} v_{ab} R_{bi}}$,
appearing in $R_{ij} \overline{((RV)^3 R)_{ji}}$, the resolvent matrix elements
$R_{ij}\ov{ R_{ja} R_{bi}}$ are off-diagonal.
Each off-diagonal matrix element of $R$ is bounded by $N^{-1/3+2\e}$
in $\Omega_R^c$, while
the diagonal terms can be estimated by $|m_{sc}|$, hence by a constant, 
at a negligible error in the set $\Omega^c\subset\Omega_R^c$.
This shows that each term in the integrand in \eqref{spQ3k}
is bounded by $C\big[ N^{-1/3+2\e}\big]^3$.
Note that  every estimate is uniform in $y$, the real
part of the spectral parameter, as long as $|y-2|\le N^{-2/3+\e}$.
Estimating $F'$ trivially, 
we thus  obtain
$$
\left|\E  [ F(x^S)- F(x^R) ] -B \right| \le CN^{-11/6+C\e}.
$$
This bound proves Lemma \ref{lemGamma} for the case $a=b$.

For $a\ne b$ this estimate would not be sufficient since
the number of pairs $a\ne b$ to sum up
in the telescopic summation is of order $N^2$.
However, we will show that in this case the expectation of the $Q_3^{(0)}$ term 
is of smaller order  than the trivial estimate gives.

From now on we assume that $a\ne b$.
By \eqref{spQ3k} we have, in $\Omega^c$ that
\begin{align}\label{Q30}
Q_3^{(0)}  &=O(N^{-13/6+C\e})-N^{-3/2}\eta  \int_{E_1}^{E_2} \rd y \sum_{j\neq a,b}
\;\; \sum_{i\neq j,a,b}  \non \\  &
\qquad \qquad \qquad
\Bigg[\Big(R_{ij}\overline{R_{ja}v_{ab}R_{bb}v_{ba}R_{aa}v_{ab}R_{bi}}  
 +R_{ia}v_{ab}R_{bb}v_{ba}R_{aa}v_{ab}R_{bj}\overline{R_{ji}}\Big)
+(a\leftrightarrow b)\Bigg] \non
\\ \non
&=O(N^{-13/6+C\e})-  N^{-3/2} \eta  \int_{E_1}^{E_2} \rd y   \sum_{j\neq a,b} \;\; \sum_{i\neq j,a,b}
\\
& \qquad\qquad\qquad
\Bigg[\left( \overline{m^2_{sc}}R_{ij}\overline{R_{ja}R_{bi}}
+m^2_{sc}R_{ia}R_{bj}\overline{R_{ji}}\right)|v_{ab}|^2v_{ab}
+(a\leftrightarrow b)\Bigg].  
\end{align}
Note that we explicitly collected those terms that contain
the most diagonal elements of $R$; these
are the main terms of $Q_3^{(0)}$. There are several other terms,
for example $R_{ij}\ov{ R_{ja} v_{ab} R_{ba} v_{ab} R_{ba} v_{ab} R_{bi}}$,
that appear in the expansion of $R_{ij}\ov {[(RV)^3R]_{ji}}$, but these are lower
order terms and can be directly included in the error term.
In the second step in \eqref{Q30} we estimated the diagonal terms by $m_{sc}$
at a negligible error in the set $\Omega^c\subset\Omega_R^c$.

We note that  $v_{ab}$ is independent of $R$ and $\E_{v_{ab}}|v_{ab}|^2v_{ab}=O(1)$. 
Combining \eqref{Q30} with \eqref{6755} and \eqref{FF123}, we obtain 
\begin{align}\label{680}
&\left|\E  [ F(x^S)- F(x^R) ] 
-B \right|\\\nonumber
\leq& CN^{-13/6+C\e}+\big| \E F'(x^R)  Q_3^{(0)}\big| \\ \nonumber
\leq &CN^{-13/6+C\e}+CN^{-5/6+C\e}\max_y\max_{i\neq j:\{i,j\}\cap \{a,b\}=\emptyset}
\Big[
\big|\E F'(x^R)  R_{ij}\overline{R_{ja}R_{bi}}\big|
+\big| \E F'(x^R) R_{ia}R_{bj}\overline{R_{ji}}\big| +(a\leftrightarrow b)\Big],
\end{align}
where we used the trivial bounds on $F'$ and $m_{sc}$
and we agsin used that every estimate is uniform in $y$, the real
part of the spectral parameter, as long as $|y-2|\le N^{-2/3+\e}$.
As before, $\max_y$ in the last line of \eqref{680} indicates
maximum over all $y$ with $|y-2|\le N^{-2/3+\e}$ and the spectral
parameter of all resolvents is $y+i\eta$.

The following lemma shows that
 the expectation of the product of the 
off-diagonal terms in \eqref{680}
is of smaller order than the trivial estimate gives.

\begin{lemma} \label{lem: 52} Under the assumption of Lemma \ref{lemGamma}
and assuming that $a, b, i, j$ are all different,
we have
\be\label{new628}  
|\E F'(x^R)R_{ij}\overline{R_{ja}R_{bi}}(y+i\eta)|\leq N^{-4/3+C\e}
\ee
for any $y$ with $|y-2|\le N^{-2/3+\e}$,
and the same estimate holds for the other three terms in the r.h.s of  \eqref{680}.  
\end{lemma}

If this lemma holds, then we have thus proved  in the case $a \not = b$ that
\be\label{FF1232}
\left|\E  [ F(x^S)- F(x^R) ] 
- B  \right|\leq N^{-13/6+C\e}
\ee
where $B$ is defined in \eqref{Adef}. With the definitions of $x$'s in \eqref{defxs}, 
this completes the proof of Lemma \ref{lemGamma} for the remaining $a\ne b$ case. \qed

\bigskip

{\it Proof of Lemma \ref{lem: 52}. } With the relation between $R$ and $S$ 
in \eqref{relRS} and \eqref{51},
one can see that \eqref{new628} is implied by 
\be\label{new629}
|\E F'(x^S)S_{ij}\overline{S_{ja}S_{bi}}|\leq N^{-4/3+C\e},\,\,\,\,\,\,\,\,\,\,\,\,\,
\ee 
under the assumption that $a,b,i,j$ are all different.
This replacement is only a technical convenience
when we apply the large deviation
estimate (Lemma \ref{generalHWT}) below. 
Lemma \ref{generalHWT} was formulated
with random variables of equal variance, while
the matrix elements of $Q$ cannot all be normalized to
have the same variance since two matrix elements are zero.
The contribution of these two elements is negligible anyway, 
but the presentation of the argument is simpler if we do not
have to carry them separately in the notation. 
Since $S$ is the Green function of a usual generalized Wigner matrix
with all variances being positive, it is easier to deal with
\eqref{new629} instead of \eqref{new628}.

From  the identity \eqref{GijGkij} applied to the Green function $S$,
 we have 
for any different $i$, $j$ and $a$ 
\be\label{GiiGjiif}
 |S_{ij}-S^{(a)}_{ij}|=\Big | 
S_{ia}S_{aj}(S_{aa})^{-1}\Big | \le C (N \eta)^{-2} \le CN^{-2/3+C\e}
 \qquad \mbox{in $\Omega^c$}.
\ee
{F}rom \eqref{basic} we have  
\be\label{sSijl}
|S_{ij}|\leq N^{-1/3+C\e}, \qquad i\ne j,\quad \mbox{in $\Omega^c$}.
\ee 
Combining \eqref{GiiGjiif} and \eqref{sSijl}, we have 
\be\label{c5}
|x^{S} - \wt x^S| \le N^{-1/3 + C \e}, \qquad 
\ee 
where $\wt x^S$ is defined using the resolvent of the matrix $H_{\gamma-1}^{(a)}$
exactly as $x^S$ was defined using the resolvent
$S$ of matrix $H_{\gamma - 1}$. As usual, $H_{\gamma-1}^{(a)}$
denotes the matrix $H_{\gamma-1}$ 
with $a$-th  row and column  removed.
Similarly, 
we have 
\be\label{Sbi}
\left|S_{ij}\overline{S_{ja}S_{bi}}-S_{ij}^{(a)}\overline{S_{ja}S_{bi}^{(a)}}\right|
\leq N^{-4/3+C\e}, \,\,\,\,\,{\rm in }\,\,\,\Omega^c.
\ee 
Hence by these inequalities  and the bounds on the derivatives of $F$, we have 
\be\label{632}
|\E F'(x^S)S_{ij}\overline{S_{ja}S_{bi}}|\leq \left|\E [F'(\wt x^S)] 
S_{ij}^{(a)}\overline{S_{ja}S_{bi}^{(a)}}\right|+O\left(N^{-4/3+C\e}\right).
\ee
Applying the identity \eqref{GijHij} to $S_{ja}^{}$, we have 
\be\label{shosja}
S_{ja}=S_{jj}S_{aa}^{(j)}Z_{ja}^{(S)}, \quad \mbox{with}\qquad
Z_{ja}^{(S)}:=\sum_{s\,t\notin\{a,j\}}h_{js}S^{(ja)}_{st}h_{ta}-h_{ja},
\ee
where $h_{\al\beta}=\left(H_{\gamma-1}\right)_{\al\beta}$. 
With the bound on
 the matrix elements of $S$ in \eqref{basic} and the identity \eqref{GijGkij}, 
in the set $\Omega^c$ we have
\be\label{SSS}
S_{jj}=m_{sc}+O(N^{-1/3+C\e}),
\qquad S_{aa}^{(j)}=m_{sc}+O(N^{-1/3+C\e}),\qquad  S_{ss}^{(ja)}=m_{sc}+O(N^{-1/3+C\e}).
\ee
Setting
$$
   \Omega_Z:= \Big\{ |Z_{ja}^{(S)}| \ge N^{-1/3+C\e}\Big\}
$$
with a sufficiently large constant $C$, 
Lemma \ref{generalHWT} implies that
$$
   \P (\Omega\cup\Omega_Z)\le  C\exp{\big[-c(\log N)^\xi\big]},
$$
for any fixed $\xi>0$
since on the set $\Omega^c$ we have
$$
   \sum_{s,t\not\in \{ a, j\} } \sigma_{js}^2\sigma^2_{ta}|S^{(ja)}_{st}|^2
   \le \frac{C_0^2}{N^2\eta}\sum_{s \neq a, j} \im S^{(ja)}_{ss} \le N^{-2/3+C\e} 
$$
using the last formula in \eqref{SSS}.
Therefore, with   \eqref{shosja}, in $\Omega^c\cap \Omega_Z^c$ we have 
 \be\label{SJA}
S_{ja}=m_{sc}^2 Z_{ja}^{(S)}
+O(N^{-2/3+C\e}).
\ee

Combining \eqref{SJA} with \eqref{632}, we see that 
\begin{align}\label{250}
|\E F'(x^S)S_{ij}\overline{S_{ja}S_{bi}}|
\leq &
|m_{sc}^2|\left|\E [F'( \wt x^S)] S_{ij}^{(a)}S_{bi}^{(a)}
\left(\sum_{s\,t\notin\{a,j\}}h_{js}S^{(ja)}_{st}h_{ta}-h_{ja}\right)\right|
+O\left(N^{-4/3+C\e}\right).
\end{align}
Since $\wt x^S$, $S_{ij}^{(a)}S_{bi}^{(a)}$, $h_{js}$ and $S^{(ja)}_{st}$
are all
independent of the $a-$th row and column of $H_{\gamma-1}$, 
 and the expectations
 of $h_{ta}$ and $h_{ja}$ are zero, the first term in r.h.s. of \eqref{250} equals 
to zero. This implies  \eqref{new629} and completes the proof of  \eqref{new628}. 
The other terms in \eqref{680} can be bounded similarly. 
This completes the proof of Lemma \ref{lem: 52}.
\qed

\section{Proof of Lemma \ref{motN}}\label{sec:Z}

\subsection{Setup and notations}

The $p$-th moment  of $\sum_{i=1}^NZ_i$ is given by  
 \be\label{21p}
 \frac{1}{N^p} \E  {\bf 1} (\Gamma^c) \left|\sum_{q=1}^N Z_q \right|^p
= \frac{1}{N^p} \E\sum_\# \sum_{q_1=1}^N \ldots\sum_{q_\al=1}^N \ldots\sum_{q_p = 1}^N
  {\bf 1} (\Gamma^c)
Z_{q_1}^\#  \ldots Z_{q_p}^\# ,
 \ee 
 where the various $\#$'s can be either $0$ or the complex conjugate.
The precise choice of $\#$ will be irrelevant for our argument
and the summation over them yields an irrelevant overall factor $2^p$.

We write up the definition of $Z_{q_\al}$  from \eqref{defZi}
as follows:
\be\label{34}
    Z_{q_\al} = \sum_{q_\al^2, q_\al^3=1}^N   G^{(q_\alpha)}_{q_\alpha^2 \,q_\alpha^3}
     \big[ h_{q_\al, q^2_\al}
  h_{q_\alpha^3, q_\alpha} -\delta_{q^3_\alpha,  q^2_\alpha}
  \sigma^2_{q^2_\alpha, q_\alpha}\big],
\ee
where the summation is over all $q_\al^2\ne q_\al$ and $q_\al^3\ne q_\al$.
To bookkeep the indices in a uniform way, we denote $q_\al$ by $q_\al^1$
and we organize the three indices $(q^1_\al, q^2_\al, q^3_\al)$
into a vector  $\fq_\alpha$ for each $\al =1,2,\ldots, p$.

 Furthermore, we organize these
$p$ vectors into a $3\times p$ matrix 
$\bq = (q_\alpha^j)$, for $\alpha =1, \ldots, p$ 
 and $j = 1, 2, 3$, with
entries taking values in $ \N_N:=\{ 1,2,\ldots, N\}$.
 The {\it slots}
 of the matrix $\bq$, parametrized by $(j,\al)$, $\al=1,2, \ldots ,p$,
 $j=1,2,3$,
are called {\it vertices}, since we will build a graph upon them. 
The element $q_\al^j$ will be called the {\it index} assigned to
the vertex $(j, \al)$. 
The first entry $q_\al^1$ in $\fq_\al$  will play a special role, it will
be called {\it location index}, the other two indices, $q_\al^2$, $q_\al^3$
will be called {\it nonlocation indices}. Similarly,
$(1, \al)$ will be called {\it location vertex} and 
$(2, \al)$, $(3,\al)$ will be called {\it nonlocation vertices}.
A pair of indices is called {\it label}.
We also define the set of labels in $\fq_\al$ that contain $q^1_\al$:
$$
Q_\al : = \{ (q^1_\al, q^2_\al), (q^1_\al, q^3_\al),
 (q^2_\al, q^1_\al) ,(q^3_\al, q^1_\al) \}, \qquad \al=1,2 \ldots p,
$$
and sometimes we will use a single letter $\nu$ or $\mu$ for labels, i.e.
for elements of $\bigcup_{\al=1}^p Q_\al$. Note that $ Q_\al$ contains
any label $\nu$ together with its {\it transpose} $\nu^t$,
where $\nu^t:= (p,q)$ if $\nu =(q,p)$. Carrying $\nu$ together with 
its transpose is necessary since $h_\nu = \bar h_{\nu^t}$, i.e.
matrix elements with labels $\nu$ and $\nu^t$ are not independent.

With these notations, we have 
 \be\label{212}
 \frac{1}{N^p} \E  {\bf 1} (\Gamma^c) \left|\sum_{i=1}^N Z_i \right|^p
= \frac{1}{N^p}  \sum_{\bq}   \Phi_{\bq},
 \ee 
where we defined  
\be\label{PHI}
 \Phi_{\bq} :=   \E   {\bf 1} (\Gamma^c)
      \prod_{\alpha=1}^{p} \Big[ G^{(q_\alpha^1)}_{q_\alpha^2 \,q_\alpha^3}
 \xi (\fq_\alpha)\Big]^\#,   
      \quad \xi (\fq_\al):= h_{q^1_\al, q^2_\al}
  h_{q_\alpha^3, q_\alpha^1} -\delta_{q^3_\alpha,  q^2_\alpha}
  \sigma^2_{q^2_\alpha, q^1_\alpha}.
\ee
The summation in \eqref{212} runs over all $3\times p$ matrices ${\bf q}$ 
with elements from $\N_N$ and with the restriction that 
\be
q_\al^2\ne q_\al^1, \quad \mbox{and} \quad
q_\al^3\ne q_\al^1.
\label{rest}
\ee
Let 
\be
Q_\bq=Q:= \bigcup_{\al=1}^p Q_\al 
\label{def:A}
\ee 
denote the set of all possible labels of $h$-variables appearing
in the $\xi(\fq_\al)$ factors and notice that its cardinality
is bounded by $|Q|\le 4p$. 

 We would like to compute the expectation in \eqref{PHI} by
first taking the expectation with respect to the $h_\nu$-variables explicitly
appearing in the $\xi$'s. 
Recall $G^{(q)}=(H^{(q)}-z)^{-1}$ is the Green function of $H^{(q)}$ 
which is an $(N-1) \times (N-1)$ matrix after removing the
 $q$-th row and column from $H$. Thus  $G^{(q_\alpha^1)}$
is independent of the random variables  $h_\nu$, $\nu\in Q_\al$,
i.e. those $h$-variables that  explicitly appear in $\xi(\fq_\al)$.
 There are, however, three complications. 
First, while each Green function $G^{(q_\al^1)}$, $\al =1,2,\ldots , p$,
is independent of  $h_\nu$, $\nu\in Q_\al$, 
by definition,  it still  depends on the
other $h_\mu$-variables, $\mu\in Q_\beta$, $\beta\ne \al$.
Second, we have
to deal with coincidences; the same $h$-variable may appear in
$\xi(\fq_\al)$ and $\xi(\fq_\beta)$ with $\al\ne \beta$; in fact
these terms give the non-zero contributions. We will develop a 
graphical scheme to bookkeep the structure of coincidences
and estimate the number of off-diagonal resolvent elements.
Finally, there is a small technical problem related to
the factor ${\bf 1}(\Gamma^c)$ that depends on all 
$h$-variables, but this factor equals one with a very high probability
so a fairly easy argument can remove it.

To resolve the first problem, we  use the resolvent expansion to express explicitly the
dependence of $G^{(q_\alpha^1)}$ on the random variables 
$h_\nu$ with label $\nu\in Q_\beta$, $\beta\ne\al$.
For $\bq $ fixed, let $U^{\langle \alpha \rangle }= 
U^{\langle \alpha \rangle }_{\bq}$ be the matrix
\be
(U^{\langle \alpha \rangle })_{ i,k}: = (H^{(q^1_\alpha)})_{ i, k }, 
\quad
\mbox{for}\;\;  (i, k) \in   Q^{(\alpha)}_\bq : = Q^{(\al)}=
 \bigcup_{\beta \in \{1, \ldots, p\}, \beta   \not = \alpha}  Q_\beta, 
\ee
and $(U^{\langle \alpha \rangle })_{ i,k}:=0$ otherwise. 
Note that the number of nonzero matrix elements of $U^{\langle \alpha \rangle }$ 
is bounded by $|Q|\le 4p$.
Define 
\[
H^\sa=H^\sa_\bq:= H^{(q_\alpha^1)} - U^{\langle \alpha \rangle }, \qquad
G^{\sa}_\bq=G^{\sa }:= (H^{(q_\alpha^1)} - U^{\langle \alpha \rangle } - z)^{-1} .
\]
Notice that $G^{\sa}_\bq$ is independent of all the $h$-factors that explicitly
appear in $\prod_\al \xi(\fq_\al)$.
{F}rom the resolvent expansion, we have 
\be
G^{(q_\alpha^1)} = \sum_{n_\alpha=0}^\infty
  (-G^{\sa } U^{\langle \alpha \rangle } )^{n_\alpha}G^{\sa }. 
\label{exp}
\ee

To estimate the size of these Green functions, we first note
that there is a positive universal constant $c$ such that
  on the set $\Gamma^c$ we have
\be
  \max_{i\ne j} | G_{ij}| = \Lambda_o\le \frac{1}{(\log N)^2}, \qquad 
  c\le  |G_{ii}| \le 1 + \frac{1}{(\log N)^2}.
\label{Gb}
\ee
This follows from the fact that  $c'\le |m_{sc}(z)|\le 1$ 
with some positive universal $c'>0$ and
 for any $z\in \bS_\ell$, see \eqref{smallz}.
By the perturbation formulas \eqref{GiiGjii} and \eqref{GijGkij} 
we have $G_{ij}^{(k)} = G_{ij}  - G_{ik}G_{kj}/ G_{kk}$ for $i,j \ne k$, thus we
also have 
\be
   \max_{i\ne j}| G_{ij}^{(k)}|\le 2\Lambda_o\le \frac{C}{(\log N)^2}, \qquad 
  c' \le  |G_{ii}^{(k)}| \le 1 + \frac{C}{(\log N)^2},
\label{Gkb}
\ee
where $i,j \ne k$. 
In the good set $\Gamma^c$,  the matrix 
elements of  $ U^{\langle \alpha \rangle }$  satisfy
\be
    | U^{\langle \alpha \rangle }_{ij}|\le \frac{(\log N)^{L/10}}{\sqrt{N}} \le N^{-1/4}
\label{Ub}
\ee
(here we used that $L\le  \log N/\log\log N$),
and $G^\sa$ is bounded as
\be
    \max_{i\ne j} |G^\sa_{ij}|\le 2\Lambda_o\le \frac{C}{(\log N)^2},
 \quad   |G_{ii}^\sa| \le 1 + \frac{C}{(\log N)^2} \qquad \mbox{in $\Gamma^c$}.
\label{gsa}
\ee
   To see \eqref{gsa}, we expand
$$
   G^\sa = \sum_{m=0}^\infty (G^{(q_\al^1)}  U^{\langle \alpha \rangle } )^{m}G^{(q_\al^1) }
  \qquad \mbox{in $\Gamma^c$},
$$
and use \eqref{Ub} and  the bounds \eqref{Gkb} on
 the matrix elements of $G^{(q)}$, $q\in \N_N$.

Using \eqref{Ub} and \eqref{gsa} and recalling that only 
finitely many matrix elements of $U$ are non-zero, we easily see that
the expansion \eqref{exp} is convergent and it can be truncated
 at finite $n_\alpha$ so
 that the error term can be estimated. Thus  there will be no convergence problem
and  we will focus on getting  estimates.

We set 
$$
  {\bf n} := (n_1, n_2, \ldots , n_p), \qquad |{\bf n} | = \sum_{\alpha=1}^p  n_\alpha.
$$
With this expansion, we can write \eqref{PHI} as
\begin{align}\label{32}
 \Phi_{\bq } & = \sum_{n=0}^\infty \sum_{ |{\bf n} | =n}  \Phi_{\bq}^{\bf n} \quad
\non\\
 \Phi_{\bq}^{\bf n} & : =  \E {\bf 1} (\Gamma^c)   \prod_{\alpha=1}^{p} 
 \Big[ \cM^{(n_\alpha)}  \xi (\fq_\al)\Big]^\#
 ,  \\
  \cM^{(n_\alpha)}= \cM_{\bq }^{(n_\alpha)} & :=   
 \big [  (-G^{\sa } U^{\langle \alpha \rangle } )^{n_\alpha}G^{\sa }  
  \big ]_{q^2_\alpha, q^3_\alpha}   \\
  & =  \sum_{ \nu^\alpha_1, \nu^\al_2, \ldots , \nu^\al_{n_\al} \in  Q^{(\alpha)}}
  V_{\bq}( \underline{\mu}^\alpha, \underline{\nu}^\alpha, n^\al )\label{summ}
\end{align}
with $\unu^\al: = ( \nu^\alpha_1, \nu^\al_2, \ldots , \nu^\al_{n_\al})$ and 
we have expanded  $U^{\langle \alpha \rangle }$ appearing in   
$\big [   (-G^{\sa } U^{\langle \alpha \rangle } )^{n_\alpha}G^{\sa } \big ]_{q^2_\alpha, q^3_\alpha}$
 and used the notation 
\be\label{def:V}
  V_{\bq}( \underline{\mu}^\alpha, \underline{\nu}^\alpha, n^\al): =  (-1)^{n_\al}
 G^{\sa }_{\mu^\alpha_1} h_{\nu^\alpha_1}
 G^{\sa }_{\mu^\alpha_2}h_{\nu^\alpha_2} \ldots 
 h_{\nu^\alpha_{n_\al}} G^{\sa }_{\mu^\alpha_{n_\alpha+1}}. 
\ee
The summation in \eqref{summ} is over all possible $\nu$-labels 
of the $h$ factors 
in \eqref{def:V}. The appearance of the $\mu^\al$-labels in
 \eqref{def:V} is just
 notational simplification,
they are explicit functions of $\underline{\nu}^\al$ and $\fq_\al$ as follows:
\be\label{121}
    \mu^\alpha_1 = (q^2_\alpha,  [\nu^\al_1]_1 ), \;\; 
  \mu^\al_2 = ( [\nu^\al_1]_2 , [\nu^\al_2]_1), \;\;
 \mu^\al_3 = ( [\nu^\al_2]_2 , [\nu^\al_3]_1), \;\; \ldots
  \mu^\al_{n_\al+1} = ( [\nu^\al_{n_\al}]_2 , q^3_\alpha ),
\ee
where $[\nu^\al_j]_1$ and  $[\nu^\al_j]_2$ denotes the first and second
element of the label $\nu^\al_j$.
Notice that  $G^{\sa } $ is independent of all matrix 
elements $h_{jk}$ explicitly appearing in the $\xi$-factors in \eqref{32}.

Hence
\begin{align}\label{213}
 \Phi_{\bq}^{\bf n}  :=    \E {\bf 1} (\Gamma^c)   \sum_{ \bnu}  
  \prod_{\alpha=1}^{p} \big[  V_\bq( \umu^\alpha, \unu^\alpha, n^\al)   \xi (\fq_\al)\big]^\#,
\end{align}
where the summation is over all $p$-tuple of label sequences 
 $\bnu= (\underline{\nu}^1, \underline{\nu}^2, \ldots , \underline{\nu}^p)
  \in A(\bq, \bn):=\prod_{\al=1}^p \big[Q^{(\al)}\big]^{n_\al}$.
The number of different $\bnu$'s is bounded by $|A(\bq,\bn)|\le(4p)^n$.

\subsection{Strategy of the proof presented in the simplest example}

In order to motivate the reader
before we start the detailed estimates, 
we show our strategy
via the simplest case $p=2$,
$$
    \frac{1}{N^2} \E {\bf 1}(\Gamma^c) \Big|\sum_{i=1}^N Z_i\Big|^2
   =   \frac{1}{N^2} \E {\bf 1}(\Gamma^c)\sum_{i, j=1}^N 
  Z_i \overline{Z_j}.
$$
We write out
$$
  Z_i = \sum_{k,l\ne i} G_{kl}^{(i)} \big[ h_{ik}h_{li} - \delta_{kl}\sigma_{ik}^2\big],
  \qquad 
 Z_j = \sum_{m,n\ne j} G_{mn}^{(j)} \big[ h_{jm}h_{nj} - \delta_{mn}\sigma_{jm}^2\big].
$$
thus we have
\be\label{summand}
 \frac{1}{N^2} \E {\bf 1}(\Gamma^c) \Big|\sum_{i=1}^N Z_i\Big|^2
 = \frac{1}{N^2} \E {\bf 1}(\Gamma^c)
\sum_{ijklmn}G_{kl}^{(i)} \big[ h_{ik}h_{li} - \delta_{kl}\sigma_{ik}^2\big]
  \ov{G_{mn}^{(j)} \big[ h_{jm}h_{nj} - \delta_{mn}\sigma_{jm}^2\big]}.
\ee
With the general notation $\alpha=1,2$ and
the six indices in the summation are organized into a
 3x2 matrix with columns $\fq_1 = (q_1^1, q_1^2, q_1^3)$
and $\fq_2= (q_2^1, q_2^2, q_2^3)$, i.e.
$$
   \bq =  \left(\begin{array}{cc} q_1^1 & q_2^1 \cr
q_1^2 & q_2^2 \cr
q_1^3 & q_2^3
\end{array}\right)=
\left(\begin{array}{cc} i & j \cr
k & m \cr
l & n
\end{array}\right).
$$
The only restriction for these indices is that the top
element of each column is distinct from the other two below.
The sets $Q_1 =\{ (i,k), (k,i), (i,l), (l,i)\} $ and $Q_2
= \{ (j,m), (m,j), (j,n), (n,j)\}$ contain the labels
of the $h$ factors that explicitly appear in $Z_i$ and $Z_j$, respectively.

Now we expand $G^{(i)}= G^{(q_1^1)}$ in the variables $h_\nu$ labelled by $\nu \in Q_2$.
We thus decompose the minor $H^{(q_1^1)} = H^{[1]}+ U^{\langle 1\rangle}$, where
 the  matrix $U^{\langle 1\rangle}$ contains
 only four non-zero entries  $h_{jm}$, $h_{mj}$, $h_{jn}$ and $h_{nj}$
with labels from $Q_2$,
and $H^{[1]}$ contains all other entries of $H^{(q_1^1)}$. 
The resolvent $G^{[1]} = (H^{[1]}-z)^{-1}$ is now independent of
all expansion variables $h_\nu$ with $\nu\in Q=Q_1\cup Q_2$.
Note, however, that this decomposition depends on $\bq$, i.e.
it will be different for each summand in \eqref{summand}.
Since $U^{\langle 1\rangle}$ is small, we can expand
$$
   G^{(i)}= G^{(q_1^1)} = G^{[1]}  - G^{[1]} U^{\langle 1\rangle} G^{[1]} + 
 G^{[1]} U^{\langle 1\rangle} G^{[1]} U^{\langle 1\rangle} G^{[1]}
\ldots,
$$
and a similar expansion holds for $G^{(j)}= G^{(q_2^1)}$.

We insert these expansions into \eqref{summand}
and organize the terms according to their number of the
explicit $h$ factors. Effectively, each $h$ factor has
a size $N^{-1/2}$ (neglecting logarithmic corrections). The 
centered random variable $\xi(\fq)= h_{q^1, q^2}h_{q^3, q^1}-\delta_{q^2,q^3}\sigma^2_{q^1,q^2}$ has size $N^{-1}$
and the subtracted expectation $\delta_{q^2,q^3}\sigma^2_{q^1,q^2}$ 
is treated on the same footing as $hh$ for the purpose of power counting.

Typically we need to show that terms with less than
 eight $h$ factors have zero expectation
to compensate for the sixfold summation of order $N^6$
with the prefactor $N^{-2}$  in \eqref{summand}.
Depending on certain coincidences among the 
summation indices, sometimes terms with less than eight $h$ factors
already give non-zero contribution, but then the combinatorial
factor from the summation is smaller.
Furthermore, we want to bookkeep the number of off-diagonal
matrix elements since the final estimate is in
terms of a power of $\Lambda_o$.

The leading term in \eqref{summand},
\be\label{leading}
   G^{[1]}_{kl} \big[ h_{ik}h_{li} - \delta_{kl}\sigma_{ik}^2\big]
  \ov{G_{mn}^{[2]} \big[ h_{jm}h_{nj} - \delta_{mn}\sigma_{jm}^2\big]},
\ee
has four $h$ factors  but its  expectation
vanishes unless at least two summation indices in \eqref{summand}
coincide, so the sixfold summation
is effectively only fourfold. Here the key observation is that if at least one
$h$ factor  appears linearly in the expansion, then 
the expectation is zero.
However, since 
the quadratic factor $\xi(\fq_1)=\big[ h_{ik}h_{li} - \delta_{kl}\sigma_{ik}^2\big]$
has zero expectation, it is not sufficient to
set $k=l$ and $m=n$ to get a non-zero contribution;
there must be coincidences between the $h$ factors
in $\big[ h_{ik}h_{li} - \delta_{kl}\sigma_{ik}^2\big]$
and in $\big[ h_{jm}h_{nj} - \delta_{mn}\sigma_{jm}^2\big]$.
For example the case $i=j$, $k=m$, $l=n$ yields a nonzero
contribution, i.e. the summation is only threefold.
Moreover, if both resolvent elements in \eqref{leading}
are off-diagonal, then
we get an estimate of order $N^{-2}N^3 (N^{-1/2})^4\Lambda_o^2
=\Lambda_o^2N^{-1} $.
If one of the resolvent elements is diagonal,
say $k=l$, then the other one has to be diagonal
as well, $m=n$, otherwise the expectation is zero. 
This forces one more coincidence, i.e.
either $i=j$ and $k=l=m=n$ or $i=m=n$, $j=k=l$.
In both cases the summation in \eqref{summand} gives only 
$N^2$ and the total estimate is of order $N^{-2}$.

The next order terms in the expansion are of the form
 $$
 \Big( G^{[1]} U^{\langle 1\rangle}G^{[1]}\Big)_{kl} \big[ h_{ik}h_{li} - \delta_{kl}\sigma_{ik}^2\big]
 \ov{ G_{mn}^{[2]} \big[ h_{jm}h_{nj} - \delta_{mn}\sigma_{jm}^2\big]}
$$
$$
\qquad\qquad = \sum_{a,b\in Q_2} G^{[1]}_{ka} U^{\langle 1\rangle}_{ab}G^{[1]}_{bl} 
\big[ h_{ik}h_{li} - \delta_{kl}\sigma_{ik}^2\big]
\ov{  G_{mn}^{[2]} \big[ h_{jm}h_{nj} - \delta_{mn}\sigma_{jm}^2\big]}
$$
with five $h$ factors. Notice that two new summation indices, $a, b$,
have appeared, but their combinatorics is of order one and not of order $N^2$.
In fact, $ U^{\langle 1\rangle}_{ab}$ is just one of $h_{jm}$, $h_{nj}$ 
or their transposes.
Again, there should be
at least three coincidences  among the indices $i,j,k,l,m,n$
 to avoid that at least one $h$ variable appears linearly
or that at least one of the quadratic factors $\xi(\fq_1)$,
$\xi(\fq_2)$ remains isolated
leading to zero expectation.
It is again easy to see that we collect at least $\Lambda_o^2$
(in fact, typically $\Lambda_o^3$) unless at least one additional index coincides.

The terms with six $h$ factors are either of the form
 $$
 \Big( G^{[1]} U^{\langle 1\rangle}G^{[1]}\Big)_{kl} \big[ h_{ik}h_{li} - \delta_{kl}\sigma_{ik}^2\big]
 \ov{\Big(
 G^{[2]}U^{\langle 2\rangle}G^{[2]}\Big)_{mn}  \big[ h_{jm}h_{nj} - \delta_{mn}\sigma_{jm}^2\big]}
$$
or of the form
$$
 \Big( G^{[1]} U^{\langle 1\rangle}G^{[1]} U^{\langle 1\rangle}G^{[1]} \Big)_{kl} 
\big[ h_{ik}h_{li} - \delta_{kl}\sigma_{ik}^2\big]
 \ov{ G_{mn}^{[2]} \big[ h_{jm}h_{nj} - \delta_{mn}\sigma_{jm}^2\big]}
$$
In both cases at least two $h$ factor appears linearly, yielding
zero expectation, unless there are two coincidences
 among $i,j,k,l,m,n$. Thus the summation in \eqref{summand}
is effectively reduced from  $N^6$ to $N^4$.
 Since $h^6\sim (N^{-1/2})^6=N^{-3}$, we obtain that \eqref{summand} is
of order $N^{-1}$.
Moreover, in all cases there are at least two offdiagonal resolvent elements,
unless an additional coincidence occurs. Thus the  estimate is $N^{-1}(\Lambda^2_o+N^{-1})$.
The seventh order terms can be dealt with similarly.

 The
lowest order non-zero terms with distinct $i,j,k,l,m,n$ 
indices have eight $h$ factors and they are of the form 
$$
 \Big( G^{[1]} U^{\langle 1\rangle}G^{[1]} U^{\langle 1\rangle}G^{[1]} \Big)_{kl} 
\big[ h_{ik}h_{li} - \delta_{kl}\sigma_{ik}^2\big]
 \ov{ \Big( G^{[2]} U^{\langle 2\rangle}G^{[2]}U^{\langle 2\rangle}G^{[2]}
\Big)_{mn} \big[ h_{jm}h_{nj} - \delta_{mn}\sigma_{jm}^2\big]}.
$$
We now have four $U$-factors, so they can 
ensure that all variables  $h_{ik}$, $h_{li}$, $h_{jm}$, $h_{nj}$
appear quadratically
to prevent zero expectation. For example, the term
$$
 G^{[1]}_{km}h_{mj}G^{[1]}_{jj} h_{jn}G^{[1]}_{nl}
\big[ h_{ik}h_{li} - \delta_{kl}\sigma_{ik}^2\big]
 \;\ov{  G^{[2]}_{mk}h_{ki} G^{[2]}_{ii} h_{il}G^{[2]}_{ln}
 \big[ h_{jm}h_{nj} - \delta_{mn}\sigma_{jm}^2\big]}.
$$
has non-zero expectation. Moreover, there are  four resolvents
in offdiagonal form, unless
there is an index coincidence, so the size of this term
is  $N^{-2} N^{6} (N^{-1/2})^8 \Lambda_o^4=\Lambda_o^4$.

\bigskip

The mechanism to estimate the term \eqref{213} for  general $p$
 is the same, but the
bookkeeping is more tedious. We will have to estimate
the size of each non-vanishing term  as powers of $N$ and $\Lambda_o^2$.

The power counting in $N$ is relatively straightforward. It is easy to see 
that if all indices in the matrix $\bq$ are distinct, then
at least $2p$ new $h$ factors must come from the $V_\bq$ factors
to ensure that none of the $h$ factors in $\prod_\al \xi(\fq_\al)$
appears linearly (otherwise the expectation would be zero).
 Thus the total number of $h$ factors is at least $4p$
and their size is estimated by $(N^{-1/2})^{4p} = N^{-2p}$. 
Together with the $N^{-p}$ prefactor in \eqref{212}, this will
compensate for the $N^{3p}$ combinatorial factor coming
from the summation over all $3\times q$ matrices. 
If  some indices in $\bq$ coincided, then the corresponding
$h$ factors could appear with a higher multiplicity 
in $\prod_\al \xi(\fq_\al)$, so their expectation would not necessarily
vanish even without an additional $h$ factor from $V_\bq$.
Each coincidence in $\bq$  reduces the number of necessary $h$ factors from $V_\bq$ 
at most by two, hence keeping the overall balance of $N$-powers.

The power counting in $\Lambda_o$ is more complicated
and it is related to the fact that the expectation
of each $\xi$ is zero. This means that an index coincidence of the form
 $q_\al^2=q_\al^3$ does not imply non-vanishing expectation yet.
The requirement of nonzero expection either forces coincidences
of indices among $h$ factors in {\it different}
$\xi$ terms, but then typically {\it two} indices have to
match, so we gain an additional $N^{-1}$;
 or it forces matching $h$ factors in the $\xi$-terms with
$U$-factors in the expansion \eqref{exp}. The latter implies, however,
that instead of a single resolvent $G^{[\al]}$ we consider
a longer expansion of the form $G^{[\al]} U^{\langle\al\rangle}
 G^{[\al]} \ldots $ which typically has at least {\it two}
off-diagonal resolvents instead of only one. These two scenarios yield an additional factor
$(\Lambda_o^2 +N^{-1})$ for each $\xi$-factor.
This gives $(\Lambda_o^2 +N^{-1})^p$ as a final estimate.

In the next section we give the precise details of this strategy.

\subsection{Detailed proof of Lemma~\ref{motN}.}

The proof will be divided into three parts.
The first part is  a technical preparation to deal with the 
very small probability event represented by the
set $\Gamma$, where either $h$ or a resolvent
is too large. It can be skipped at the first reading.
In the second part we organize the expansion
by encoding the coincidence structure of various
terms by a graph. Finally, in the third part
we estimate the size of each term
with the help of the graphical representation.

\subsubsection{Cutoff of small probability events}

 Since 
$|h_{ij}|\le (\log N)^{L/10} N^{-1/2}$ in the set $\Gamma^c$
and  \eqref{gsa} also holds in $\Gamma^c$,
we clearly have 
\be\label{omxi}
\big|{\bf 1} (\Gamma^c)    \xi (\fq_\al)  
 V_\bq( \umu^\alpha, \unu^\alpha, n^\al) \big|\le 
C \bigg [ \frac {C (\log N)^{L/10}} { \sqrt N } \bigg]^{n_\alpha} \frac{(\log N)^{L/10} }{ N}. 
\ee
Hence we have 
\begin{align}
 |\Phi_{\bq}^{\bf n}| \le    
 (Cp)^n  \bigg [ \frac {(\log N)^{L/10} } { \sqrt N } \bigg]^{n}
\Big( \frac{(\log N)^{L/10} }{ N}\Big)^p,
\end{align}
where $(Cp)^n$ is the combinatorics of the summation 
over $\bnu$ in \eqref{213}. Thus we have 
\be\label{omm}
\frac{1}{N^p}  \sum_{\bq}   \Phi_{\bq}
= \frac{1}{N^p}  \sum_{\bq}    \sum_{n=0}^\infty \sum_{ |{\bf n} | =n}  \Phi_{\bq}^{\bf n} 
\le (\log N)^{pL/10} N^{p} \sum_{n=0}^\infty (Cp)^n\sum_{ |{\bf n} | =n}  
 \bigg [ \frac { (\log N)^{L/10} } { \sqrt N } \bigg]^{n},
\ee
where we used that the summation over all $\bq$ yields a factor $N^{3p}$.
Since the number of $\bn=(n_1, n_2, \ldots, n_p)$
 with $|\bn| = n$ is bounded by $2^{n+p}$,
  the last term is bounded by 
\be
(C N (\log N)^{L/10})^{p} \sum_{n=0}^\infty   
 \bigg [ \frac { Cp  (\log N)^{L/10} } { \sqrt N } \bigg]^{n}.
\ee
Since $p\le (\log N)^{L/10}$ and $L\le \log N/\log \log N$, 
the sum of the tail terms with $n \ge  6 p$
 is bounded by $CN^{-5p/2}$, for sufficiently large $N$,
hence for
the bound \eqref{52}
 we only have to estimate terms with $n \le 6 p$.

We denote all independent  random variables by   $\bh = (h_\nu) $
and split them according to the set $Q$ (see \eqref{def:A}), i.e., 
we will write  $\bh= (\bh_1, \bh_2 )$
 with  $\bh_2 = (h_\nu:  \nu \in Q)$ and  $\bh_1 = (h_\nu:  \nu 
\not \in Q)$. Denote the corresponding projection by  $\pi_j, j=1, 2$, i.e.
$\pi_j \bh = \bh_j$.
Define
\be
(\Gamma^c)_1:= \pi_1(\Gamma^c), \qquad
 Y^c:=  \prod_{\nu \in Q} \{ h_\nu: |h_\nu| \le (\log N)^{L/10}
 |\sigma_{\nu}| \} \subset \C^Q, \qquad Y:=  \C^Q\setminus Y^c.
\ee
By definition, $G^{\sa}$ depends only on variables  $ \bh_1$.
 Furthermore, for any $\bh_1 \in (\Gamma^c)_1 $
there exists $\bh_2$ such that $\bh= (\bh_1, \bh_2 ) \in \Gamma^c$,
in particular, the estimates \eqref{gsa} hold for any $\bh_1 \in (\Gamma^c)_1 $.
By definition of $\Gamma^c$, we have 
\be\label{Ycont}
\Gamma^c \subset (\Gamma^c)_1  \times Y^c
\ee
and \eqref{omxi} holds in the set $(\Gamma^c)_1  \times Y^c$.
{F}rom the resolvent expansion, we have for $i \not = j$, 
and for $\bh_1\in (\Gamma^c)_1$,
\begin{align}\label{221}
G^{\sa }_{ij}(\bh_1) & =G^{\sa}_{ij} (\bh)=  
(H^{(\alpha)} - U^{\langle \alpha \rangle } - z)^{-1}_{ij} 
 = \sum_{n_\alpha=0}^\infty  \bigg [ (G^{(\al) }
 U^{\langle \alpha \rangle } )^{n_\alpha}G^{(\alpha) } \bigg ]_{ij}\\
 & =G^{(\al) }_{ij}+  \sum_{n_\alpha=1}^\infty  
\bigg [ (G^{(\al) } U^{\langle \alpha \rangle } )^{n_\alpha}G^{(\alpha) } \bigg ]_{ij}.
\end{align}
Using 
\be\label{y71}
 {\bf 1} (\Gamma^c) =  {\bf 1} ((\Gamma^c)_1) -   
  {\bf 1} \Big( (\Gamma^c)_1 \times Y^c  \setminus \Gamma^c \Big)
  -  {\bf 1} ((\Gamma^c)_1){\bf 1}( Y),
\ee
 we can rewrite $ \Phi_{\bq}^{\bf n}$ as
 \begin{align}\label{2131}
 \Phi_{\bq}^{\bf n}  &:=   \wt  \Phi_{\bq}^{\bf n} +X^\bn_{\bq,1}+X^\bn_{\bq,2}  \\
   \wt \Phi_{\bq}^{\bf n} & :=  \sum_{\bnu\in A(\bq,\bn)} 
 \wt \Phi_{\bq,\bnu}^{\bf n} \nonumber \\
 \wt \Phi_{\bq,\bnu}^{\bf n} & := \E {\bf 1} ((\Gamma^c)_1)    
  \prod_{\alpha=1}^{p}   V_\bq( \umu^\alpha, \unu^\alpha, n^\al)   \xi (\fq_\al)
\label{2111}  \\
X^\bn_{\bq,1}  & :=  - \E {\bf 1}((\Gamma^c)_1)  {\bf 1}(Y) \sum_{ \bnu}  
  \prod_{\alpha=1}^{p}   V_\bq( \umu^\alpha, \unu^\alpha, n^\al)   \xi (\fq_\al)  \\
X^\bn_{\bq,2}& : =  
 -  \E {\bf 1} \Big( (\Gamma^c)_1 \times Y^c  \setminus \Gamma^c \Big)   \sum_{ \bnu}  
  \prod_{\alpha=1}^{p}   V_\bq( \umu^\alpha, \unu^\alpha, n^\al)   \xi (\fq_\al).
\end{align}
 Analogously to \eqref{omxi}--\eqref{omm},
we can bound $X^\bn_{\bq,2}$ as follows 
\be\label{X1}
\frac{1}{N^p}\sum_\bq \sum_{n=0}^{6p}\sum_{|\bn|=n}
|X^\bn_{\bq,2}| \le (C p )^{6p}(N (\log N)^{L/10} )^{p} \P (\Gamma)
 \le  C\exp{\big[-c(\log N)^{\phi L} \big]},
\ee
using  the fact that the estimate
 \eqref{omxi} holds even on $ (\Gamma^c)_1 \times Y^c$
since all $G^{[\al]}$ appearing in $V_\bq$  depend
 only on $\{ h_\nu\; : \; \nu\not\in Q\}$.
 In the last step we used \eqref{Gammaest}, $n\le 6p \le 6(\log N)^{\phi L-2}\le 6 (\log N)^{L/10}$.
For the other error term we have
\be\label{X2}
\frac{1}{N^p}\sum_\bq \sum_{n=0}^{6p}\sum_{|\bn|=n}
|X^\bn_{\bq,1}| \le (C p )^{6p}(N (\log N)^{L/10} )^{p}  \, 
\exp{\big[-c(\log N)^{\psi L } \big]} \le  C\exp{\big[-c(\log N)^{\psi L } \big]}.
\ee
Here we have used that for a sufficiently large $L$, the integration of  $\bh_2$ over the set $Y$,
i.e. an $O(p)$-moment of the random variables $\bh_\nu$, $\nu\in Q$, in the
regime where $|h_\nu|\ge (\log N)^{L/10}\sigma_\nu$, 
 is bounded by  $ C\exp{\big[-c(\log N)^{\psi L} \big]}$ with some 
positive $\psi$, depending on $\ttau$ due to the subexponential decay 
\eqref{subexp}
 and due to the fact that
 $p\le (\log N)^{\psi L-2}$.
 In the estimate \eqref{X2} we also used that \eqref{gsa} holds on $(\Gamma^c)_1$ 
to estimate the $G^\sa$ factors remaining from the $V_\bq$ terms after
integrating out the random variables $h_\nu$, $\nu\in Q$.

Collecting the estimates from  \eqref{2131}, \eqref{X1} and \eqref{X2}, we 
have
\be
\frac{1}{N^p}  \sum_{\bq}   \Phi_{\bq}
\le  \frac{ 1 }{N^p}  \sum_{\bq}    \sum_{n=0}^{6p}  \sum_{ |{\bf n} | =n} 
\big| \wt \Phi_{\bq}^{\bf n}\big| + C\exp{\big[-c(\log N)^{\psi L } \big]}.
\label{222}
\ee
 The last error term can be absorbed into the $N^{-p}$ term in \eqref{52}
using that $p\le (\log N)^{\psi L -2}$.  
Hence we only have to estimate the contribution of $ \wt \Phi_{\bq}^{\bf n}$.
The key observation is that
\be
   \E_{h_\nu} {\bf 1}( (\Gamma^c)_1) \xi(\fq_\al) =0
\label{obs}
\ee
for any $\al=1,2,\ldots, p$ and for any $\nu\in Q $.
Furthermore, any resolvent $G^{[\al]}$
appearing explicitly in
\be\label{def:prodV}
\prod_{\alpha=1}^{p}   V_\bq( \umu^\alpha, \unu^\alpha, n^\al)
=  \prod_{\alpha=1}^{p} (-1)^{n_\al}G^{\sa }_{\mu^\alpha_1} h_{\nu^\alpha_1}
 G^{\sa }_{\mu^\alpha_2}h_{\nu^\alpha_2} \ldots 
 h_{\nu^\alpha_{n_\al}} G^{\sa }_{\mu^\alpha_{n_\alpha+1}} 
\ee
is independent of
any $h_\nu$, $\nu\in Q$. Therefore the expectation 
in \eqref{2111} is nonzero only if for each $\nu\in Q$, 
either $h_\nu$ (or its transpose $h_{\nu^t}$) 
appears explicitly in \eqref{def:prodV} or 
$h_\nu$ (or its transpose $h_{\nu^t}$)
 appears in two different $\xi(\fq_\al)$ factors
in \eqref{2111}. 
 The first scenario imposes restrictions on 
the indices of the two resolvents $ G^{\sa }$ neighboring $h_\nu$
in \eqref{def:prodV} and we will infer that some of these
resolvents must be off-diagonal that can be estimated by
 $\Lambda_o$. The second scenario restricts the total
combinatorics of the summation over the $\bq$ indices
in \eqref{222}, which gain can also be expressed 
as a power of $N^{-1/2}$. In the next step we
set up a graphical representation to effectively bookkeep
all possible situations.

\medskip

\subsubsection{Combinatorics}

Recall that  $\bq$ is a $3 \times p $ matrix
with $3p$ slots.
The  estimate of  $ \wt \Phi_{\bq}^{\bf n}$ defined
in the previous section depends on the structure 
of the indices $\bq = (q_\al^j)$, more precisely, it depends on which of
the indices $q_\al^j$ coincide. The relevant structure of
these coincidences will be encoded by a graph, $\cG(\bq)$, to be 
defined below. 
Roughly speaking (with some modifications specified below),
 the vertex set of  $\cG(\bq)$ 
will be the set of possible slots of the matrix $\bq$;
 two vertices $(j,\al)$ and $(i,\beta)$
are connected by an edge if the
corresponding indices coincide, $q_\al^j =q_\beta^i$.
Then the summation over $\bq$ in the right side of
 \eqref{222} will be performed 
in two steps: first we sum over all possible graphs, then 
we sum over all possible $\bq$'s compatible with this graph, i.e. we
write
\be\label{fubini} 
  \sum_\bq = \sum_{G} \sum_{\bq\, : \, \cG(\bq) = G},
\ee
where the first summation is over all graphs with at most $3p$ vertices.
In fact, only certain special graphs $G$ will be compatible
with a choice of indices $\bq$ that occur 
in our expansion and their number will be bounded by
$p^{Cp}$.

The reason for this resummation is that the
size of $\wt\Phi_\bq^\bn$ is essentially given 
by the number of off-diagonal resolvents  in the expansion \eqref{2111},
but considering only those terms which are not zero due to
the expectation (see \eqref{214} below). This number
can be estimated via the coincidence graph.

We now define the graph $\cG(\bq)$, describing the relevant coincidence
structure of $\bq$, by  performing the following four-step procedure.
Strictly speaking, the graph is defined on a subset of the $3p$
vertices (or slots in the matrix) labelled  by  
coordinates $(j, \alpha)$ with $1\le j \le 3$ and $1 \le \alpha \le p$.
We will say that a vertex $(j, \alpha)$  has the {\it value} 
$r$ if $q^j_\alpha = r$, in other words, the index $q^j_\al$ assigned
to the vertex  $(j, \alpha)$ will be sometimes also referred to
as the value of that vertex.  
 If it does not
lead to confusion,   we will often simply refer to $q^j_\alpha$ 
instead of the vertex $(j,\al)$, e.g. we will
 say that two indices,  $q^j_\alpha$
and $q_\beta^i$ are connected by an edge, meaning that the vertices $(j,\al)$
and $(i,\beta)$ are connected.

Let   $\ell (\bq)$ denote the number of different location indices, i.e., 
\be\label{elldef}
\ell=\ell(\bq) : = \big| \,   \{ q^1_\alpha:  1 \le \alpha \le p\} \, \big|,
\ee  
where $| \cdot|$ denotes the cardinality of the set, disregarding
multiplicity.
We group together all columns  with the same location indices;
the union of these columns will be called {\it group}.
Let $m_1, m_2, \ldots m_\ell$ denote the multiplicity of the groups, 
i.e., the number of columns with the same location indices. We clearly have
\be
   \sum_{s=1}^\ell m_s = p.
\label{msum}
\ee
We start with the matrix $\bq$ and perform the following operations
to obtain $\cG(\bq)$. In Step 1 and 2 we specify the vertex-set
of $\cG(\bq)$ by removing some of 
the original $3p$ vertices.
Step 3 and 4 specify the edges of $\cG(\bq)$.  After each step
we give an intutive explanation.

\begin{description}
\item [Step 1.] If $q^2_\al= q^3_\al$, we replace  $q^3_\al$ 
by $\ast$ and the vertex $(3,\al)$ will not be part of the graph $\cG(\bq)$.
In the matrix, we put a $\ast$ in its location.
We now call $q^2_\al$ a {\it duplex} and put a subscript $d$ to indicate it.

{\it Explanation:}  If $q^2_\al= q^3_\al$ then the
 two $h$ factors in $\xi(\fq_\al)$ are the same. This coincidence has to be treated
separately, since it does not automatically lead
to non-zero expectation due to $\E \xi(\fq_\al)=0$. It will thus
be easier
 to merge the vertices $(2, \alpha)$ and  $(3, \alpha)$
into one vertex.

\item [Step 2.] For $\al\ne \beta$ and any $i, j \in \{ 2, 3\}$
 we call the vertices
 $(j,\al)$ and $(i,\beta)$ (and the corresponding indices 
 $q^j_\al$ and  $q^i_\beta$)  {\it twin}
if  $ q^j_\al  =  q^1_\beta$ and $ q^i_\beta  =  q^1_\al$.
We now replace $q^j_\al $ and $q^i_\beta$ by $t$ to indicate 
 a twin but we do not make any change on  location index.
Vertices with $t$  will  not be part of the graph $\cG(\bq)$.
 Notice that by the restriction \eqref{rest},
 $q^j_\al \not = q^1_\al$ and thus 
$q^1_\al \not = q^1_\beta$, i.e., 
twins can only be formed in different groups, i.e.
in columns with different location indices.

{\it Explanation:} This is the situation where there is
a coincidence among the $h$ factors in two different 
$$
 \xi(\fq_\al) = h_{q^1_\al, q^2_\al}
  h_{q_\alpha^3, q_\alpha^1} -\delta_{q^3_\alpha,  q^2_\alpha}
  \sigma^2_{q^2_\alpha, q^1_\alpha}
\quad \mbox{and} \quad \xi(\fq_\beta) 
= h_{q^1_\beta, q^2_\beta}
  h_{q_\beta^3, q_\beta^1} -\delta_{q^3_\beta,  q^2_\beta}
  \sigma^2_{q^2_\beta, q^1_\beta},\qquad \al\ne \beta,
$$
 e.g. $ q^2_\al  =  q^1_\beta$ and $ q^2_\beta  =  q^1_\al$.
 Such coincidence results in
nonzero expectation with respect to $ h_{q^1_\al, q^2_\al}$
 {\it without forcing}  
 $ h_{q^1_\al, q^2_\al}$ to also appear somewhere in
the resolvent expansions, i.e. in  one of the $V_\bq$ factors
 in \eqref{def:prodV}. This means that $h_{q^1_\al, q^2_\al}$
may not generate an additional off-diagonal resolvent element.
We will remove such vertices from the graph 
to allow a more uniform treatment for the rest and
we will account for the twins  separately.

\item [Step 3.]  Two vertices are connected by an edge in
$\cG(\bq)$ if the indices assigned to them are the same,
{\it except} if both vertices are in the first row of the matrix.
I.e., edges connect vertices with identical indices, except that
 there is no edge between any two location indices.

{\it Explanation.} Since the location index plays a different
role than the two non-location indices, their
 possible coincidence 
 have separately been taken into account by the concept of groups.

 \item[Step 4.]   We add an edge between a duplex   $(q^2_{\alpha})_d$
 and its location index $q^1_{\alpha}$ if the multiplicity of 
 the group that the duplex belongs to  is one, i.e. if
the duplex is {\it isolated.}

{\it Explanation.} This is a purely technical convenience. Later
we will consider connected components of $\cG(\bq)$. Isolated duplex
will  be treated separately (see Case 1. below in the
proof of Proposition~\ref{prop:powercount}), but artificially
making the two vertices of a duplex into 
one connected component  will allow us to
simplify the argument of Lemma \ref{lemma:Nleq1}.

\end{description}

We remark that the number of different graphs arising in
via this procedure is bounded by $p^{Cp}$. This is because
$\cG(\bq)$ has the following special structure. Its vertices are
partitioned into equivalence classes
(according to the  common value of their indices)
and any  two vertices within an equivalence class
 are connected by an edge, unless
they are both location vertices. The  number of partitions
of the vertices is at most $p^{Cp}$.  Furthermore, there are
additional edges between duplexes and their location vertices
if the corresponding location index appears only once in $\bq$,
but the possible combinatorics of these additional edges
is at most a factor of $2^p$.

Having defined $\cG(\bq)$, 
 the next step is to   
 assign a {\it weight}  to all vertices as follows.

\begin{definition} [Weight of vertices and groups in $\cG(\bq)$]\label{def:weight}

\begin{itemize}
\item[(i)]  In a  group with multiplicity $m_s =1$
each vertex 
 has  weight  zero.

\item [(ii)] In a   group with multiplicity $m_s > 1$
we assign  a weight $1$ to each duplex in the group;
 all other non-location vertices in the group will have a  weight $1/2$.

\item [(iii)] The total weight of a group is the sum of weights
of its vertices.

\item[(iv)] The total weight  $W=W(\bq)$ of the graph is
the sum of the weights of all vertices.

\end{itemize}
\end{definition}

Clearly, 
the total weight of each group is at most $ m_s \le 2(m_s-1)$. Thus
the total weight of the graph satisfies, by \eqref{msum},
\be
 W \le  \sum_{s=1}^\ell
 2 (m_s-1) = 2(p -\ell). 
\label{totalweight}
\ee 

If all location indices are distinct, then all weights are zero.
In this case, each nonlocation index in $\cG(\bq)$ 
forces a new $h$ term in
$V_\bq$, see \eqref{def:prodV}; note that this statement  used
that twins are taken out of the graph.
 If some location indices coincide,
i.e. we have a group with multiplicity larger than one, then
the possible coincidences of non-location indices
 within the group may  yield non-zero expectation without
forcing a corresponding $h$ factor in $V_\bq$. This may
shorten the expansion  \eqref{def:prodV}, hence
reduce the total number of off-diagonal elements.
The weight measures the maximal reduction of off-diagonal
elements in \eqref{def:prodV} due to the larger multiplicity,
compared with the multiplicity one case.

\begin{definition}[Independent nonlocation indices] \label{131uu}
Denote by $N_{ind}$ the number of different nonlocation indices
that do not coincide with any location index i.e., 
\be\label{Nind}
N_{ind}=
N_{ind}(\bq) :=
 \bigg | \, \{ q^j_\alpha: 2 \le j \le 3, \;  1 \le \alpha \le p \}
 \setminus  \{ q^1_\alpha:  1 \le \alpha \le p \} \, \bigg |,
\ee 
where again $| \cdot|$ denotes the cardinality of the set, disregarding
multiplicity. The elements of this set will be called 
independent nonlocation indices.
\end{definition} 

 Note that $N_{ind}$ gives the actual number of different
$q^2_\al$ and $q^3_\al$ in the second sum in the right hand side of \eqref{fubini}.
Together with the number of groups $\ell$, i.e. the number of
different location indices, the number of terms in the $\sum_\bq$
summation will be bounded by $N^{N_{ind}+\ell}$.

\medskip

We show an example to illustrate this procedure and definitions. Let $p=13$ and 
\be
\bq=\left(\begin{array}{ccccccccccccc}1 & 2 & 3 & 3 & 4 & 4&  5 & 5&  5& 5 & 6 & 7&8
 \\10 & 1 & 2 & 9 & 7 &15 & 9& 9&   9 &9 & 2 & 4&14 
\\10 & 11 & 5 & 6& 7 &12 &9  & 9&13& 13  & 2& 12&14\end{array}\right)
\ee
Then after the first step, we get 
\be
\left(\begin{array}{ccccccccccccc}1 & 2 & 3 & 3 & 4 &4& 5& 5 &5&5 & 6 & 7&8 \\
 (10)_d & 1 & 2 & 9& 7_d&15& 9_d & 9_d&9&9 & 2_d & 4 &14_d\\ 
\ast & 11 & 5 & 6 &\ast &  12 & \ast 
& \ast&13&13 & \ast & 12& \ast \end{array}\right)
\ee
After the second step we have 
\be\label{ex3}
\left(\begin{array}{ccccccccccccc}1 & 2 & 3 & 3 & 4&4& 5 & 5&5&5 & 6 & 7&8 \\
 (10)_d & 1 & 2 & 9& t &15& 9_d  & 9_d&9&9  & 2_d & t &14_d\\ 
\ast & 11 & 5 & 6& \ast &12 & \ast & \ast &13&13& \ast & 12&\ast\end{array}\right)
\ee

In this example, the graph $\cG(\bq)$ will have 31 vertices,
identified with the slots of the matrix in \eqref{ex3} that contain numbers. 
The slots with stars and $t$'s do not count as vertex of $\cG(\bq)$.
The different  location indices are $1, 2, 3,4, 5,6,7,8$
and the different  non-location 
 indices are $9, 10, 11, 12, 13, 14, 15$, thus $\ell=8$, $N_{ind}=7$.
The multiplicity of the groups with different location
indices are $m_1=m_2=m_6=m_7=m_8=1$, $m_3= m_4=2$, $m_5=4$.
For simplicity, in this example, we chose $1, 2, 3,4, 5,6,7,8$
to be the eight different location indices and we used them
to label the groups as well. We also used the consecutive
seven numbers for non-location indices. In general, both the location 
and non-location indices can
be arbitrary numbers between 1 and $N$.

For brevity,
we will often use the index associated to  a vertex to refer to a vertex,
e.g., when we refer to the index $2_d$ in \eqref{ex3}, we really mean
the
vertex $(2, 11)$ since $q_{11}^2= 2_d$. This sometimes
 creates confusion (e.g., there are two vertices  $9$)
 and in that case, we will be specific. 

All vertices  with identical  indices are connected by an edge,
except that there is never an edge between
any two vertices in the first row.
Furthermore, there is an edge  between 
$2_d$ and $6$ (more precisely, between the vertices $(2, 11)$
and $(1,11)$); similarly for $14_d$ and $8$, but
there is no  edge between the  non-location indices $9_d$ 
and their location indices
 $5$ since they belong to a group with multiplicity bigger than one (four)
due to the four location indices 5.
 The vertices with $2,5, 9, 6$ (with common location index $3$)
the vertices with $12, 15$ (with location index $4$) 
and the two $9$'s and $13$'s 
(with common location index $5$) all receive a weight 
$1/2$. The weight of both $9_d$'s is $1$ and all other vertices
have weight zero.
 Notice that the index pair $(5, 9)$ appears twice but they are not twins
(there are no twins inside a group), similarly the two $(5, 9_d)$
are not twin indices.

We will consider connected components of this graph. Due to
the special rule involving duplexes, a connected component may contain
different indices, for example
\be
C= \{ (1,2), (2,3), (2,11),  (3,4), (1,11)  \}
\label{Cex}
\ee
is a connected component in \eqref{ex3},
since $q_2^1= q_3^2=q_{11}^2=2$, $q_4^3=q_{11}^1=6$ and
$q_{11}^1=6$ is connected
to $q_{11}^2=2_d$. With as slight abuse of notation, 
encoding the elements of $C$ only with the indices
$q_\al^i$ instead of the vertices $(i, \al)$ 
we can 
write $C =\{ 2\mbox{(loc.)}, 2, 2_d,  6, 6\mbox{(loc.)}  \} $,
where (loc.) refers to location index.
The list of all  connected components in \eqref{ex3} is
\begin{align}\nonumber
&\{1, 10_d, 1\}, \quad \{11\}, \quad  \{2 (\text{loc.}) ,
 2, 2_d, 6, 6 (\text{loc.})\}, \quad  \{3\}, \quad \{3\}, \quad \{4\},\quad \{4\},\quad \{7\}; \\  
 & \{5, 5 (\text{loc.}), 5 (\text{loc.}), 5 (\text{loc.}), 5 (\text{loc.}) \},\quad 
\{15\},\quad \{12, 12\},\quad  \{9_d, 9_d,9,9,9\},\quad \{13, 13\},\quad \{ 8, 14_d\},
\end{align}
using  the shorter and somewhat ambiguous index-notation.

\subsubsection{Estimates on the integrals}  

\medskip

We now estimate $  \wt \Phi_{\bq,\bnu}^{\bf n}  $ from \eqref{2111}. 
Let $  O=O(\bq, \bn,\bnu)$ be the number of the off-diagonal Green functions
 appearing in the  expansion of the right hand side of \eqref{2111}, i.e., in 
\be
\prod_{\alpha=1}^{p}   V_\bq( \umu^\alpha, \unu^\alpha, n^\al)\xi(\fq_\al)
=  \prod_{\alpha=1}^{p} (-1)^{n_\al}G^{\sa }_{\mu^\alpha_1} h_{\nu^\alpha_1}
 G^{\sa }_{\mu^\alpha_2}h_{\nu^\alpha_2} \ldots 
 h_{\nu^\alpha_{n_\al}} G^{\sa }_{\mu^\alpha_{n_\alpha+1}} \xi(\fq_\al)
\label{expans}
\ee
(see  \eqref{def:V} and \eqref{121}). 
 Define  
\be
 \wt\Lambda_o   = \max_{\alpha = 1, \ldots p; \, i \not = j }  \big | G^{\sa }_{ij} \big |  
\ee  
to be the maximum of the off-diagonal elements of the Green functions $G^{\sa }$. 
 Note that $\wt\Lambda_o$ is independent of the random variables $h_\nu$, $\nu\in Q$.
In particular, the bound $\wt\Lambda_o\le C/(\log N)^2\le 1 $
from \eqref{gsa} holds not only on $\Gamma^c$ but on $(\Gamma^c)_1$ as well. 
Then, with $O=O(\bq, \bn,\bnu)$, and using $n\le 6p$, we have 
\be\label{214}
\big|\wt  \Phi_{\bq}^{\bf n}\big|  = 
\Big| \E {\bf 1} ((\Gamma^c)_1)   \sum_{ \bnu}  
\prod_{\alpha=1}^{p}   V( \umu^\alpha, \unu^\alpha, n^\al)
   \xi (\fq_\al) \Big| 
\le  N^{-n/2-p}(Cp)^{Cp} \sum_{ \bnu}  \E \big[
{\bf 1}((\Gamma^c)_1)  \big ( \wt \Lambda_o \big )^{O}\big],
\ee 
where  for the expectation of the random variables $h_\nu$, $\nu\in Q$,
we have used estimate of the form 
\be\label{hhhhhh}
\E |h_1|^{a_1} \ldots |h_k|^{a_k} \le \big  ( Cm^CN^{-1/2})^m, \quad m:=\sum_{j} a_j 
\ee
for any $a_j$  nonnegative integers, where 
the constant $C$ depends only on $\ttau$.
The total number of $h$ factors appearing in \eqref{expans} is
$n_1+n_2+\ldots +n_p +2p = n+2p$, and \eqref{hhhhhh} shows that
their expectation can be bounded in terms of their total number 
$\sum_j a_j$ irrespective of the precise distribution of the individual
exponents $a_1, a_2, \ldots  a_k$. Thus $N^{-1/2}$ appears to the power
$n+2p$ in \eqref{214}.

 We also recall that the number 
of terms in the summation over $\bnu\in A(\bq,\bn)$
in \eqref{214}
 is bounded by $(4p)^n$, see 
remark below \eqref{213}.

Since we have $\tilde \Lambda_o\le 1$ on the set $(\Gamma^c)_1$, we also  
 have  the trivial estimate
$$
\wt  \Lambda_o^O \le N^{[p - O/2]_+}\big[\wt \Lambda_o^{2} + N^{-1}\big]^p
$$
where $[\;]_+$ denotes the positive part. 
Thus the main term in \eqref{222}
is estimated as  
\begin{multline}\label{sss}
 \frac{ 1 }{N^p}  \sum_{\bq}    \sum_{n=0}^{6p}  \sum_{ |{\bf n} | =n} 
\big| \wt \Phi_{\bq}^{\bf n} \big|  \\
\le  (Cp)^{Cp} \E\big[
{\bf 1}((\Gamma^c)_1)  \wt \Lambda_o^{2} + N^{-1}\big]^p 
\sum_G\sum_{\bq\, : \, \cG(\bq)=G}
 \sum_{n=0}^{6p}  \sum_{ |{\bf n} | =n}  \sum_{\bnu\in A(\bq,\bn)}
  N^{-2p-n/2+ [p-O/2]_+} 
{\bf 1}( \wt \Phi_{\bq,\bnu}^{\bf n}
\ne 0). 
\end{multline}
From \eqref{y71} we have the decomposition 
\be\label{y72}
 {\bf 1} ((\Gamma^c)_1) =  {\bf 1} (\Gamma^c) +
  {\bf 1} \Big( (\Gamma^c)_1 \times Y^c  \setminus \Gamma^c \Big)
  +  {\bf 1} ((\Gamma^c)_1){\bf 1}( Y).
\ee
Since $\wt\Lambda_o\le 1$ on the set $(\Gamma^c)_1$, the
 contributions from the sets $(\Gamma^c)_1 \times Y^c  \setminus \Gamma^c$ 
and  $(\Gamma^c)_1\times  Y$ can be estimated in the same way as in  \eqref{X1}, \eqref{X2}
 by  $C\exp{\big[-c(\log N)^{\psi L } \big]}$. Finally, we can use 
\[
\wt \Lambda_o^O \le 2  \Lambda_o^O
\]
on the set $\Gamma^c \subset (\Gamma^c)_1 \times Y^c$ (see \eqref{gsa})
and thus we can replace $ \E\big[
{\bf 1}((\Gamma^c)_1)  \wt \Lambda_o^{2} + N^{-1}\big]^p $ in \eqref{sss} by $ 2^p \E\big[
{\bf 1}( \Gamma^c)  \Lambda_o^{2} + N^{-1}\big]^p$ with a  negligible error
 $C\exp{\big[-c(\log N)^{\psi L } \big]}$.

By \eqref{elldef} and Definition \ref{131uu}, the total number of different 
summation indices $\bq$ in \eqref{sss} is 
$N_{ind} + \ell$. 
We will prove that
\be\label{236}
 2p+ n +  O \ge  2 N_{ind}+ 2\ell 
\ee
and
\be\label{2361}
  4p+ n \ge  2 N_{ind}+ 2\ell
\ee
hold for any $\bq$, $\bn$ and $\bnu$ for which 
$\wt \Phi_{\bq,\bnu}^{\bf n}\ne 0$.
Since the summations over $G$, $n$, $\bn$ and $\bnu$ give a factor
at most $p^{ C p}$,  
 these two inequalities imply that 
\eqref{sss} is bounded by the right hand side of \eqref{52}. This  proves Lemma \ref{motN}
assuming \eqref{236} and \eqref {2361}. 

\medskip

We  now we prove \eqref{236} and \eqref {2361}.  Recalling  the total weight
 of the graph $W$ satisfies
$W \le 2(p-\ell)$ by \eqref{totalweight},
 the inequality \eqref{236} is a consequence of the following
\begin{proposition}\label{prop:powercount}
 For any $\bq$, $\bn$ and $\bnu$ such that $\wt\Phi_{\bq,\bnu}^{\bf n}\ne 0$, 
we have
\be\label{237}
W (\bq) +  |\bn| +    O (\bq, \bn, \bnu) \ge  2   N_{ind}(\bq).
\ee
\end{proposition}

{\it Proof.} We consider connected components $C$ of the graph $\cG(\bq)$.
If a connected component consists of
 only one location index, we call it {\it trivial}, and we
 will consider only non-trivial components. Nontrivial components
always contain at least one nonlocation vertex since location indices
are never connected directly by an edge. 
We will prove that \eqref{237} holds for each nontrivial connected components
and then we will sum these inequalities. 

To formulate the statement
precisely, we need a few notations. We will fix
 $\bq$, $\bn$ and $\bnu\in A(\bq,\bn)$;
all quantities in the following notations will depend on these parameters.

For each nontrivial connected component $C$ of $\cG(\bq)$, let $I_C$ denote 
the set of all nonlocation indices   appearing in $C$, i.e., 
\be
I_C := \{ q^i_\alpha:  (i, \al) \in C, \; i=2,3 \},
\ee
and for the purpose of $I_C$ we do not distinguish between
indices with or without 
a possible $d$ (duplex) subscript.
Let $L_C$ denote the set 
of all labels associated with  $C$ together 
with their transposes $\nu^t$, where  $ \nu^t = (q, p)$ if $\nu=(p, q)$, i.e., 
\be
L_C := \{ ( q^1_\alpha,  q^i_\alpha): (i, \alpha) \in C\} 
\cup \{ ( q^i_\alpha,  q^1_\alpha): (i, \alpha) \in C\}.
\ee
For example, $ L_C = \{
(2,3 ), (3,2), (6,2), (2,6), (3,6), (6,3) \} $
for  the connected component
$C$ from \eqref{Cex}.
Let 
$$
  n(C)= n(C; \bq, \bn, \bnu):
 = \sum_{\al=1}^p \sum_{m=1}^{n_\al} {\bf 1}( \nu_m^\al \in L_C)
$$
be the total number of $h_\nu$-factors 
with $\nu\in L_C$ appearing in the expansion \eqref{expans}
without the $h$ factors from $\prod_\al \xi(\fq_\al)$.
Finally,  we define $W(C)= W(C; \bq)$ as the total weight of the component $C$,
i.e. the sum of the weights of vertices in $C$.

The following key quantity will be used to count the number of offdiagonal
resolvent matrix elements appearing in the expansion.

\begin{definition}
For $\sigma\in I_C$, let  
$$
  2O(\sigma) := \sum_{\al=1}^p \sum_{m=1}^{n_\al+1}\Big[ {\bf 1}( [\mu_m^\al]_1 
=\sigma,  [\mu_m^\al]_2 \ne\sigma) + {\bf 1}( [\mu_m^\al]_2 
=\sigma,  [\mu_m^\al]_1 \ne\sigma) \Big],
$$
i.e., $2O(\sigma)$ is the number
 of times that $\sigma$   appears as one of the two indices of an
off-diagonal Green function
in the expansion \eqref{expans}.
 Let
\be
O(C) = O(C; \bq,\bn, \bnu): = \sum_{\sigma \in I_C} O(\sigma) 
\ee  
i.e., $2O(C)$ is the number of times that an index associated with $C$ appears
in an off-diagonal Green function in \eqref{expans}.
\end{definition} 

Note that we do not directly count the total 
number $O$ of off-diagonal resolvent matrix elements, 
we rather count how often a fixed non-location index contributes
to an off-diagonal Green function factor. 
In this way we can determine how much each non-location index
contributes to off-diagonal matrix elements and we
can perform our estimates for each component separately.

By definition of the edges in the graph,
 two different nontrivial components $C_1, C_2$ have disjoint 
sets of nonlocation indices; $I_{C_1}\cap I_{C_2}=\emptyset$. 
As a corollary,  the sets $L_C$ for different components are
 also disjoint since the twins are eliminated and for any fixed $\bq, \bn$ and $\bnu$ we have
\be
   \sum_C n(C; \bq,\bn,\bnu) \le |\bn| , \qquad \sum_C O(C; \bq,\bn,\bnu)\le
  O(\bq,\bn,\bnu),
\label{sumC}
\ee
where the summations are over all nontrivial connected components.
Strict inequality can happen as there are indices left out in twins.  
Moreover, we define 
$$
   N_{ind}(C) = N_{ind}(C;\bq):=
\bigg | \, \{ q^j_\alpha: 2 \le j \le 3, \;  1 \le \alpha \le p\; , \; 
(j,\al)\in C \}
 \setminus  \{ q^1_\alpha:  1 \le \alpha \le p \} \, \bigg | 
$$
to be the number of independent  nonlocation indices in the component $C$.
This is the same concept as $N_{ind}(\bq)$ defined in \eqref{Nind} but
restricted to a fixed component $C$.
We clearly have
\be
    \sum_C W(C) = W, \qquad \sum_C N_{ind} (C ; \bq) = N_{ind}(\bq).
\label{sumCC}
\ee
We will prove below that \eqref{237} holds in each nontrivial component $C$, i.e.
for $\wt\Phi_{\bq,\bnu}^{\bf n}\ne 0$, we have
\be\label{2371}
W (C;\bq) +  n(C; \bq,\bn,\bnu)  +   
 O (C; \bq, \bn, \bnu) \ge  2   N_{ind}(C;\bq).
\ee
then \eqref{237} will follow from  \eqref{sumC} and \eqref{sumCC}.

\bigskip

\begin{lemma} \label{lemma:Nleq1}
 Let $C$ be a nontrivial connected component
 of ${\cal G}(\bq)$. Then $N_{ind}(C) \le 1$.
\end{lemma}

\begin{proof}  Suppose that $C$ contains at least two
different independent nonlocation indices $q_\al^j\ne q_\beta^i$
and consider a path in $\cG(\bq)$
connecting their vertices $W_1=(j,\al)$ and $W_2=(i,\beta)$. 
Along this path there must be 
 two subsequent vertices whose indices are different. 
 Considering the
construction of $\cG(\bq)$, this can happen only
along an edge created by  the special rule in Step 4
in the definition of $\cG(\bq)$, i.e. there is a duplex
connected to its location vertex (any other edge connects
identical indices).  
For definiteness, we may choose the
notation $W_1$ and $W_2$ in such a way that along the
path from $W_1$ to $W_2$  the first special edge
created by Step 4 with different indices 
is reached at its non-location vertex
(duplex vertex), call it $U_1$. Clearly $U_1$ and $W_1$
have the same index.
Let now $E$ be the edge connecting    $U_1$ to its location vertex 
 $V_1 \in C$, then by the choice of $U_1$ the index of $V_1$ 
differs from that of $U_1$. 
Let $D$  be the set of all vertices with the same value as $U_1$
and let  $D_1$  be the set of all vertices with the same value as $V_1$,
then $D$ and $D_1$ are disjoint subsets of $C$.

We claim that apart from $V_1$, $D_1$ consists of  nonlocation vertices only.
 Suppose this is not the case. 
Then there is another location vertex $V_1'$ taking the same
 value as $V_1$.  But this implies that  $V_1$ and $V_1'$ 
belong to a group with multiplicity 
at least two. In this case, however, we did not connect the
 duplex $V$ to its location vertex and this leads to contradiction. 

The number of independent nonlocation indices in $D$ is exactly one,
namely the index of $W_1$.
The number of independent nonlocation indices in $D_1$ is zero
since they take  the same value as a location index.

Suppose that $D\cup D_1$ did not exhaust $C$.
In order that $D_1\cup D$ is connected to another vertex with 
a different value, once again, there must be an edge $E'$ connecting a duplex
 vertex to its location vertex; one of these two vertices
must be $D_1\cup D$, the other one must be in the complement.
We claim that the duplex is in $D_1\cup D$.
Indeed, the location vertex cannot be in $D_1\cup D$, 
since $D$ has no location vertex at all (otherwise the index of $U_1$ 
would not be independent) and $D_1$ has only  one location index, $V_1$,
that is already connected within $D\cup D_1$ to its duplex. 

Let  $U_2$ denote the duplex in  $D\cup D_1$ that is connected
to its location index $V_2\not\in D\cup D_1$ and let $D_2$ denote
the set of vertices  with the same value as $V_2$. As before, we
can establish that $D_2$ contains only non-location indices,
apart from $V_2$, and there is no independent nonlocation index in $D_2$.

If $D\cup D_1\cup D_2$ did not exhaust $C$, we continue the process
by defining new sets $D_3$, $D_4$, etc. until $C$ is
exhausted, but we never get a new independent nonlocation index.
This proves that $N_{ind}(C)\le1$.

\end{proof}

We  can start proving  \eqref{2371}. We fix the parameters
$\bq, \bn$ and $\bnu$ and omit them from the notation.
We  will distinguish the following
cases that clearly cover all possibilities.

\newcommand{\sbe}{{[\beta]}}

\begin{description}

\item[ Case 1.]  $C$ consists  of
 a duplex $(q^2_\alpha)_d$ and 
 its location index $q^1_\alpha$. 

Setting $\nu:=(q_\al^1, q_\al^2)$, we know,
in particular, that  $h_\nu$ or $h_{\nu^t}$
do not appear in any other $\xi(\fq_\beta)$, $\beta\ne \al$
since $C$ is an isolated component, not connected 
to any other vertices.
Then,  by the observation made in \eqref{obs},
  $h_\nu$ (or  $h_ {\nu^t}$)
must explicitly appear in \eqref{def:prodV} and it
clearly  must appear in one of the following ways, with some $\beta\ne \al$,
\begin{align} 
(1):&  \qquad   G_{f_1, q^1_\alpha}^\sbe   h_ {q^1_\alpha  q^2_\alpha}
 G_{q^2_\alpha, f_2}^\sbe,
 \quad \mbox{or} \quad  G_{f_2, q^2_\alpha}^\sbe   h_ {q^2_\alpha  q^1_\alpha}
 G_{q^2_\alpha, f_2}^\sbe, \quad f_i \not = q^i_\alpha \\
(2):&     \qquad  h_ {q^1_\alpha  q^2_\alpha} G_{q^2_\alpha, q^2_\alpha}^\sbe \
 h_ {q^2_\alpha  q^1_\alpha}, \quad \mbox{or} \quad
 h_ {q^2_\alpha  q^1_\alpha} G_{q^1_\alpha, q^1_\alpha}^\sbe  h_ {q^1_\alpha  q^2_\alpha}. 
\end{align}
The main reason why only one of these  possibilities
 occurs is because the indices $q^i_\al, i = 1, 2$, appear only in $C$. 
So  either (1) both Green functions neighboring $ h_\nu$ (or $h_{\nu^t}$)
are off-diagonal, or (2)  either of the neighboring Green function is diagonal.
In the latter case, however, the expansion must continue on the
other side of this diagonal Green function with another factor
 $ h_ {q^1_\alpha  q^2_\alpha}$ (or $h_ {q^2_\alpha  q^1_\alpha}$).
The reason for this last statement is
 that the expansion cannot start or terminate with 
a diagonal Green function of the form $ G_{q^1_\alpha, q^1_\alpha}^\sbe$
or $ G_{q^2_\alpha, q^2_\alpha}^\sbe$ since that would entail that $q_\al^1$ (or
$q_\al^2$) equals to $q^2_\beta$ or $q^3_\beta$, which would mean
that $C$ contained other elements as well.

In the first  case, $n(C) \ge 1$ and we have identified two indices of 
off-diagonal Green functions associated with $q^2_\alpha$, i.e. $O(C)\ge 1$.
In the second case, 
we find that $h_\nu$ or $h_{\nu^t}$ appear altogether
 twice and hence  $n(C) \ge 2$. Since $N_{ind}(C)=1$ in this case,
we have thus proved that in both cases 
\be\label{234}
W(C)+  n(C) +  O(C) \ge   2= 2  N_{ind}(C) .
\ee
Notice that we did not use weight $W(C)$ here.

\item[ Case 2.] $C$ is  an isolated non-duplex vertex.

Since $C$ is nontrivial, we can assume that
$C$ consists of a single vertex $(2, \al)$ 
(the case  of $(3, \al)$ is identical). Let $\nu:=(q_\al^1, q_\al^2)$.
Consider the expansion 
of $G^\sa$, see \eqref{def:V}. The first and the last
Green functions in this expansion will be called {\it extreme}
Green functions; if $n_\al=0$, then the single Green function
$G^{(q_\al^1)}_{q_\al^2 q_\al^3}$ will be called extreme. 
Since this expansion contains  $h_\mu$ factors only with
$\mu\in Q^{(\al)}$ and $q_\al^2\ne q_\beta^i$ for any $\beta\ne\al$
(since $C$ is an isolated vertex),
 thus $q_\al^2$
cannot appear as an index of any $h_\mu$. Then the first Green
function in  \eqref{def:V} must be of the form  $G_{q^2_\al, f}^\sa$ with some
 $f \not = q^2_\al$, i.e. it
must be off-diagonal, thus $O(C)\ge \frac{1}{2}$.
 Furthermore, $h_\nu$ or $h_{\nu^t}$ must appear as 
\begin{align} 
(1):&     \qquad h_ {q^1_\alpha  q^2_\alpha} G^\sbe_{q^2_\alpha, f} \quad \mbox{or}
\quad G^\sbe_{f,  q^2_\alpha}  h_ {q^2_\alpha  q^1_\alpha},
\quad  f \not = q^2_\alpha \\
(2):&     \qquad  h_ {q^1_\alpha  q^2_\alpha} G^\sbe_{q^2_\alpha, q^2_\alpha}
  h_ {q^2_\alpha  q^1_\alpha} \quad \mbox{or} \quad
   h_ {q^2_\alpha  q^1_\alpha} G^\sbe_{q^1_\alpha, q^1_\alpha}
  h_ {q^1_\alpha  q^2_\alpha}
\end{align}
In the first  case (1), we have identified another  index
  of off-diagonal Green function associated with $q^2_\alpha$,
so $O(C)\ge 1$ and $n(C)\ge 1$.
In the second case (2), 
we find that $h_\nu$ and $h_{\nu^t}$ appears altogether twice 
 and thus $n(C) \ge 2 $. In both cases
 we have proved \eqref{2371} since $N_{ind}(C) = 1$. Again, the weight 
$W(C)$ was not used.

\item[Case 3.] $C$ has only one non-location vertex, $(i, \al)$, $i=2,3$ 
and at least one location vertex $(1,\beta)$ with $\beta\ne\al$.  

In this case the non-location index $q_\al^i$ is equal to a location index,
hence $N_{ind}(C)=0$ and \eqref{2371} is obvious.

\item[ Case 4.] $C$  has more than one  non-location vertex.

Suppose the weight
 of a non-location vertex  $ (2,\alpha)$ in $C$ is zero. Then 
$h_ {q^1_\alpha  q^2_\alpha}$ (or $h_ {q^2_\alpha  q^1_\alpha}$)
  must appear in \eqref{expans} (apart from the $\xi$ factors)
  and thus it contributes to $n(C)$
 by one. Here we are using the following reason: 
\begin{description}
\item[$(\dagger)$]  If  $h_ {q^1_\alpha  q^2_\alpha}$   and $h_ {q^2_\alpha  q^1_\alpha}$
  appear in  $\prod_\beta \xi(\fq_\beta)$ at least twice, 
then either  $(2,\al)$ is a twin vertex or the multiplicity of the group
 containing $(2,\al)$  is more than one. 
\end{description}

Both cases contradict  our definitions; twins are
not part of $\cG(\bq)$, and non-location vertices in 
groups with higher multiplicity have nonzero weight.
 But if  $h_ {q^1_\alpha  q^2_\alpha}$ and $h_ {q^2_\alpha  q^1_\alpha}$
 appear only once in  $\prod_\beta \xi(\fq_\beta)$  (namely, only in
the factor $\xi(\fq_\al)$), then at least one of them
 need to appear at least one more times in \eqref{expans} 
 to make the expectation nonzero. 

Hence if we have at least 
two weight zero non-location vertices in $C$, then $n(C)\ge2$  and \eqref{2371} holds.
Note that each of these two vertices contribute to  $n(C)$ by one,
since together with their own location vertex they
must form two different labels,
otherwise they would be part of a twin or a group with multiplicity 
at least 1 and their weight  would not be zero. We can also
 assume that the total weight $W(C)$ is less than $2$ or, 
if there is a weight zero non-location vertex, hence $n(C)\ge 1$, then
 the total weight is at most $W(C)\le 1/2$. In all other cases
\eqref{2371} follows trivially from $N_{ind}(C)\le1$.

So we only have to consider the following remaining cases:

\begin{enumerate} 
\item  {\it  The non-location vertices of $C$  consist of exactly two weight $1/2$ vertices
$v_1, v_2$.}

First notice that these two vertices  must have the 
same index.
Otherwise they could be in the same connected component only
if  one of them, say $v_2$, would be
equal to a duplex $(q_\beta^2)_d$ with some $\beta\ne\al$
where $v_1=q_\al^j$ ($j\in \{2,3\}$),
and this duplex would belong to a group with
multiplicity one (a connecting edge between vertices
with different indices can be provided only via
a special edge from Step 4. between a duplex and its location vertex
and only if the corresponding group has multiplicity one).
But in this case the weight of the non-location vertex $(2,\beta)$ in $C$
would be zero by (i) of Definition~\ref{def:weight}.

Thus the two vertices $v_1, v_2$ cannot be in the same
column of the matrix (otherwise they formed a duplex),
so without loss of
generality we can assume that they are of the form
$(2,\al)$ and $(2,\beta)$ with $\al\ne \beta$
and  we know that $ q^2_\alpha =   q^2_\beta$.

Consider first the case  $ q^1_\alpha\ne  q^1_\beta$.
 By the fact that the common value  $ q^2_\alpha =   q^2_\beta$
appears only twice in $C$, both
factors
  $h_ {q^1_\alpha  q^2_\alpha}$ and $h_ {q^1_\beta  q^2_\beta}$ (or their transposes) 
have to appear in \eqref{expans}. 
 Thus  $n(C) \ge 2$ and  \eqref{2371} holds.

Finally, consider the case   $ q^1_\alpha=  q^1_\beta$.  Since  
$ q^2_\alpha$ and $   q^2_\beta$ have 
weight $1/2$, they are not duplex. 
By construction, we have to expand the Green function $G^{(q^1_\al)}_{q^2_\al, q^3_\al}$.
 Since  $q^2_\al \not = q^3_\al$, 
in the expansion \eqref{expans}, the first Green function  $G^{\sa }_{\mu^\alpha_1}$
 is off-diagonal (otherwise the beginning of the
 expansion were $G^\sa_{q^2_\al, q^2_\al} h_{q^2_\al q^1_\al}\ldots$, but $h_{q^2_\al q^1_\al}$
cannot appear in the expansion of $G^{(q^1_\al)}_{q^2_\al, q^3_\al}$).
Hence  $q^2_\al$  appears  as an index of an extreme  off-diagonal Green function.  
Similar statement holds for $q^2_\beta$. Hence we have identified two indices of 
 off-diagonal Green functions  
associated with $C$  so that  $O(C) \ge 1 $  and  together with $W(C)\ge 1$
we obtain that \eqref{2371} holds.

\item {\it The non-location vertices of $C$ consist  of exactly one weight 1/2  vertex,
and one weight $1$  vertex.}

Since the weight $1$ vertex is a duplex, these two vertices
cannot be in the same column of $\bq$. Without loss of generality, let $(2,\al)$ be the
weight 1/2 vertex and let $(2,\beta)_d$ be the
 weight $1$ vertex, $\al\ne\beta$.
We can consider two cases:  $ q^1_\alpha \not =   q^1_\beta$ and
 $ q^1_\alpha  =   q^1_\beta$. As before, for the first case, $n(C) \ge 1$. For the 
second case,  $q^2_\al$ cannot appear as an index of any  $h_\nu$
in any other $\xi(\fq_\gamma)$ for $\gamma\ne \al,\beta$ since $C$ consist of exactly
two columns, namely the columns $\alpha$ and $\beta$.
Thus $h_{q_\al^1,q_\al^2}$ or its transpose must appear in the expansion
of  $G^{[\al]}$ and therefore  we can find   $q^2_\al$ as one of the indices
of an extreme off-diagonal Green function. Hence
we have $O(C) \ge 1/2$ in the second case. Since $W(C)\ge 3/2$, we obtain in both cases that
 \eqref{2371} holds.

\item  {\it The non-location vertices of $C$ consist  of exactly one weight $1/2$ 
 vertex  one  weight zero vertex.}

Since the two vertices have different weights, they are in different
columns of the matrix. Without loss of generality, we can
assume that the weight 1/2 vertex is $(2,\al)$ and the weight zero vertex
is $(2,\beta)$ with $\al\ne \beta$.
 In this case, both $h_ {q^1_\alpha  q^2_\alpha}$ and $h_ {q^1_\beta  q^2_\beta}$
(or their transposes) 
 have to appear in the expansion, thus $n(C)\ge 2$  and \eqref{2371} holds.

\item {\it The non-location vertices of $C$ consist  of exactly three weight $1/2$ vertices.}

Similar arguments as in the first case, we can show that these three
vertices are in different columns and we can thus assume that they
 are of the form $(2,\al)$, $(2, \beta)$ and $(2,\gamma)$
with different $\al, \beta,\gamma$. If $q_\al^1=q_\beta^1=q_\gamma^1$, then
 $q^2_\al$  appears  as an index of an extreme  off-diagonal Green function
in the expansion of $G^{(q^1_\al)}_{q^2_\al, q^3_\al}$ and $O(C)\ge 1/2$. On the other hand,
if one of the three location indices, say $q_\al^1$, differed from the other
two, then $h_ {q^1_\alpha  q^2_\alpha}$  (or its transpose) have to appear in 
the expansion and $n(C)\ge 1$. In either case, together with $W(C)\ge 3/2$,
we obtain  \eqref{2371}.

\end{enumerate} 

 \end{description} 

The main reason of the previous proof is that any weight $1/2$
 vertex either associated with an index of an extreme 
 off-diagonal Green function or there is an $h$ factor associated with it. 
We have thus proved Proposition \ref{prop:powercount} \qed
\bigskip

Finally, we have to prove the inequality \eqref{2361}. 
Let $d$ denote the number of duplexes.
 Let
$a_1$ be the number of nontrivial components $C$ 
that contain only one non-location vertex
and let $a_2$ be the number of nontrivial components $C$ that contain
at least two nonlocation vertices. 
Since by Lemma \ref{lemma:Nleq1} we have 
 $N_{ind} = \sum_{C} N_{ind}(C) \le a_1+a_2$ 
and obviously $\ell \le  p$, it is sufficient
to show that 
$$ 
   2p + n \ge  2(a_1+a_2) .
$$
Since we there are  $2p-d$ nonlocation vertices, we have
$2p -d \ge a_1 + 2a_2$, thus it is sufficient to show that
$n+d\ge a_1$. But each component with a single non-location vertex,
say $(2, \al)$, is either a duplex
or it gives rise to a factor $h_{q_\al^1 q_\al^2}$ (or its
transpose) that must appear in the expansion, hence it
contributes to $n$. This shows  \eqref{2361}
and this completes the proof of Lemma \ref{motN}. \qed

\end{document}